\documentclass[sigconf, nonacm]{acmart}
 
 \newcommand\vldbdoi{10.14778/3510397.3510401}
 \newcommand\vldbpagestyle{plain} 
 \newcommand\vldbpages{XXX-XXX}
 \newcommand\vldbvolume{15}
 \newcommand\vldbissue{5}
 \newcommand\vldbyear{2022}
 \newcommand\vldbauthors{\authors}
 \newcommand\vldbtitle{\shorttitle} 
 \newcommand\vldbavailabilityurl{https://github.com/shaleen/rankedenumprojections}

\pdfoutput=1
\usepackage{balance}
\usepackage{tcolorbox}
\usepackage{graphicx}   

\usepackage{xcolor,listings}
\usepackage{textcomp}
\usepackage{color}
\usepackage{booktabs}

\definecolor{codegreen}{rgb}{0,0.6,0}
\definecolor{codegray}{rgb}{0.5,0.5,0.5}
\definecolor{codepurple}{HTML}{C42043}
\definecolor{backcolour}{HTML}{F2F2F2}
\definecolor{bookColor}{cmyk}{0,0,0,0.90}  
\color{bookColor}

\usepackage{hyperref}   
\usepackage{thm-restate}
\usepackage{mdwlist}
\usepackage{mathtools}
\usepackage{xspace}
\usepackage{verbatim}   
\usepackage{xcolor}
\usepackage{makecell}
\usepackage[mathscr]{euscript}
\usepackage{tikz}
\usepackage{relsize}
\usepackage{booktabs}
\usepackage{tikz-qtree}
\usepackage{subcaption}
\usepackage{multirow}
\usetikzlibrary{arrows.meta}
\usetikzlibrary{fit,shapes,trees,shapes.geometric}
\usetikzlibrary{patterns,decorations.pathreplacing,calc}
\usetikzlibrary{matrix, positioning, arrows}
\usetikzlibrary{chains,shapes.multipart}
\usetikzlibrary{shapes,calc}
\usetikzlibrary{automata}
\usepackage[linesnumbered,algoruled, lined, noend]{algorithm2e}

\newsavebox{\bullmagenta}
\newsavebox{\bullolive}
\newsavebox{\bullblue}
\newsavebox{\bullteal}
\newsavebox{\bullbrown}
\newsavebox{\bullorange}
\sbox\bullmagenta{\tikz{\draw[magenta,fill=magenta] circle (.3ex);}}
\sbox\bullolive{\tikz{\draw[olive,fill=olive] circle (.3ex);}}
\sbox\bullblue{\tikz{\draw[blue,fill=blue] circle (.3ex);}}
\sbox\bullteal{\tikz{\draw[teal,fill=teal] circle (.3ex);}}
\sbox\bullbrown{\tikz{\draw[brown,fill=brown] circle (.3ex);}}
\sbox\bullorange{\tikz{\draw[orange,fill=orange] circle (.3ex);}}

\usepackage[normalem]{ulem}

\definecolor{purple}{rgb}{0.65,0.05,0.3}
\definecolor{red}{rgb}{1,0,0}
\definecolor{RED}{rgb}{1,0,0}

\newcommand{\out}[0]{\mathsf{output}}

\newcommand{\cut}[1]{}   
\makeatletter
\newcommand{\mytag}[2]{%
	\text{#1}%
	\@bsphack
	\protected@write\@auxout{}%
	{\string\newlabel{#2}{{#1}{\thepage}}}%
	\@esphack
}
\makeatother

\newenvironment{packed_item}{
	\begin{itemize}
		\setlength{\itemsep}{1pt}
		\setlength{\parskip}{0pt}
		\setlength{\parsep}{0pt}
	}
	{\end{itemize}}

\newcommand{\anst}[1]{\mathsf{output(#1)}}

\newcommand{\SUM}{\sqlhighlight{SUM}}
\newcommand{\lexi}{\sqlhighlight{LEXICOGRAPHIC}}

\newcommand{\introparagraph}[1]{\medskip \noindent {\bf  #1.}}  

\newcommand{\setof}[2]{\{{#1}\mid{#2}\}}        

\usepackage{aliascnt} 

\newtheorem{definition}{Definition}

\newtheorem{lemma}{Lemma}
\newtheorem{example}{Example}

\newtheorem{theorem}{Theorem}

\providecommand{\bA}[0]{\mathbf{A}}

\providecommand{\bu}[0]{\mathbf{u}}

\providecommand{\mA}[0]{\mathcal{A}}

\providecommand{\pq}[0]{\textsf{PQ}}

\providecommand{\mL}[0]{\mathcal{L}}

\providecommand{\prnt}[1]{\texttt{parent}({#1})}
\providecommand{\child}[1]{\texttt{child}({#1})}
\providecommand{\key}[1]{\texttt{anchor}({#1})}

\providecommand{\rf}[0]{\texttt{rank}}

\providecommand{\fhw}[0]{\texttt{fhw}}
\providecommand{\subw}[0]{\texttt{subw}}
\providecommand{\panda}[0]{\texttt{PANDA}}

\providecommand{\edges}[0]{\mathcal{E}}

\providecommand{\domain}[0]{\mathbf{dom}}

\providecommand{\htree}[0]{\mathscr{T}}
\providecommand{\bag}[0]{\mathscr{B}}
\providecommand{\fhw}[1]{\mathsf{fhw}(#1)}

\providecommand{\mO}[0]{\mathcal{O}}
\providecommand{\mT}[0]{\mathcal{T}}

\providecommand{\twopath}[0]{Q_{\mathsf{2path}}}
\providecommand{\threepath}[0]{Q_{\mathsf{3path}}}
\providecommand{\fourpath}[0]{Q_{\mathsf{4path}}}
\providecommand{\threestar}[0]{Q_{\mathsf{3star}}}

\providecommand{\rank}[0]{\mathsf{rank}}

\providecommand{\eat}[1]{}
\providecommand{\sqlhighlight}[1]{{\color{codepurple} \ttfamily #1}}

\AtBeginDocument{%
}

\AtBeginDocument{
\hypersetup{colorlinks,
	 linkcolor=ACMPurple,
	 citecolor=ACMPurple,
	 urlcolor=ACMDarkBlue,
	 filecolor=ACMDarkBlue}}
\begin{document}
	
	\title{Ranked Enumeration of Join Queries with Projections}
	\author{Shaleen Deep}
	\affiliation{	
	\institution{University of Wisconsin - Madison}}
	\email{shaleen@cs.wisc.edu}
		
	\author{Xiao Hu}
	\affiliation{	
	\institution{Duke University}}
	\email{xh102@cs.duke.edu}
	
	\author{Paraschos Koutris}
		\affiliation{	\institution{University of Wisconsin - Madison}}
	\email{paris@cs.wisc.edu}

	\begin{abstract}
Join query evaluation with ordering is a fundamental data processing task in relational database management systems. \textsf{SQL} and custom graph query languages such as \textsf{Cypher} offer this functionality by allowing users to specify the order via the  \sqlhighlight{ORDER BY} clause. In many scenarios, the users also want to see the first $k$ results quickly (expressed by the \sqlhighlight{LIMIT} clause), but the value of $k$ is not predetermined {as user queries are arriving in an online fashion}. Recent work has made considerable progress in identifying optimal algorithms for ranked enumeration of join queries that do \emph{not} contain any projections. In this paper, we initiate the study of the problem of enumerating results in ranked order for queries {\em with projections}. Our main result shows that for any acyclic query, it is possible to obtain a near-linear (in the size of the database) delay algorithm after only a linear time preprocessing step for two important ranking functions: sum and lexicographic ordering. For a practical subset of acyclic queries known as star queries, we show an even stronger result that allows a user to obtain a smooth tradeoff between faster answering time guarantees using more preprocessing time. Our results are also extensible to queries containing cycles and unions. We also perform a comprehensive experimental evaluation to demonstrate that our algorithms, which are simple to implement, improve up to three orders of magnitude in the running time over state-of-the-art algorithms implemented within open-source RDBMS and specialized graph databases.
\end{abstract}

	\maketitle
	
	\pagestyle{\vldbpagestyle}
	\begingroup\small\noindent\raggedright\textbf{PVLDB Reference Format:}\\
	\vldbauthors. \vldbtitle. PVLDB, \vldbvolume(\vldbissue): \vldbpages, \vldbyear.\\
	\href{https://doi.org/\vldbdoi}{doi:\vldbdoi}
	\endgroup
	\begingroup
	\renewcommand\thefootnote{}\footnote{\noindent
		This work is licensed under the Creative Commons BY-NC-ND 4.0 International License. Visit \url{https://creativecommons.org/licenses/by-nc-nd/4.0/} to view a copy of this license. For any use beyond those covered by this license, obtain permission by emailing \href{mailto:info@vldb.org}{info@vldb.org}. Copyright is held by the owner/author(s). Publication rights licensed to the VLDB Endowment. \\
		\raggedright Proceedings of the VLDB Endowment, Vol. \vldbvolume, No. \vldbissue\ %
		ISSN 2150-8097. \\
		\href{https://doi.org/\vldbdoi}{doi:\vldbdoi} \\
	}\addtocounter{footnote}{-1}\endgroup
	
	\ifdefempty{\vldbavailabilityurl}{}{
		\vspace{.3cm}
		\begingroup\small\noindent\raggedright\textbf{PVLDB Artifact Availability:}\\
		The source code, data, and/or other artifacts have been made available at \url{\vldbavailabilityurl}.
		\endgroup
	}


\section{Introduction} \label{intro}

Join processing is one of the most fundamental problems in database research with applications in many areas such as anomaly and community detection in social media, fraud detection in finance, and health monitoring. In many data analytics tasks, it is also required to rank the query results in a specific order. This functionality is supported by the \sqlhighlight{ORDER BY} clause in \textsf{SQL, Cypher} and \textsf{SPARQL}. We demonstrate a practical example use-case.

\begin{example} \label{ex:1}
Consider the DBLP dataset as a single relation $R(A,B)$, indicating that $A$ is an author of paper $B$. Given an author $a$, the function \textsf{h-index}$(a)$ returns the h-index of $a$. A popular analytical task asks to find all co-authors who authored at least one paper together. Additionally, the pairs of authors should be returned in decreasing order of the sum of their h-indexes, since users are only interested in the top-100 results. The following SQL query captures this task.
	\lstset{upquote=true}
	
	\lstdefinestyle{mystyle}{
		commentstyle=\color{codegreen},
		keywordstyle=\color{codepurple},
		stringstyle=\color{codepurple},
		basicstyle=\small\ttfamily,
		breaklines=true,
		columns=fullflexible,
		frame=single,
	}
	\lstset{style=mystyle}
	
	\newcommand\numberstyle[1]{%
		
	}
	\begin{lstlisting}[ language=SQL,
	deletekeywords={IDENTITY},
	deletekeywords={[2]INT},
	morekeywords={clustered},
	mathescape=true,
	xleftmargin=0pt,
	framexleftmargin=0pt,
	frame=tb,
	framerule=0pt ]
	SELECT DISTINCT $R_1.A, R_2.A$ FROM $R$ AS $R_1$, $R$ AS $R_2$
	WHERE $R_1.B = R_2.B$ 
	ORDER BY h_index($R_1.A$) + h_index($R_2.A$) LIMIT 100;
	\end{lstlisting}
\end{example}

The above task is an example of a join query with projections (join-project queries) because attribute $B$ has been projected out (i.e. it is not present in the selection clause). The \sqlhighlight{DISTINCT} clause ensures that there are no duplicate results. 

\smallskip
{\introparagraph{Importance of joins with projections} Join queries containing projections appear in several practical applications such as recommendation systems~\cite{feng2019attention, li2019supervised}, similarity search~\cite{yu2012user}, and network reachability analysis~\cite{elmacioglu2005six, biryukov2008co}. In fact, as Manegold et al.~\cite{manegold2009database} remarked, joins in real-life queries almost always come with projections over certain attributes. Matrix multiplication~\cite{amossen2009faster}, path queries (equivalent to sparse matrix multiplication), and reachability queries~\cite{green2013datalog} are all examples of join-project queries that have widespread applications in linear and relational algebra. Other data models such as \textsf{SPARQL}~\cite{perez2009semantics} also support the projection operator and evaluation of join-project queries has been a subject of research, both theoretically~\cite{angles2011subqueries} and practically~\cite{corby2007implementation}. In fact, as \textsf{SPARQL} supports \sqlhighlight{ORDER BY}/\sqlhighlight{LIMIT} operator, ranked enumeration for  queries (that include projections) and top-k over knowledge bases in the \textsf{SPARQL} model has also been explicitly studied recently~\cite{leeka2016quark, christmann2021efficient}. As many practical \textsf{SPARQL} evaluation systems~\cite{harris2005sparql, sequeda2013ultrawrap} evaluate queries using RDBMS, it is important to develop efficient algorithms for such queries in the relational model. Similarly,~\cite{graphgen2017} argued that since a large fraction of the data of interest resides in RDBMS, efficient execution of graph queries (such as path and reachability queries that contain projections and ranking) using RDBMS as the backend is valuable. In the relational setting, join-project queries also appear in the context of probabilistic databases (see Section 2.3 in~\cite{dalvi2007efficient}). This motivates us to develop efficient algorithms, both in theory and practice, that address the challenge of incorporating the ranked enumeration paradigm for join-project queries.}

\smallskip
\introparagraph{Prior Work} Efficient evaluation of join queries in the presence of ranking functions has been a subject of intense research in the database community. {Recent work~\cite{tziavelis2020optimal, deep2019ranked, yang2018any, chang2015optimal, tziavelis2021beyond} has made significant progress in identifying optimal algorithms for enumerating query results in ranked order. In each of these works, the key idea is to perform on-the-fly sorting of the output via the use of priority queues by taking into account the query structure.~\cite{chang2015optimal} considered the problem of top-k tree matching in graphs and proposed optimal algorithms by combining Lawler's procedure~\cite{lawler1972procedure} with the ranking function.~\cite{tziavelis2020optimal} introduced multiple dynamic programming algorithms that lazily populate the priority queues.~\cite{yang2018any} took a different approach  where all possible candidates were eagerly inserted into the priority queues and~\cite{deep2019ranked} generalized these ideas to present a unified theory of ranked enumeration for full join queries. Very recently,~\cite{tziavelis2021beyond} was able to extend some of these results to non equi-joins as well.} The performance metric for enumerating query results is the \emph{delay}~\cite{bagan2007acyclic}, defined as the time difference between any two consecutive answers. Prior work was able to obtain logarithmic delay guarantees, which were shown to be optimal. {However, all prior work in this space suffer from one fundamental limitation: it assumes that the join query is full, i.e. there are no non-trivial projections involved. In fact,~\cite{tziavelis2021beyond} explicitly remarks that in presence of projections, the strong guarantees obtained for full queries do not hold anymore. Their suggestion to handle this limitation is to convert the query with projections into a projection-free result, i.e. materialize the join query result, apply the projection filter, and then rank the resulting output. However, this conversion requires an expensive materialization step. An alternate approach is to modify the weights of the tuples/attribute values to allow re-use of existing algorithms (we describe this approach in~\autoref{sec:framework}). As we show later, this approach also does not fare any better and requires enumerating the full output of the join query, which can be polynomially slower than the optimal solution.}

On the practical side, all RDBMS and graph processing engines evaluate join-project queries in the presence of ranking functions by performing three operations in serial order: $(i)$ materializing the result of the full join query, $(ii)$ de-duplicating the query result (since the query has \sqlhighlight{DISTINCT} clause), and $(iii)$ sorting the de-duplicated result according to the ranking function. The first step in this process is a show-stopper. Indeed, the size of the full join query result can be orders of magnitude larger than the size of the final output after applying projections and de-duplicating it. Thus, the materialization and the de-duplication step introduces significant overhead since they are blocking operators. Further, if the user is interested in only a small fraction of the ordered output, the user still has to wait until the entire query completes even to see the top-ranked result.

 \subsection{Our Contribution and Key Ideas}   

In this paper, we initiate the study of ranked enumeration over join-project queries. We focus on two important ranking functions: \SUM\ ($f (x,z) = x+z$) and \lexi\ ($f (x,z) = x,z$) for two reasons. First, both of these functions are very useful in practice~\cite{ilyas2008survey}. Second, extending the algorithmic ideas to other functions, such as \sqlhighlight{MIN}, \sqlhighlight{MAX}, \sqlhighlight{AVG} and circuits that use sum and products, is quite straightforward. 
More specifically, we make three contributions. 

\smallskip
\introparagraph{1. Enumeration with Formal Delay Guarantees}
Our first main result shows that for any {\em acyclic} query (the most common fragment of queries in practice~\cite{bonifati2020analytical}) with arbitrary projection attributes, it is possible to develop efficient enumeration algorithms (\autoref{sec:acyclic}).

\begin{theorem} \label{thm:general}
	For an acyclic join-project query $Q$, an instance $D$, and a ranking function $\rank \in \{\text{\SUM}, \text{\lexi}\}$, the query result $Q(D)$ can be enumerated according to $\rank$ with worst-case delay $O(|D| \log |D|)$, after $O(|D|)$ preprocessing time.
\end{theorem} 

This result implies that top-$k$ results in $Q(D)$ can be enumerated in $O(k |D| \log |D|)$ time.~\autoref{thm:general} is able to recover the prior results for ranked enumeration of full queries as well~\cite{deep2019ranked}. The key idea of our algorithm is to develop multiway join plans~\cite{ngo2012worst} by exploiting the properties of join trees. Embedding the priority queues in the join tree strategically allows us to generate the sorted output on-the-fly and avoid the binary join plans that all state-of-the-art systems use. Further, since we formulate the problem in terms of delay guarantees, it allows our techniques to be limit-aware: for small $k$, the answering time is also small.

\smallskip
\introparagraph{2. Faster Enumeration with More Preprocessing}
Our second contribution is an algorithm that allows for a smooth tradeoff between preprocessing time and delay guarantee for a subset of join-project queries known as {\em star} queries over binary relations of the form $R_i(A_i, B)$ (denoted as $Q^\star_m)$ (\autoref{sec:star}):

\lstset{upquote=true}

\lstdefinestyle{mystyle}{
	commentstyle=\color{codegreen},
	keywordstyle=\color{codepurple},
	stringstyle=\color{codepurple},
	basicstyle=\small\ttfamily,
	breaklines=true,
	columns=fullflexible,
	frame=single,
}
\lstset{style=mystyle}

\begin{lstlisting}[ language=SQL,
deletekeywords={IDENTITY},
deletekeywords={[2]INT},
morekeywords={clustered},
framesep=2pt,
mathescape=true,
xleftmargin=2pt,
framexleftmargin=1pt,
frame=tb,
framerule=0pt ]
SELECT DISTINCT $A_1, \dots, A_m$ FROM $R_1, \dots, R_m$
WHERE $R_1.B = \dots = R_m.B$ ORDER BY $A_1 + \dots + A_m$ LIMIT k;
\end{lstlisting}


\begin{theorem} \label{thm:ranked:star}
	For a star join-project query $Q^\star_m$, an  instance $D$, and a ranking function $\rank \in \{\text{\SUM}, \text{\lexi}\}$, the query result $Q(D)$ can be enumerated according to $\rank$ with worst-case delay $O\left(|D|^{1 - \epsilon} \log |D|\right)$, using $O\left(|D|^{1 + (m - 1) \epsilon}\right)$ preprocessing time and $O\left(|D|^{m(1 - \epsilon)}\right)$ space, for any $0 \leq \epsilon \leq 1$.
\end{theorem}

\autoref{thm:ranked:star} enables users to carefully control the space usage, preprocessing time and delay. For both~\autoref{thm:general} and~\autoref{thm:ranked:star}, we can show that the delay guarantee is optimal subject to a conjecture about the running time of star join-project queries in \autoref{sec:optimality}.

\smallskip
\introparagraph{3. Experimental Evaluation}
Our final contribution is an extensive experimental evaluation for practical join-project queries on real-world datasets (\autoref{sec:experiments}). To the best of our knowledge, this is the first comprehensive evaluation of how existing state-of-the-art relational and graph engines execute join-project queries in the presence of ranking. We choose MariaDB, PostgreSQL, two popular open-source RDBMS, and Neo4j as our baselines. We highlight two key results. First, our experimental evaluation demonstrates the bottleneck of serially performing materialization, de-duplicating, and sorting. Even with \sqlhighlight{LIMIT 10} (i.e. return the top-$10$ ranked results), the engines are orders of magnitude slower than our algorithm. For some queries, they cannot finish the execution in a reasonable time since they run out of main memory. On the other hand, our algorithm has orders of magnitude smaller memory footprint that allows for faster execution. 
The second key result is that all baseline engines are agnostic of the ranking function. The execution time of the queries is identical for both the sum and lexicographic ranking function. However, our algorithm uses the additional structure of lexicographical ordering and can execute queries $2-3\times$ faster than the sum function. For queries with unions and cycles, our algorithm maintains its performance improvement over the baselines. 

	\section{Problem Setting} \label{sec:framework}

In this section we present the basic notions and terminology, and then discuss our framework. We focus on the class of {\em join-project queries}, which are defined as
$$ \label{eq:q}
Q = \pi_{\bA} (R_1(\bA_1) \Join R_2(\bA_2) \Join \ldots \Join R_m(\bA_m)) 
$$

Here, each relation has schema $R_i(\bA_i)$, where $\bA_i$ is an ordered set of attributes. Let $\mathbb{A} = \bA_1 \cup \bA_2 \cup \cdots \cup \bA_m$. The projection operator $\pi_{\bA}$ only keeps a subset of the attributes from $\mathbb{A}$. The join we consider is {\em natural join}, where tuples from two relations can be joined if they share the same value on the common attributes.  A join-project query is {\em full} if $\bA = \mathbb{A}$. Unlike prior work on ranked enumeration, in this paper we place no restriction on the set of attributes in the projection operator.  For simplicity of presentation, we do not consider selections; these can be easily incorporated into our algorithms. 
As an example, the SQL query in~\autoref{ex:1} corresponds to the following query: $\pi_{A,B} (R_1(A,C) \Join R_2(B,C) )$. 

A database $D$ is a set of relations, whose size is defined as the total number of tuples in all relations denoted as $|D|$.  For tuple $t$, we will use the shorthand $t[A]$ to denote $\pi_A(t)$. {We use $\edges$ to denote the set of all relations in the database.}

\introparagraph{Acyclic Queries and Join Trees} A join-project query $Q$ is \emph{acyclic} if and only if it admits a {\em join tree} $\mT$. In a join tree, each relation is a node, and for each attribute $A$, all nodes in the tree containing $A$ form a connected subtree.  For simplicity, we will use node $i$ to refer to the node corresponding to relation $R_i$ in $\mT$. 
Given a join tree $\mT$, pick any node to be the root, and then orient each edge towards the root. Let $\mT_i$ be the subtree rooted at node $R_i$. 
Let $\prnt{R_i}$ be the (unique) parent of $R_i$, and $\key{R_i} = R_i \cap \prnt{R_i}$ {to be the {\em anchor} attributes between $R_i$ and its parent}.  Let $\child{R_i}$ be the set of children nodes of $R_i$. Finally, we fix the ordering of the projection attributes in $\bA$ to be the order of visiting them in the in-order traversal of $\mT$. {Finally, we define $\bA^\pi_i$ as the ordered set of projection attributes in subtree rooted at node $i$ (including projection attributes of node $i$).} As a convention,  we define $\key{r} = \emptyset, A^\pi_{r} = \emptyset$  for the root $r$ and $\child{R_i} = \emptyset$ for a leaf node $R_i$.

\begin{example}
	\label{exp:join-tree}
	Consider a join-project query $Q = \pi_{A,E} ( R_1(A, B) \Join R_2\\(B, C) \Join R_3(C, D) \Join R_4(D, E) )$ under the ranking function \SUM\ defined over attributes $A,E$. In other words, for every output tuple $t$, the score of the tuple is $t[A]+t[E]$.~\autoref{fig:decomp} shows the join tree for the query. We fix $R_3$ as the root with $R_2$ as the left child and {$R_4$} (a leaf node) as the right child. $R_1$, as a leaf node, is also the only child of $R_2$.
\end{example}
	
\begin{figure}[t]
	\includegraphics[scale=0.34]{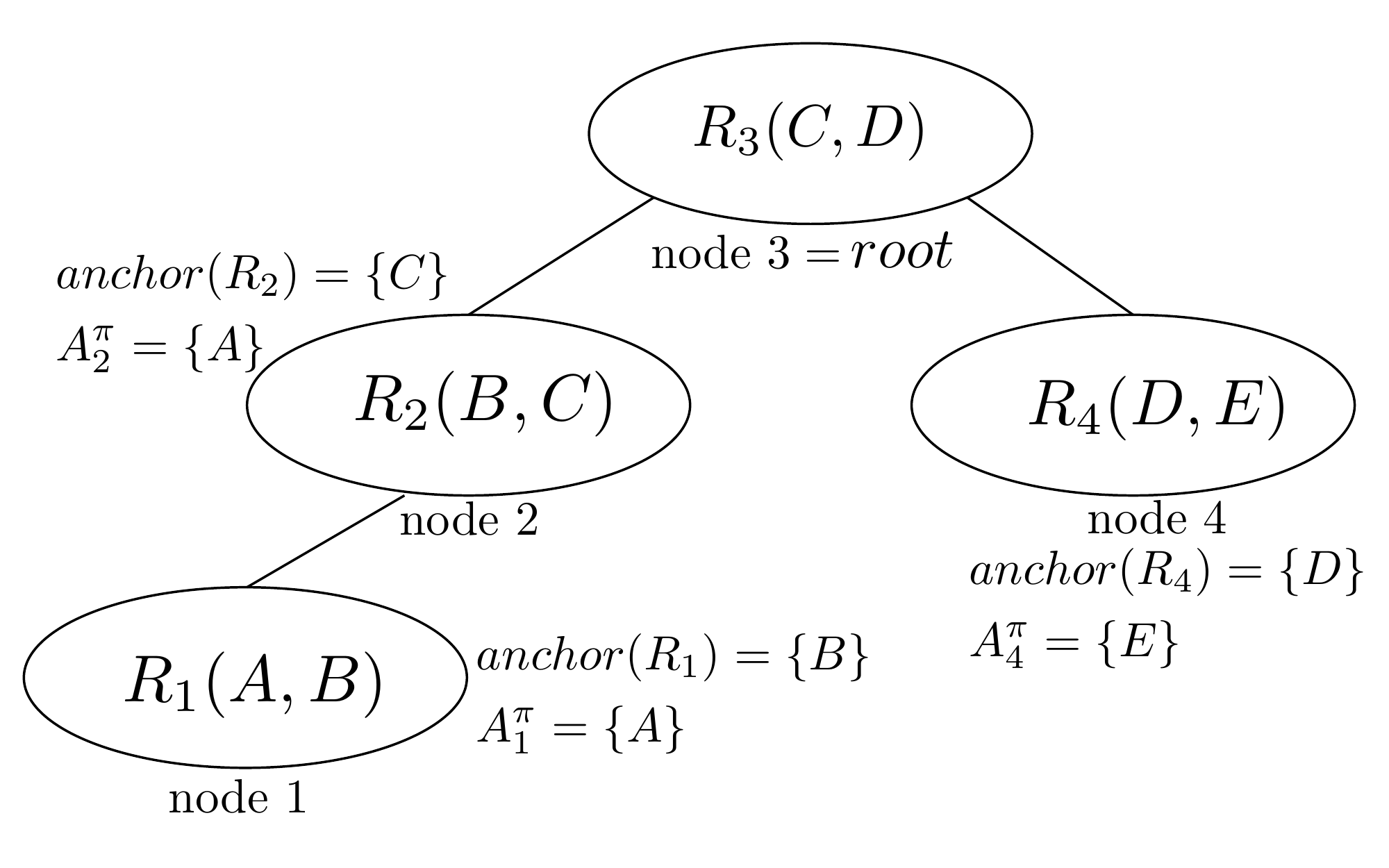}
	\caption{Illustration of join tree for a join-project query $Q = \pi_{A,E} ( R_1(A, B) \Join R_2(B, C) \Join R_3(C, D) \Join R_4(D, E))$. } \label{fig:decomp}
\end{figure}	

\begin{figure}
	\centering
	\includegraphics[scale=0.35]{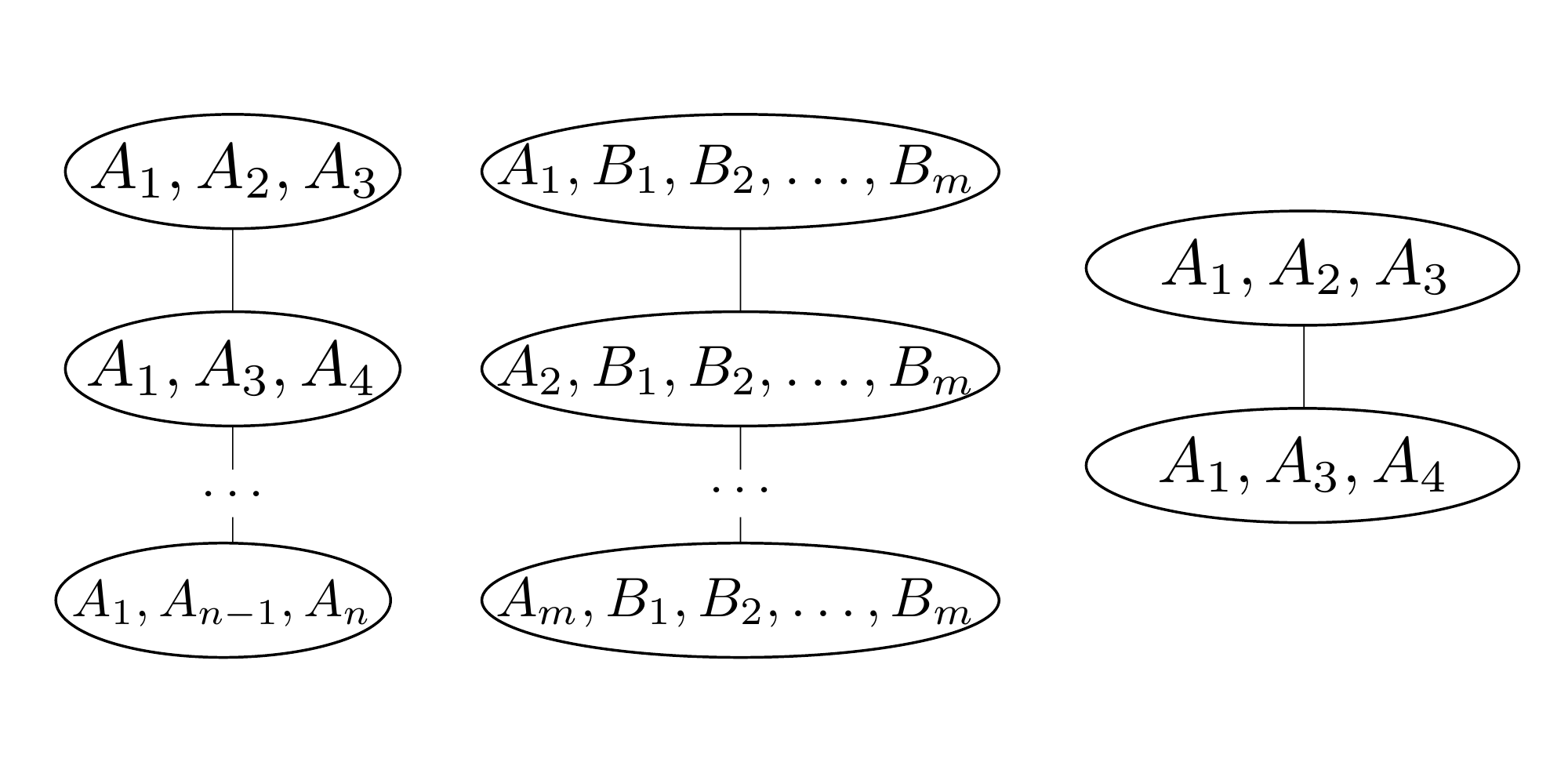}
	\caption{{Examples of GHD and \fhw. The leftmost is the minimal GHD of a cycle join $Q = R_1(A_1, A_2) \Join R_2(A_2, A_3) \Join \cdots \Join R_{n-1}(A_{n-1},A_n) \Join R_n(A_n,A_1)$ with $\fhw = 2$. The middle is the minimal GHD of a bi-clique join $Q = \Join_{i\in[n], j\in[m]} R_{(i-1)m + j}(A_i, B_j)$ with $\fhw = m$. The rightmost is the minimal GHD of a butterfly join $Q = R_1(A_1, A_2) \Join R_2(A_2, A_3) \Join R_3(A_1, A_4) \Join R_4(A_4, A_3)$ with $\fhw = 2$. }}
	\label{fig:ghd}
\end{figure}

{
\smallskip
\introparagraph{Generalized Hypertree Decompositions} A weight assignment $\bu = (u_F)_{F \in \edges}$ is said to be a fractional edge cover if $(i)$ for every $F \in \edges$, $u_F \geq 0$; and $(ii)$ for every $X \in \mathbb{A}, \sum_{F: X \in F} u_F \geq 1$. The \textit{fractional edge cover number} for $\mathbb{A}$, denoted $\rho^*(\mathbb{A})$ is the minimum of $\sum_{F \in \edges} u_F$ over all possible edge covers. A generalized hypertree decomposition (GHD) of a query $Q$ is a tuple $(\htree, (\bag_t)_{t \in V(\htree)})$ where $\htree$ is a tree and every $\bag_t$ (called the bag of $t$) is a subset of $\mathbb{A}$ for each node $t$ of the tree such that: $(i)$ each variables of each $F \in \edges$ is contained in some bag; and $(ii)$ for each $A \in \mathbb{A}$, the set of nodes $\setof{t}{A \in \bag_t}$ is connected in $\htree$. The fractional hypertree width of a decomposition is defined as $\max_{t \in V(\htree)} \rho^*(\bag_t)$, where $\rho^*(\bag_t)$ is the fractional edge cover number of the attributes in $\bag_t$. The fractional hypertree width of a query $Q$, denoted $\fhw(Q)$, is the minimum fractional hypertree width over all possible GHDs of $Q$.~\autoref{fig:ghd} gives examples of GHDs of popular queries and their width. For an acyclic query, it holds that $\fhw = 1$ and any join tree is a valid GHD.

}

\introparagraph{Computational Model}
To measure the running time of our algorithms, we use the uniform-cost RAM 
model~\cite{hopcroft1975design}, where data values as well as pointers to
databases are of constant size. Throughout the paper, all complexity results are  with respect to data complexity, where the query is assumed fixed. It is important to note that we focus on the main memory setting. We further assume existence of perfect hashing that allows constant time lookups in hash tables. 

\subsection{Ranking Functions}

The ordering of query results in $Q(D)$ can be specified by a {\em ranking function}, or through
the \sqlhighlight{ORDER BY} clause of a SQL query in practice.  Formally,  a total order $\succeq$ on the tuples in $Q(D)$ defined over the attributes $\bA$, is induced by a ranking function $\rf$ that maps each tuple $t \in Q(D)$ to a real number $\rank(t) \in \mathbb{R}$. In particular, for two tuples $t_1, t_2$, we have $t_1 \succeq t_2$ if and only if $\rank(t_1) \geq \rank(t_2)$. We assume that $\domain(A)$ for any $A \in \bA$ is also equipped with a total order $\succeq$. We present an example of a ranking function below.

\begin{example} \label{ex:vertex} Consider a function $w:\domain(A) \to \mathbb{R}$ for any  attribute $A \in \bA$.  
For each query result $t$, we define its rank as $\rf(t) = \sum_{A \in \bA} w(t[A])$, the total sum of the weights over all attributes in $\bA$. 
\end{example}

We will focus on \SUM\ and \lexi\ in this paper. We note that both functions 
are instantiations of a more general class of \emph{decomposable functions}~\cite{deep2019ranked}. The ideas introduced for \SUM\ and \lexi\ are readily applicable to more complicated functions including products, a combination of sum and products, etc.

\subsection{Problem Parameters}
\label{sec:problem-parameter}

Given a join-project query $Q$  and a database $D$, an enumeration query asks to enumerate the tuples of $Q(D)$ according to some specific ordering defined by $\rf$. We study this problem in a similar framework as~\cite{Segoufin15}, where an algorithm is decomposed into:
\begin{packed_item}
	\item a {\bf preprocessing phase} that takes time $T_p$ and computes a data structure of size $S_p$, and
	\item an {\bf enumeration phase} {(i.e. the online query phase) that outputs $Q(D)$ without duplicates under the specified ordering whenever a user query is issued. } This phase has full access to any data structures constructed in the preprocessing phase.
	The time between outputting any two consecutive tuples (and also the time to output the first tuple, and the time to notify that the enumeration has completed after the last tuple) is at most $\delta$. 
\end{packed_item}

Prior work~\cite{deep2019ranked} has shown that for acyclic joins without projections, there exists an algorithm with $T_p =  S_p = O(|D|)$ that can achieve $\delta = O(\log|D|)$ delay under ranking. However, the problem of ranked enumeration when projections are involved is wide open. 

{\introparagraph{Using Existing Algorithms} One possible solution to the problem is to set the weights of non-projection attributes to $0$. This will ensure that for \sqlhighlight{SUM} function, only the projection attributes are considered in the ranking and existing algorithms for full join queries could be used. However, this proposal gives poor delay guarantees and is as expensive as enumerating the full join result. For example, for the four path query in~\autoref{exp:join-tree}, the output of the query could be constant in size but the full join can be as large as $\Omega(|D|^2)$ which is prohibitively expensive, but our algorithm would only require $O(|D|)$ in this case. In general, a join with $\ell$ relations may require as much as $\Omega(|D|^{\ell-1})$ time to output the smallest tuple. We describe more details and the formal proof in~\autoref{sec:existing}.}

	\section{General acyclic queries} \label{sec:acyclic}

We first describe the main algorithm of enumerating acyclic join-project queries for \SUM\ ordering in~\autoref{subsec:general}, followed by a specialized algorithm for \lexi\ ordering in~\autoref{sec:lex}. {Before we describe the algorithm, we introduce two key data structures that will be used: \textit{cell} and \textit{priority queues}.
	
	\begin{definition}
		A cell, denoted as $c = \langle t, [p_1, \dots, p_k], q \rangle$, is a vector consisting of three values: (i) a tuple $t \in R_i$ for node $i$ in the join tree $\mT$, (ii) an array of pointers $[p_1, \dots, p_k]$ where the $\ell^{th}$ pointer points to a cell defined for $\ell^{th}$ child of node $i$ in $\mT$, (iii) a pointer $q$ that can only point to another cell defined for node $i$.
	\end{definition}
	
Given a cell $c$ defined for node $i$, one can reconstruct the tuple over $\bA^\pi_i$in constant time (dependent only on the query size, which is a constant) by traversing the pointers recursively. We will use $\anst{c}$ to denote the utility method that performs this task. Note that the time and space complexity of creating a cell is $O(1)$ since the size of the query and the database schema is assumed to be a constant. This implies that we only need to insert/access a constant number of entries in the vector representing a cell. Similarly, $\anst{c}$ also takes $O(1)$ time since the join tree size is a constant.

\introparagraph{Priority queue}	A priority queue is a data structure for maintaining a set $S$ of elements, each with an associated value called a \textsf{key}. The space complexity of a priority queue containing $|S|$ elements is $O(|S|)$. We will use an implementation of a priority queue (e.g., a Fibonacci heap~\cite{fredman1987fibonacci}) with the following properties: (i) an element can be inserted in $O(1)$ time, (ii) the min element can be obtained in $O(1)$ time, and (iii) the min element can be popped and deleted in $O(\log |D|)$ time. We will use the priority queue in conjunction with a cell in the following way: for two cells $c_1$ and $c_2$, the priority queue uses $\rf(\anst{c_1})$ and $\rf(\anst{c_2})$ in the comparator function to determine the relative ordering of $c_1$ and $c_2$. If $\rf(\anst{c_1}) = \rf(\anst{c_2})$, then we break ties according to the lexicographic order of $\anst{c_1}$ and $\anst{c_2}$. The choice of lexicographic ordering is not driven by any specific consideration; as long as the ties are broken consistently, we can use other tie-breaking criteria too. Once again, the comparator function only takes a $O(1)$ time to compare since the ranking function $\rf(\anst{c})$ can be evaluated in constant time.

}

\vspace{-0.5em}

\subsection{General Algorithm} \label{subsec:general}

In this section, we present the algorithm for~\autoref{thm:general}. {At a high level, each node $i$ in the join tree will materialize, in an incremental fashion, all tuples over the attributes $\bA^\pi_i \cup \key{R_i}$ in sorted order.} In order to efficiently store the materialized output, we will use the cell data structure. Since we need to sort the materialized output, each node in the join tree maintains a set of priority queues indexed by $\pi_{\key{R_i}}(u), u \in R_i$. The values of the priority queue are the cells of node $i$. {For example, given the join tree from~\autoref{exp:join-tree}, node $2$ containing $R_2$ will incrementally materialize the sorted result of the subquery $\pi_{C,A} (R_2(B,C) \Join R_1(A,B))$ that is indexed by the values $\pi_C (R_2(B, C))$ since $\bA^\pi_2 = \{A\}$ and $\key{R_2} = \{C\}$. Note that there may be multiple possible join trees for a given acyclic query. Our algorithm is applicable to all join trees. In fact, any node in the join tree can be chosen as the root without any impact on the time and space complexity.}

\introparagraph{Preprocessing Phase} We begin by describing the algorithm for preprocessing in~\autoref{algo:preprocess}. We assume that a join tree has been fixed and the input instance $D$ does not contain any dangling tuples, i.e., tuples that will not contribute in the join; otherwise, we can invoke the Yannakakis algorithm~\cite{yannakakis1981algorithms} to remove all dangling tuples. We initialize a set of empty priority queues for every node in the join tree. We proceed in a bottom up fashion and perform the following steps. For each leaf relation $R_i \in \mT$, we create a cell $\langle t, [], \bot \rangle$ for each tuple $t \in R_i$ and insert it into $\textsf{PQ}_{i}[\pi_{\key{R_i}}(t)]$. For each non-leaf relation $R_j \in\mT$, we create a cell for $t \in R_j$, which points to the top of the priority queue in each child node of $R_j$ that can be joined with $t$. This cell is then added to the  priority queue $\textsf{PQ}_{i}[\pi_{\key{R_i}}(t)]$. Note that we only have one priority queue for the root relation $r$ since $\key{r} = \emptyset$ by definition. 

\begin{algorithm}[!htp]
	\SetCommentSty{textsf}
	\DontPrintSemicolon 
	\SetKwInOut{Input}{Input}
	\SetKwInOut{Output}{Output}
	\SetKwFunction{len}{\textsf{len()}}
	\SetKwFunction{topmost}{\textbf{top()}}
	\SetKwFunction{pop}{\textbf{pop()}}
	\SetKwFunction{first}{\textbf{first}}
	\SetKwFunction{append}{\textsf{append}}
	\SetKwFunction{insertt}{\textbf{insert}}
	\SetKwFunction{listmerge}{\textbf{ListMerge}}
	\SetKwData{pq}{$\mathsf{PQ}$}
	\Input{Input query $Q$, database instance $D$; join tree $\mT$; ranking function $\rank$.}
	\Output{Priority Queues $\pq$}
	\SetKwData{ptr}{\textsf{ptr}}
	\SetKwProg{myproc}{\textsc{procedure}}{}{}
	\SetKwData{return}{\textbf{return}}
	\SetKwData{dedup}{\textsf{dedup}}
	\SetKwData{counter}{\textsf{counter}}
	\SetKwFunction{isequal}{\textsc{compare}}
	\SetKwFunction{rank}{\textsc{rank}}
	
	
	\ForEach{$R_i \in \mT$ in post order traversal}{
		\ForEach{$t \in R_i$}{
			$u \leftarrow \pi_{\key{R_i}}(t);$ \;
			\If{$\pq_{i}[u]$ does not exist}{
				$\pq_{i}[u] \leftarrow \emptyset;$    \  \   \  \   \  \tcc{Initialize a priority queue}
			}
			$L \leftarrow \emptyset$; \;
			\ForEach{$R_j$ is the child of $R_i$ \label{for:1}}{ 
				$L.\insertt(\pq_{j}[\pi_{\key{R_j}}(u)].\topmost)$; \label{line:insert}
			}
			
		}
		$\pq_{i}[u].\insertt(\langle t, L, \bot \rangle)$;
	}
	
	\caption{{\sc PreprocessAcyclic}}
	\label{algo:preprocess}
\end{algorithm}

\begin{algorithm}[!htp]
	\SetCommentSty{textsf}
	\DontPrintSemicolon 
	\SetKwInOut{Input}{Input}
	\SetKwInOut{Output}{Output}
	\SetKwFunction{len}{\textsf{len()}}
	\SetKwFunction{topmost}{\textbf{top()}}
	\SetKwFunction{pop}{\textbf{pop}}
	\SetKwData{next}{\textsf{next}}	
	\SetKwFunction{append}{\textsf{append}}
	\SetKwFunction{insertt}{\textsf{insert}}
	\SetKwFunction{break}{\textsf{break}}
	\SetKwProg{myproc}{\textsc{procedure}}{}{}
	\SetKwData{pq}{$\mathsf{PQ}$}
	\Input{Input query $Q$, database instance $D$; join tree $\mT$; ranking function $\rank$; Priority queues $\pq$}
	\Output{$Q(D)$ in ranked order}
	\SetKwData{last}{\textsf{last}}
	
	\SetKwData{return}{\textbf{return}}
	\SetKwData{next}{\textsf{next}}
	\SetKwData{temp}{\textsf{temp}}
	\SetKwData{counter}{\textsf{counter}}
	\SetKwFunction{enum}{\textsc{Enum}}	
	\SetKwFunction{topdown}{\textsc{Topdown}}
	\SetKwFunction{isequal}{\textsc{is\_equal}}
	\SetKwFunction{rank}{\textsc{rank}}

	\myproc{\enum{}}{
		$\last \leftarrow \emptyset$; \;
		\While{$\pq_r[\emptyset] \neq  \emptyset$}{
			$o \leftarrow \pq_r[\emptyset].\topmost$; \;
			\If{\isequal{$o, \last$} $=$ false \label{line:if} }{
				{\bf print} $\mathsf{output}(o), \last \leftarrow o$;  \tcc*{new output}
				 \label{user:output}
			}
			\topdown$(o, r)$;\; 
		}
	}
	
	\myproc{\topdown{$c, $ $j$}   /* $c = \langle t,  [p_1, \dots, p_k], \next \rangle$ */ }{ 
		$u \leftarrow \pi_{\key{R_j}}(c.t)$;\;
		\If{$c.\next = \bot$}{
			\While{true}{
				$\temp \leftarrow \pop(\pq_{j}[u])$;\;
				\ForEach{$R_i$ is a child of $R_j$ \label{for:loop}}{
					$p'_i \leftarrow \topdown(c.p_i, i)$ ;\;
					\If{$p_i' \neq \bot$ }{ 
						$\pq_{j}[u].\insertt(\langle t,  [c.p_1, \dots, p_i', \dots c.p_k], \bot \rangle)$ \label{for:loopend}
					}
				}
				\If {$R_j$ is not the root}{
					$c.\next \gets {\mathsf{addressof}}(\pq_{j}[u].\topmost)$;
				}
				\lIf{\isequal{$\temp, \pq_{j}[u].\topmost$} = false \label{equal:check}}{
					\break;
				}			
			}
		}
		\KwRet{$c.\next$};\;
	}
	\myproc{\isequal{$c_1, c_2$}}{
		\lIf{\small \rank{$\mathsf{output}(c_1)$}$ \neq  $ \rank{$\mathsf{output}(c_2)$}}{
			\texttt{return false}}
		\ForEach{$A \in \bA$}{
			\lIf{$\mathsf{output}(c_1)[A] < \mathsf{output}(c_2)[A]$}{\texttt{return false}}
		}
		\texttt{return} true;
	}

	\caption{{\sc EnumAcyclic}}
	\label{algo:enumerate}
\end{algorithm} 

\begin{figure*}[!htp]
	\centering
	\hspace*{0em}
	\begin{subfigure}{0.45\linewidth}
		\includegraphics[scale=0.4]{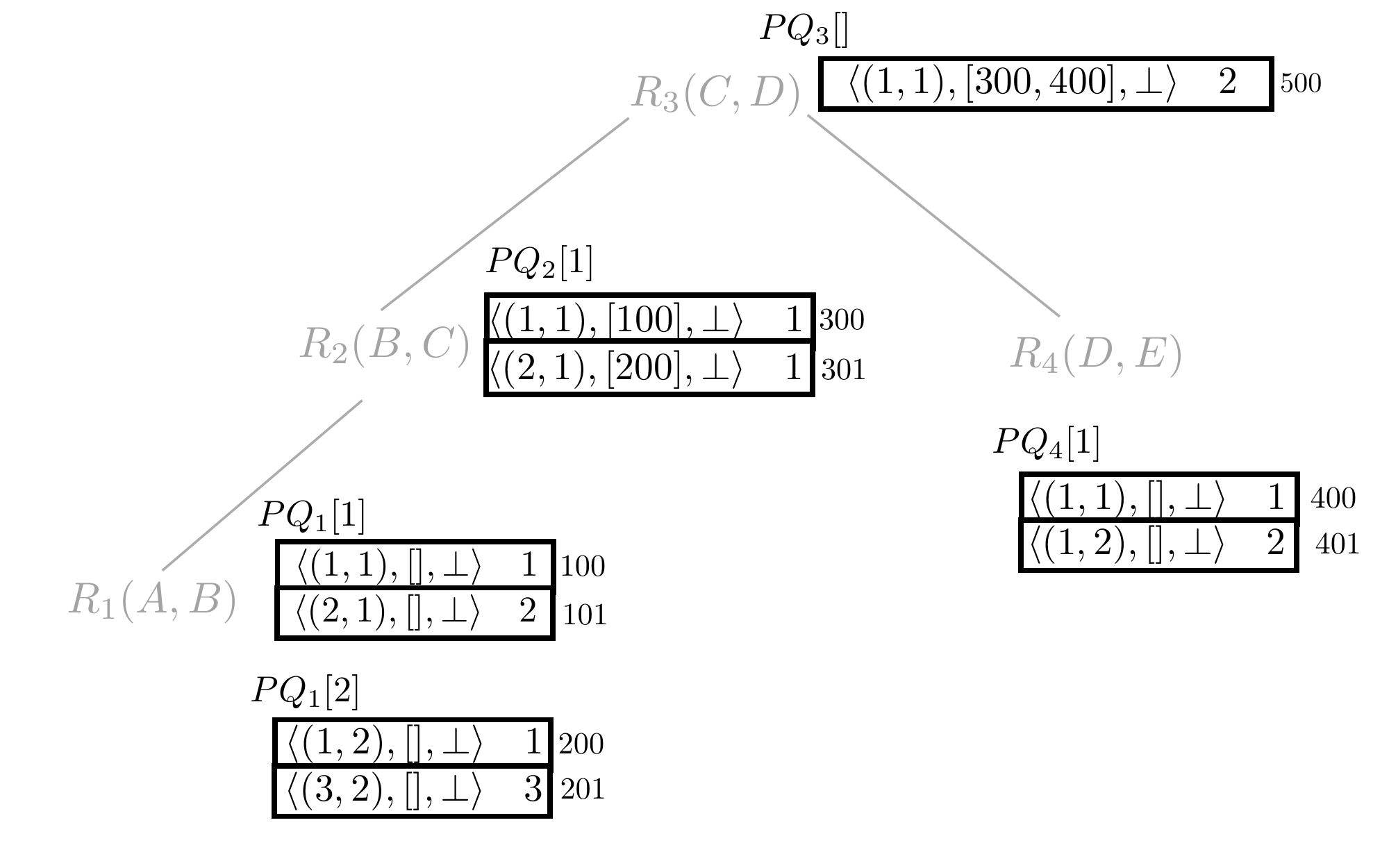}
		\caption{Data structure state after the preprocessing phase. Each memory location has a cell and the partial score of the partial answer}  \label{fig:preprocess}
	\end{subfigure}
	\hspace{4em}
	\begin{subfigure}{0.45\linewidth}
		\includegraphics[scale=0.4]{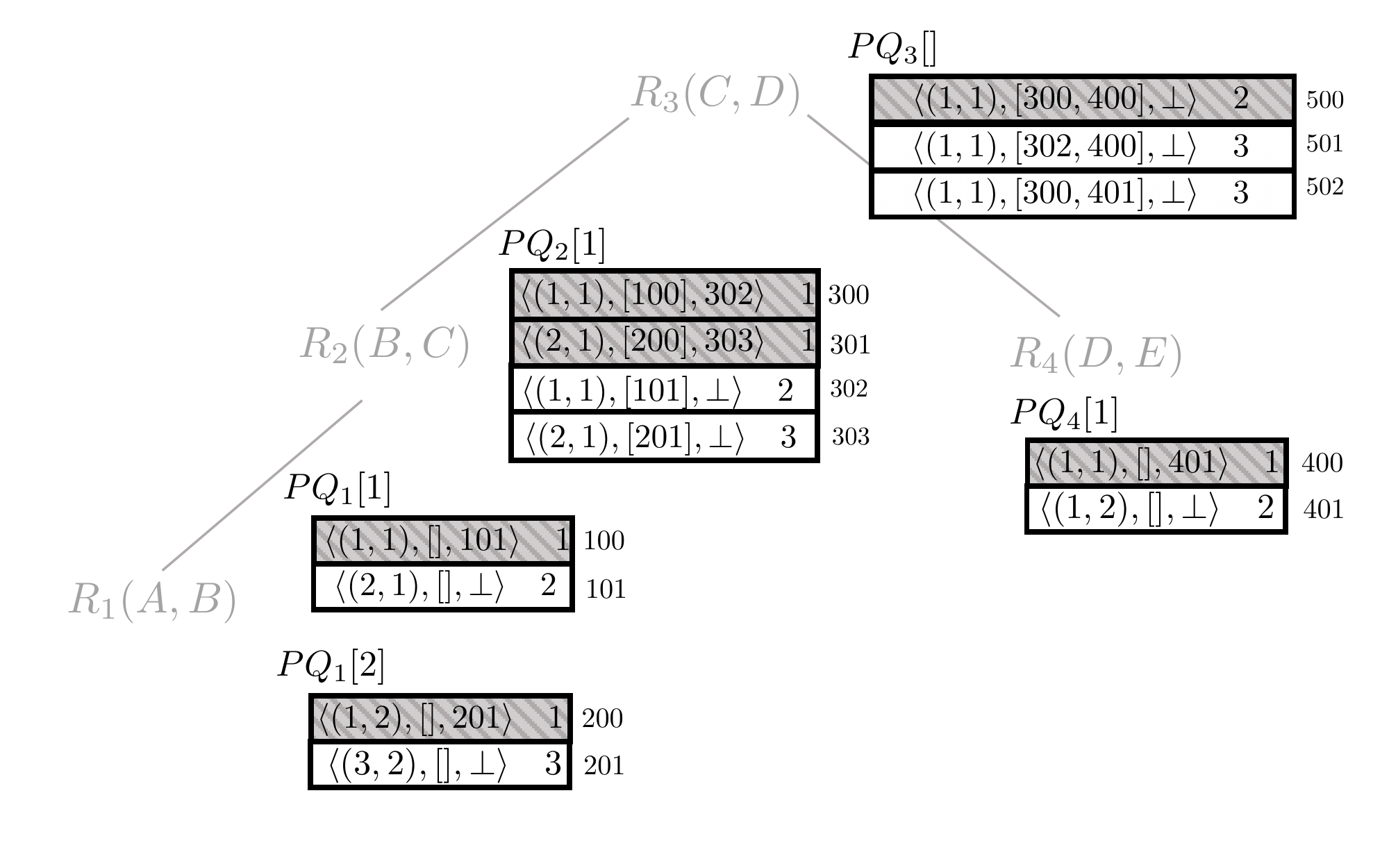}
		\caption{Data structure after one iteration of procedure \textsc{Enum()}}  \label{fig:enumerate}
	\end{subfigure}
\caption{Example to demonstrate the preprocessing and enumeration phase of the general algorithm}
\end{figure*}

\begin{example}
	Continuing with the 4-path query running example, consider the following instance $D$ as shown below.
	
	\scalebox{0.80}{
	\begin{minipage}[t]{0.24\linewidth}
		\centering
		\begin{tabular}[t]{ !{\vrule width1pt} c|c  !{\vrule width1pt} } 
			\Xhline{1pt}
			${A}$ & ${B}$ \\ 
			\Xhline{1pt}
			1 & 1 \\ 
			2 & 1 \\
			1 & 2 \\ 
			3 & 2 \\			
			\Xhline{1pt}
		\end{tabular}		
		\vspace{1em}
		$R_1$
	\end{minipage}
	\begin{minipage}[t]{0.24\linewidth}
		\centering
		\begin{tabular}[t]{ !{\vrule width1pt} c|c !{\vrule width1pt} } 
			\Xhline{1pt}
			${B}$ & ${C}$ \\ 
			\Xhline{1pt}
			1 & 1 \\ 
			2 & 1 \\
			\Xhline{1pt}
		\end{tabular}		
		\vspace{1em}	
		$R_2$
	\end{minipage}
	\begin{minipage}[t]{0.24\linewidth}
		\centering
		\begin{tabular}[t]{ !{\vrule width1pt} c|c !{\vrule width1pt} } 
			\Xhline{1pt}
			${C}$ & ${D}$ \\ 
			\Xhline{1pt}
			1 & 1 \\ 
			1 & 2 \\			
			\Xhline{1pt}
		\end{tabular}		
		\vspace{1em}	
		$R_3$
	\end{minipage}
	\begin{minipage}[t]{0.24\linewidth}
		\centering
		\begin{tabular}[t]{ !{\vrule width1pt} c|c !{\vrule width1pt} } 
			\Xhline{1pt}
			${D}$ & ${E}$ \\ 
			\Xhline{1pt}
			1 & 1 \\ 
			1 & 2 \\			
			\Xhline{1pt}
		\end{tabular}		
		\vspace{1em}	
		$R_4$
	\end{minipage}}

	As we saw before, \autoref{fig:decomp} shows the join tree along with the anchor attributes in each relation.~\autoref{fig:preprocess} shows the  state of priority queues after the preprocessing step. After the full reducer pass, tuple $(1,2)$ is removed from $R_3$ because there is no join tuple that can be formed using it. Then, we start constructing the cells for each node starting with the leaf nodes. Since $B$ is the anchor for relation $R_1$, we create two priority queues $\textsf{PQ}_{1}[1]$ and $\textsf{PQ}_{1}[2]$. For $\textsf{PQ}_{1}[1]$, we create the cells for tuples $(1,1)$ and $(2,1)$. For convenience, the cells are followed by the partially aggregated score. Consider relation $R_2(B,C)$. The cell for tuple $(1,1)$ in $\textsf{PQ}_{2}[1]$ points to the top of $\textsf{PQ}_{1}[1]$ (shown as pointer with address $100$). The root bag consists of a single tuple entry which points to the cells at locations $300$ and $400$. The output tuple that can be formed by the root bag is $(A = 1, E = 1)$.
\end{example}

\introparagraph{Enumeration Phase} We describe the enumeration procedure in~\autoref{algo:enumerate}. The high-level idea is to output answers by repeatedly popping elements from the root priority queue. It may be possible that multiple tuples of the root priority queue output the same final result. In order to deduplicate answers, we compare the answer at the current top of the priority queue with the previous answer (\autoref{line:if}), and output it only if they are different.
Then, we invoke the procedure \textsc{Topdown} to insert new candidates into the priority queue.  This procedure will be recursively propagated over the join tree until it reaches the leaf nodes. Observe that once the new candidates have been inserted, the \textsf{next} pointer of a cell is updated by pointing to the topmost element in the priority queue. This chaining materializes the answers for a particular node that can be reused and is key to avoiding repeated computation.

\begin{example}
	Continuing our running example,~\autoref{fig:enumerate} shows the state of the priority queues after one complete iteration of procedure \textsc{Enum()}. We first pop the only element in root priority queue and note that the output tuple $(A=1, E=1)$ is enumerated. Then we call \textsc{Topdown} with cell at memory $500$ and root (node $3$) as arguments (denoted as $\textsc{Topdown}(*500, 3) $). The \textsf{next} for the cell is $\bot$ so we pop the cell at $500$ from the priority queue (shown as greyed out in the figure) and recursively call $\textsc{Topdown}(*300, 2) $. The cell at memory location $300$ has $\textsf{next} = \bot$. Therefore, we enter the while loop, pop the cell and recursively call $\textsc{Topdown}(*100, 1)$. We have now reached the leaf node. The anchor attribute value for cell at $100$ is $u = 1$, so we pop the current cell from $\textsf{PQ}_{1}[1]$ (greyed out cell at $100$), find the next candidate at the top of $\textsf{PQ}_{1}[1]$ (which is cell at $101$), chain it to the cell at $100$ by assigning $\textsf{next} = 101$ and return the cell at $101$ to the parent. When the program control returns from the recursive call back to node $2$, we create a new cell (at memory address $302$) that points to $101$ and insert it into the priority queue. However, observe that the cell at memory location $301$ also generates $A=1$, a duplicate since cell at $300$ also generated it. This is where the equality check at~\autoref{equal:check} comes in. Since both cells at $300$ and $301$ generate the same value, we also pop off the cell at $301$ in the subsequent while loop iteration, find its next candidate and create the cell at $303$, and insert into the priority queue. This ensures that all elements in $\textsf{PQ}_{2}[1]$ generating the same $A$ value are removed, ensuring no duplicates at the root level. Finally, the control returns to the root level \textsc{Topdown} call. The recursive call to the right child (node $4$) create a new cell $401$ and we insert two cells at the root priority queue, cell $501$ and $502$ that correspond to output tuple $(A=2, E=1) $ and $(A=1, E=2)$ respectively.
\end{example}

{
	We are now ready to formally prove~\autoref{thm:general}.
	
	\begin{restatable}{lemma}{lemdelay} \label{lem:delay}
		The delay guarantee of \textsc{EnumAcyclic} is at  most $O(|D| \log |D|)$.
	\end{restatable}
	
	\begin{restatable}{lemma}{lempreprocess} \label{lem:preprocess}
		{\sc PreprocessAcyclic}  running in $O(|D| \log |D|)$ time, generates a data structure of size $O(|D|)$.
	\end{restatable}
	
	\begin{restatable}{lemma}{lemcorrect} \label{lem:correctness}
		\textsc{EnumAcyclic} enumerates the query result $Q(D)$ in ranked order correctly.
	\end{restatable}
	
	Together, the above lemmas establish~\autoref{thm:general}. We defer the full proofs to~\autoref{sec:correctness}.} We also show how we can recover logarithmic delay guarantee for full queries from~\cite{tziavelis2020optimal, deep2019ranked}.

\subsection{Improvement for Lexicographic Ranking} \label{sec:lex}

The algorithm from last section is also applicable to \lexi\ ranking function. In fact, we can transform \lexi\ with an attribute ordering of $A_1, A_2, \cdots, A_m$, into \SUM\ by defining a ranking function $\rank(t) = \sum_{i =1}^{m} 10^{m-i} \cdot w(\pi_{A_i} (t))$ for tuple $t$, while preserving the \lexi\ ordering. In this section, we present an alternative algorithm by exploiting the special structural properties of \lexi, that the global ranking also implies local ranking over every output attribute. 
Moreover,  it admits to enumerate query results not only in lexicographic order as given by \sqlhighlight{ORDER BY} $A_1, A_2, \cdots, A_m$ but also arbitrary ordering on each attribute (for instance, \sqlhighlight{ORDER BY} $A_1$ \sqlhighlight{ASC}, $A_2$ \sqlhighlight{DESC} $\dots$).

	\smallskip
	\introparagraph{Preprocessing Phase} In this phase, we perform the full reducer pass to remove all dangling tuples and create hash indexes for the base relations in sorted order. We also sort $\domain(A_i)$.
	
	\smallskip
	\introparagraph{Enumeration Phase}  Given an attribute order of output attributes $\bA = \{A_1, A_2, \cdots, A_m\}$, we start by fixing the minimum value in $\domain(A_1)$ as $a_1$. Then, we perform the two-phase semi-joins to remove tuples that cannot be joined with value $a_1$, and find the values in $\domain(A_2)$ that survive after semi-joins, denoted as $\mL_{A_2}(a_1)$.  Similarly, we fix the minimum value in $\mL_{A_2}(a_1)$ as $a_2$, and perform the two-phase semi-joins for finding the values in $\domain(A_3)$ that can be joined with both $a_1, a_2$. We continue this process until all attributes in $\bA$ have been fixed, and end up with enumerating such a query result (with fixed values). Then, we backtrack and continue the process until all values in attribute $A_1$ are exhausted. 
	
	\begin{algorithm}[!htp]
		\SetCommentSty{textsf}
		\DontPrintSemicolon 
		\SetKwInOut{Input}{Input}
		\SetKwInOut{Output}{Output}
		\SetKwFunction{len}{\textsf{len()}}
		\SetKwFunction{topmost}{\textbf{top()}}
		\SetKwFunction{pop}{\textbf{pop()}}
		\SetKwFunction{first}{\textbf{first}}
		\SetKwFunction{append}{\textsf{append}}
		\SetKwFunction{listmerge}{\textbf{ListMerge}}
		\SetKwData{pq}{$\mathsf{PQ}$}
		\Input{Input query $Q$, database $D$}
		\Output{$Q(D) \ltimes t$ in lexicographic order of $A_i, \cdots, A_m$}
		\SetKwData{ptr}{\textsf{ptr}}
		\SetKwProg{myproc}{\textsc{procedure}}{}{}
		\SetKwData{return}{\textbf{return}}
		\SetKwData{dedup}{\textsf{dedup}}
		\SetKwData{counter}{\textsf{counter}}
		\SetKwFunction{join}{\textsc{Join}}
		\SetKwFunction{rank}{\textsc{rank}}
		
		
		\lIf{$i = m$}{{\bf output} $t$ and \return; }
		\ForEach{$a \in \mL$}{
			$\mL' \leftarrow \pi_{A_{i+1}}(\sigma_{A_i = a} (R_{i+1} \ltimes t))$ \label{line:semijoin}; \tcc*{by semi-joins} 
			$t' \leftarrow (t,a)$;\tcc*{create new tuple}
			{\sc EnumAcyclicLexi}($t', \mL', i+1$);\;
		}
		
		\caption{{\sc EnumAcyclicLexi}($t$, $\mL$, $i$)}
		\label{algo:lexicographic}
	\end{algorithm} 
	
	Algorithm~\ref{algo:lexicographic} takes as input an acyclic query $Q$, an database $D$, an integer $i \in \{1,\cdots, m\}$, a tuple $t$ defined over attributes $A_1, \cdots, A_{i-1}$, and a set of values $\mL \subseteq \domain(A_i)$ that can be joined with $t$ in $D$. The original problem can be solved by invoking \\$\textsc{EnumAcyclicLexi}(\emptyset, \domain(A_1), 1)$ for sorted $\domain(A_1)$. 
	
	{
		\begin{restatable}{lemma}{lemlexi} \label{lem:lexi}
		$\textsc{EnumAcyclicLexi}$ enumerates $Q(D)$ correctly in lexicographic order with delay guarantee $O(|D|)$ after preprocessing time $T_p = O(|D|\log|D|)$ and space complexity $O(|D|)$.
	\end{restatable} }

\section{Star Queries} \label{sec:star}

In this section, we present a specialized data structure for the {\em star query}, which is represented as: $ Q^\star_m = \pi_{\bA} (R_1(A_1, B) \Join R(A_2,B) \Join \dots \Join R_m(A_m, B)).$
where $\bA = \{A_1, \cdots, A_m\}$. All relations in a star query join on exactly the same attribute(s). 
In this following, we present a specialized data structure on ranked enumeration for $Q^\star_m$ in Section~\ref{sec:algorithm}, and prove the optimality in Section~\ref{sec:optimality}. 

\subsection{The Algorithm}
\label{sec:algorithm}
Consider the star query $Q^\star_m$, a database $D$ and a ranking function $\rank$.  Now we present a data structure for Theorem~\ref{thm:ranked:star}. 

\begin{algorithm}[t]
	\SetCommentSty{textsf}
	\DontPrintSemicolon 
	\SetKwInOut{Input}{Input}
	\SetKwInOut{Output}{Output}
	\SetKwFunction{len}{\textsf{len()}}
	\SetKwFunction{insertt}{\textbf{insert}}
	\SetKwData{pq}{$\mathsf{PQ}$}
	\Input{Input star query $Q^\star_m$, ranking function $\mathsf{rank}$ and database $D$; degree threshold $\delta \ge 1$}
    \Output{Heavy output $\mO^H$ and priority queue $\pq$}
	\SetKwData{ptr}{\textsf{ptr}}
	
	\SetKwData{return}{\textbf{return}}
	\SetKwData{dedup}{\textsf{dedup}}
	\SetKwData{next}{\textsf{next}}
	\ForEach{$i \in \{1,2,\cdots, m\}$}{
		$R^H_i \leftarrow \{ t \in R_i : |\sigma_{A_i = \pi_{A_i} (t)}| \ge \delta\}$;\;
		$R^L_i \leftarrow \{ t \in R_i : |\sigma_{A_i = \pi_{A_i} (t)}| < \delta\}$;\;
	}
	Compute $\mO^H \leftarrow \pi_{\bA} \left(R^H_1 \Join \dots \Join R^H_m\right)$;\;
	Sort $\mO^H$ by \rank;\;
	\For{$i \in \{0, 1, \dots, m - 1\} $}{
		$Q_i \gets R^H_1 \Join \dots \Join R^H_{m-1}\Join R^L_i \Join R_{i+1} \Join \dots \Join R_m$;\;
		$\mT_i \gets$ a join tree for $Q$ with $R_i$ as root and all other relations as children of $R_i$;\;
		{\sc PreprocessAcyclic}$(Q_i, \mT_i)$;\;  
		$\next \leftarrow ${\sc EnumAcyclic}$(Q_i, \mT_i)$;\;
		$\pq.\insertt(\next)$; \tcc{insert the smallest tuple into $\pq$}}
	
	\caption{{\sc PreprocessStar}}
	\label{algo:preprocess:star}
\end{algorithm} 

\smallskip
\introparagraph{Preprocessing Phase} 
Without loss of generality, assume that there is no dangling tuples in $D$. Moreover, if $\bA$ does not include an attribute $A$, we can remove efficiently $R_i$ using a semi-join.   We first fix a degree threshold $\delta \ge 1$ (whose value will be determined later). For each $i \in \{1,2,\cdots,m\}$, a value $a_i \in \domain(A)$ is \emph{heavy} if it has degree larger than $\delta$ in $R_i$, i.e., $|\sigma_{A=a_i} (R_i)| \geq \delta$, and \emph{light} otherwise. A tuple $t = (a_i, b) \in R_i$ is {\em heavy} if $a_i$ is heavy. For $R_i$, let $R^H_i, R^L_i$ be the set of heavy and light tuples in $R_i$. An output $t = (a_1, a_2, \dots, a_m) \in Q^\star_m(D)$ is \emph{heavy} if $a_i$ is heavy in $R_i$ for each $i \in \{1,2,\cdots,m\}$, and \emph{light} otherwise. In this way, we can divide the output $Q^\star_m(D)$ into $\mathcal{O}^{H}$ and $\mathcal{O}^{L}$, containing all heavy and light output tuples separately. In the preprocessing phase, our goal is to materialize all heavy output tuples ($\mathcal{O}^{H}$) ordered by \rank. Details are described in~\autoref{algo:preprocess:star}. We compute 
$\mO^H = \pi_{\bA} \left(R^H_1 \Join R^H_2 \Join \cdots \Join R^H_m \right)$ by invoking the Yannakakis algorithm~\cite{yannakakis1981algorithms}, and then sort $\mO^H$ by \rf. Next, we insert the smallest query result from $\mO^H$ into the priority queue. Then, we define $m$ different subqueries as
$Q_i = \pi_{\bA} \left(R^H_1 \Join \dots \Join R^H_{i-1}\Join R^L_i \Join R_{i+1} \Join \dots \Join R_m \right)$
where tuples in relation $R_j$ are heavy for any $j < i$ and tuples in relation $R_i$ are light. For such $Q_i$, we consider a join tree $\mT_i$ in which $R_i$ is the root and all other relations are children of $R_i$.  We preprocess a data structure for $Q_i$ with $\mT_i$, by invoking Algorithm~\ref{algo:preprocess}.

\smallskip
\introparagraph{Enumeration Phase} As described in \autoref{algo:enumeration:star}, the high-level idea in the enumeration is to perform a ($m+1$)-way merge over $\mO^H$ and $Q_i$'s. Specifically, we maintain a priority queue $\pq$ with one entry for each subquery $Q_i$ and one entry for $\mO^H$.  Once the smallest element is extracted from \pq\ (say $t$ generated by $Q_i$), we extract the next smallest candidate from $Q_i$ (if there is any) and insert it into $\pq$.  Moreover, finding the smallest candidate output result from $\mO^H$ is trivial since $\mO^H$ have been materialized in a sorted way in the preprocessing phase. We conclude this subsection with the formal statement of the result.

\begin{algorithm}[t]
	\SetCommentSty{textsf}
	\DontPrintSemicolon 
	\SetKwInOut{Input}{Input}
	\SetKwInOut{Output}{Output}
	\SetKwFunction{len}{\textsf{len()}}
	\SetKwFunction{insertt}{\textbf{insert}}
	\SetKwFunction{pop}{\textbf{pop()}}
	\SetKwData{pq}{$\mathsf{PQ}$}
	\Input{Star query $Q^\star_m$, ranking function $\mathsf{rank}$ and database $D$; 
		Output of $\mO^H$ and priority queue $\pq$}
	\Output{$Q^\star_m(D)$ in ranked order }
	\SetKwData{ptr}{\textsf{ptr}}
	
	\SetKwData{return}{\textbf{return}}
	\SetKwData{dedup}{\textsf{dedup}}
	\SetKwData{counter}{\textsf{counter}}
	\SetKwData{next}{\textsf{next}}	
	
	\While{$\pq \neq \emptyset$}{
		$t \leftarrow \pq.\pop$;\;
		{\bf output} $t$; \ \ \ \ \tcc*{enumerate the result}
		\If{$t \notin \mO^H$}{
			$i \leftarrow $ smallest positive index such that $\pi_{A_j}(t)$ is heavy for all $j < i$ and $\pi_{A_i}(t)$ is light;\;
			$\next \leftarrow$ $\textsc{EnumAcyclic}(Q_i, \mT_i)$;\;
			$\pq.\insertt(\next)$;\;
		} \lElse {$\pq.\insertt(\mO^H.\pop)$;}
	}
	
	\caption{{\sc EnumStar}}
	\label{algo:enumeration:star}
\end{algorithm} 

{
\begin{restatable}{lemma}{lemstar} \label{lem:star}
	\autoref{algo:preprocess:star} runs in time $T = O(|D| \cdot (|D|/\delta)^{m-1})$ and requires space $S = O((|D|/\delta)^m)$. \autoref{algo:enumeration:star}  correctly enumerates the result of the query in ranked order with delay $O(|D|\log|D|/\delta)$.
\end{restatable}}

\subsection{Tradeoff Optimality} \label{sec:optimality}

We next present conditional optimality for our tradeoff achieved in Theorem~\ref{thm:ranked:star}. Before showing the proof, we first revisit a result on unranked evaluation for $Q^\star_m$ in~\cite{amossen2009faster}:
\begin{lemma}[\cite{amossen2009faster}] \label{lem:twopath:evaluation}
	There exists a combinatorial\footnote{An algorithm is called combinatorial if it does not use algebraic techniques such as fast matrix multiplication.} algorithm that can evaluate $Q^\star_m$ on any database $D$ in time $O\left(|D| \cdot {|Q^\star_m(D)|}^{1-\frac{1}{m}}\right)$.
\end{lemma}
This result was presented over a decade ago without any improvement since then. Thus, it is not unreasonable to conjecture that \autoref{lem:twopath:evaluation} is optimal. Based on its conjectured  optimality, we can show the following result for unranked enumeration.

\begin{restatable}{lemma}{lemoptimal} \label{lemma:optimal}
	Consider star query $Q^\star_m$, database $D$ and some constant $\epsilon \in [0,1]$. If there exists an algorithm that supports $O(|D|^{1-\epsilon} \log |D|)$-delay enumeration after $O(|D|^{1 + (m-1)\epsilon - \epsilon'})$ preprocessing time for some constant $\epsilon' > 0$, the optimality of Lemma~\ref{lem:twopath:evaluation} will be broken.
\end{restatable} 

The  lower bound holds for any ranking function. \autoref{lemma:optimal} implies that for star queries, both \autoref{thm:general} and \autoref{thm:ranked:star} are optimal. Before concluding this section, we also remark on the question of whether the logarithmic factor that we obtain in the delay guarantee is removable. Prior work~\cite{deep2019ranked} showed that for the following simple join query $Q =  R(x) \Join S(y)$ over \SUM, there exists no algorithm supporting constant-delay enumeration after linear preprocessing time. Note that this does not rule out a sub-logarithmic delay guarantee, which remains an open problem.

	\section{General queries}
\label{sec:cyclic}

In this section, we will describe how to extend the algorithm for acyclic queries to handle cyclic queries. {The key idea is to transform the cyclic query into an acyclic one, by constructing a GHD as defined in~\autoref{sec:framework}. A GHD automatically implies an algorithm for cyclic joins. After materializing the results of the subquery induced by each node in the decomposition, the residual query becomes acyclic.  Hence, we can apply our algorithm for acyclic queries directly obtaining the following}:

\begin{theorem} \label{thm:cq}
	For a join-project query $Q$, a database instance $D$ and a ranking function $\rank \in \{\textrm{\SUM}, \textrm{\lexi} \}$, the query results $Q(D)$ can be enumerated according to $\rank$ with  $O(|D|^\fhw \log |D|)$ delay, after $O(|D|^\fhw \log |D|)$ preprocessing time.
\end{theorem}

We now go one step further and extend our algorithm to queries that are {\em unions} of join-project queries (UCQs) using an idea introduced by~\cite{deep2019ranked, tziavelis2020optimal}. A UCQ query is of the form $Q = Q_1 \cup Q_2 \cup \cdots \cup Q_m$, where each $Q_i$ is a join-project query defined over the same projection attributes $\bA$. Semantically, $Q(D) = \bigcup_i Q_i(D)$. Recent work by Abo Khamis et al.~\cite{abo2017shannon} presents an improved algorithm (called \panda) that constructs multiple GHDs by partitioning the input database into disjoint pieces and build a GHD for each piece. In this way, the size of materialized subquery can be bounded by $O(|D|^\subw)$, where $\subw$ is the \emph{submodular width}~\cite{marx2013tractable} of input query $Q$. Moreover, $\subw \leq \fhw$ holds generally for query $Q$, thus improving the previous result on $\fhw$. By using~\autoref{thm:general} in conjunction with data-dependent decompositions from \panda\, we can immediately obtain the following result:
 
\begin{theorem} \label{thm:ucq}
For a join-project query $Q$, a database instance $D$ and a ranking function $\rank \in \{\textrm{\SUM}, \textrm{\lexi} \}$, the query results $Q(D)$ can be enumerated according to $\rank$ with $O(|D|^\subw \log |D|)$ delay, after $O(|D|^\subw \log |D|)$ preprocessing time.
\end{theorem}

\begin{example}
	Consider the 4-cycle (butterfly) query $\pi_{A,C} (R_1(A,B) \Join R_2(B,C) \Join R_3(C,D) \Join R_4(D,A))$ with ranking function $\rank(t) = \pi_A(t)  + \pi_C(t)$. With $\fhw = 2$, ~\autoref{thm:ucq} implies that the query results can be enumerated according to $\rank$ with $O(|D|^2 \log |D|)$ delay, after $ O(|D|^2)$ preprocessing time. 
	With $\subw = \frac{3}{2}$,~\autoref{thm:ucq} implies that query results can be enumerated according to $\rank$ with delay $O(|D|^{3/2} \log |D|)$,  after $O(|D|^{3/2} \log |D|)$ preprocessing time.
\end{example}
{\introparagraph{A note on optimality} The reader may wonder whether the exponent of \fhw\ and \subw\ in~\autoref{thm:cq} and~\autoref{thm:ucq} are truly necessary. For the triangle query $Q_\triangle(x,y) = R(x,y) \Join S(y,z) \Join T(z,x)$ which is the simplest cyclic query, $\fhw = \subw = 3/2$ and even after $30$ years, the original AYZ algorithm~\cite{alon1994finding} that detects the existence of a triangle in $O(|D|^{3/2})$ time still remains the best known combinatorial algorithm. It is widely conjectured~\cite{marx2021modern, abboud2014popular, abboud2018matching, kopelowitz2016higher} that there exists no better algorithm. As noted in~\cite{abo2017shannon}, the notion of submodular width was suggested as the yardstick for optimality. Indeed, the groundbreaking results by Marx~\cite{marx2013tractable} rules out algorithms with better dependence than \subw\ in the exponent for a small of class of queries but a general unconditional lower bound still remains out of reach. Thus, any improvement in the exponent would automatically imply a better algorithm for cycle detection since ranked enumeration is at least as hard. In~\autoref{sec:lowerbound}, we formally show that the exponential dependence of \subw\ in~\autoref{thm:ucq} is unavoidable subject to popular conjectures.}

	\section{Experimental Evaluation} \label{sec:experiments}
\begin{figure*}
	\begin{lstlisting}[language=SQL,
	deletekeywords={IDENTITY},
	deletekeywords={[2]INT},
	morekeywords={clustered},
	mathescape=true,
	xleftmargin=0pt,
	framexleftmargin=0pt,
	frame=tb,
	framerule=0pt ]
	
	$\mathsf{DBLP}_{2\mathsf{hop}}$ = SELECT DISTINCT $\mathsf{A}_1.\mathsf{name}, \mathsf{A}_2.\mathsf{name}$ FROM $\mathsf{Author}$ AS $\mathsf{A_1}$, $\mathsf{Author}$ AS $\mathsf{A_2}$, $\mathsf{AuthorPapers}$ AS $\mathsf{AP}_1$, $\mathsf{AuthorPapers}$ as $\mathsf{AP_2}$, $\mathsf{Paper}$ AS $\mathsf{P}$ WHERE $\mathsf{AP_1.pid} = \mathsf{AP_2.pid}$ AND $\mathsf{AP_1.aid} = \mathsf{A_1.aid}$ AND $\mathsf{AP_2.aid} = \mathsf{A_2.aid}$ AND $\mathsf{P.is\_research} = $true  ORDER BY $\mathsf{A}_1.\mathsf{weight} + \mathsf{A}_2.\mathsf{weight}$ LIMIT k;
	
	$\mathsf{DBLP}_{3\mathsf{hop}}$ = SELECT DISTINCT $\mathsf{A}.\mathsf{name}, \mathsf{P}.\mathsf{name}$ FROM $\mathsf{Author}$ AS $\mathsf{A}$, $\mathsf{Paper}$ AS $\mathsf{P}$, $\mathsf{AuthorPapers}$ AS $\mathsf{AP}_1$, $\mathsf{AuthorPapers}$ as $\mathsf{AP_2}$, $\mathsf{AuthorPapers}$ as $\mathsf{AP_3}$ WHERE $\mathsf{AP_1.pid} = \mathsf{AP_2.pid}$ AND $\mathsf{AP_2.aid} = \mathsf{AP_3.aid}$ AND $\mathsf{AP_1.aid} = \mathsf{A.aid}$ AND $\mathsf{AP_3.pid} = \mathsf{P.pid}$ AND $\mathsf{P.is\_research} = $true  ORDER BY $\mathsf{A}.\mathsf{weight} + \mathsf{P}.\mathsf{weight}$ LIMIT k;
	
	$\mathsf{DBLP}_{4\mathsf{hop}}$ = SELECT DISTINCT $\mathsf{A}_1.\mathsf{name}, \mathsf{A}_2.\mathsf{name}$ FROM $\mathsf{Author}$ AS $\mathsf{A_1}$, $\mathsf{Author}$ AS $\mathsf{A_2}$, $\mathsf{AuthorPapers}$ AS $\mathsf{AP}_1$, $\mathsf{AuthorPapers}$ as $\mathsf{AP_2}$, $\mathsf{AuthorPapers}$ as $\mathsf{AP_3}$, $\mathsf{AuthorPapers}$ as $\mathsf{AP_4}$, $\mathsf{Paper}$ AS $\mathsf{P_1}$, $\mathsf{Paper}$ AS $\mathsf{P_2}$ WHERE $\mathsf{AP_1.pid} = \mathsf{AP_2.pid}$ AND $\mathsf{AP_2.aid} = \mathsf{AP_3.aid}$ AND $\mathsf{AP_3.pid} = \mathsf{AP_4.pid}$ AND $\mathsf{AP_3.pid} = \mathsf{P_2.pid}$ AND $\mathsf{AP_1.pid} = \mathsf{P_1.pid}$ AND $\mathsf{AP_1.aid} = \mathsf{A_1.aid}$ AND $\mathsf{AP_4.aid} = \mathsf{A_2.aid}$ AND $\mathsf{P_1.is\_research} = $true AND $\mathsf{P_2.is\_research} = $true  ORDER BY $\mathsf{A}_1.\mathsf{weight} + \mathsf{A}2.\mathsf{weight}$ LIMIT k;
	
	$\mathsf{DBLP}_{3\mathsf{star}}$ = SELECT DISTINCT $\mathsf{A}_1.\mathsf{name}, \mathsf{A}_2.\mathsf{name}, \mathsf{A}_3.\mathsf{name}$ FROM $\mathsf{Author}$ AS $\mathsf{A_1}$, $\mathsf{Author}$ AS $\mathsf{A_2}$, $\mathsf{Author}$ AS $\mathsf{A_3}$, $\mathsf{AuthorPapers}$ AS $\mathsf{AP}_1$, $\mathsf{AuthorPapers}$ as $\mathsf{AP_2}$, $\mathsf{AuthorPapers}$ as $\mathsf{AP_3}$, $\mathsf{Paper}$ AS $\mathsf{P}$  WHERE $\mathsf{AP_1.pid} = \mathsf{AP_2.pid} = \mathsf{AP_3.pid}$ AND $\mathsf{AP_1.aid} = \mathsf{A_1.aid}$ AND $\mathsf{AP_2.aid} = \mathsf{A_2.aid}$ AND $\mathsf{AP_3.aid} = \mathsf{A_3.aid}$ AND $\mathsf{AP_3.pid} = \mathsf{P.pid}$ AND $\mathsf{P.is\_research} = $true ORDER BY $\mathsf{A}_1.\mathsf{weight} + \mathsf{A}_2.\mathsf{weight} + \mathsf{A}_3.\mathsf{weight}$ LIMIT k;
	\end{lstlisting}
	\vspace*{-2em}
	\caption{Network analysis queries for DBLP. Queries for IMDB are defined similarly (see~\cite{fullversion}).} \label{fig:dblpqueries}
\end{figure*}

In this section, we perform an extensive evaluation of our proposed algorithm. Our goal is to evaluate three aspects: $(a)$ how fast our algorithm is compared to state-of-the-art implementations for both \SUM\ and \lexi\ ranking functions on various queries {and datasets}, $(b)$ test the empirical performance of the space-time tradeoff in~\autoref{thm:ranked:star}, $(c)$ investigate the performance of our algorithm on various cyclic queries based on different shapes {and $(d)$ test the scalability behavior of our algorithm.}

\subsection{Experimental Setup}

We use Neo4j $4.2.3$ community edition, MariaDB $10.1.47$\footnote{Compared with MySQL, MariaDB performed better in our experiments, hence we report the results for MariaDB} and PostgreSQL $11.12$ for our experiments. All experiments are performed on a Cloudlab machine~\cite{duplyakin2019design} running Ubuntu $18.04$ equipped with two Intel E5-2630 v3 8-core CPUs@$2.40$ GHz and $128$ GB RAM. We focus only on the main memory setting and all experiments run on a single core. Since the join queries are memory intensive, we take special care to ensure that only one DBMS engine is running at a time, restart the session for each query to ensure temp tables in main memory are flushed out to avoid any interference, and also monitor that no temp tables are created on the disk. We only keep one database containing a single relation when performing experiments. We switch off all logging to avoid any performance impact. For PostgreSQL and MariaDB, we allow the engines to use the full main memory to ensure all temp tables are resident in the RAM and sorting (if any) happens without any disk IOs by increasing the sort buffer limit. For Neo4j, we allow the JVM heap to use the full main memory at the time of start up. We also build bidirectional B-tree indexes for each relation ahead of time and create named indexes in Neo4j. All of our algorithms are implemented in C++ and compiled using the GNU C++ $7.5.0$ compiler that ships with Ubuntu $18.04$. Each experiment is run $5$ times and we report the median after removing the slowest and the fastest run.

\subsubsection{Small-Scale Datasets} We use two real world small scale datasets for our experiments: the DBLP dataset, containing relationship between authors and papers, and the IMDB dataset, containing relationship between actors, directors, and movies. {We use these datasets for two reasons: $(i)$ both datasets have been found to be useful and studied extensively in practical problems such as similarity search~\cite{yu2012user}, citation graph analysis~\cite{rahm2005citation}, and network analysis~\cite{elmacioglu2005six, biryukov2008co}. $(ii)$ small-scale datasets allow experiments to finish for all systems allowing us to make a fair comparison and develop a fine-grained understanding. In line with prior work~\cite{kargar2020effective}, for each tuple we assign the weight attribute (and add it to the table schema) in two ways: first, we assign a randomly chosen value, and second, logarithmic weights in which the weight of the entity (author and paper in \textsf{DBLP}) $v$ is $\log_2(1+deg_v)$, where $deg_v$ denotes its degree in the relation. The schema for both datasets is as shown below (underlined attributes are primary keys for the relation):

\begin{enumerate*}
	\item DBLP: AuthorPapers$(aid, pid)$, Author$(\underline{aid}, name, weight)$, 
	
	 Paper$(\underline{aid}, title, venue, year, weight, is\_research)$.
	\item IMDB: PersonMovie$(pid, mid)$, Company$(\underline{cid}, name, nation)$,
	
	Person$(\underline{pid}, name, role, weight)$, 
	 
	 Movie$(\underline{mid}, name, year, genre, cid, weight)$

\end{enumerate*}

\introparagraph{Queries} We consider $4$ acyclic join queries as shown in~\autoref{fig:dblpqueries} for the small-scale datasets, which are  commonly seen in practice~\cite{sun2013mining, bonifati2020analytical}. Intuitively, the first three queries find all the top-k weighted 2-hops, 3-hops and 4-hops reachable attribute pairs within the DBLP network. As remarked by in~\cite{sun2013mining, chen2018inferring, kuccuktuncc2012recommendation}, these queries are of immense practical interest (e.g., see Table 4 in~\cite{sun2013mining}). Queries for IMDB dataset are defined similarly in~\autoref{sec:moreexp}}.  In Subsection~\ref{subsec:cyclic}, we also investigate the performance for cyclic queries. 

\subsubsection{Large-Scale Datasets} We also perform experiments on two real-world large scale relational datasets and one relational benchmark. The first dataset is from the Friendster~\cite{snapnets} online social network that contains $1.8B$ tuples. In the social network each, user is associated with multiple groups. The second dataset is the Memetracker~\cite{snapnets} dataset which describes user generated memes and which users have interacted with the meme. The dataset contains $418M$ tuples. For both Friendster and Memetracker, we use weights for users as the number of groups they belong to and the number of memes they create respectively. Finally, we also use the queries containing a ranking function from the LDBC Social Network Benchmark~\cite{erling2015ldbc} with scale factor $\mathsf{SF=10}$, a publicly available benchmark, to perform scalability experiments.

\introparagraph{Queries} For Friendster and Memetracker, we use two popular queries that are used in network analysis. Similar to the DBLP queries, we identify the ranked user pairs in the two hop and three hop neighborhoods for all users. The ranking is the sum of weights of the user pair. These queries have widespread application in understanding information flow in a network~\cite{myers2014information} and are used in recommendation systems~\cite{feng2019attention, li2019supervised}.  For LDBC benchmark, we use the multi-source version of \textbf{Q3, Q10} and \textbf{Q11}. Each of these queries are variants of the neighborhood analysis and contains \sqlhighlight{UNION}.

\subsection{Small Scale Experiments} \label{thm1:exp}

In this section, we compare the empirical performance of the algorithm given by~\autoref{thm:general} (labeled as \textsc{LinDelay} in all figures) against the baselines for each query. In order to perform a fair comparison, we materialize the top-$k$ answers in-memory since other engines also do it. However, a strength of our system is that if a downstream task only requires the output as a stream, we are able to enumerate the result instead of materializing it, which is not possible with other engines. 

\introparagraph{Sum ordering}~\autoref{fig:linear:notradeoff} shows the main results for the DBLP and IMDB datasets when the ranking function is the sum function and the weights are chosen randomly. Let us first review the results for the DBLP dataset. \autoref{fig:dblp:2path} shows the running time for different values of $k$ in the limit clause. The first observation is that all engines materialize the join result, followed by deduplicating and sorting according to the ranking function which leads to poor performance for all baselines. This is because all engines treat sorting and distinct clause as blocking operators, verified by examining the query plan. On the other hand, our approach is limit-aware. For small values of $k$, we are up to two orders of magnitude faster and as the value of $k$ increases, the total running time of our algorithm increases linearly. Even when our algorithm has to enumerate and materialize the entire result, it is still faster than asking the engines for the top-$10$ results. This is a direct benefit of generating the output in deduplicated and ranked order. As the path length increases from two to three and four path (\autoref{fig:dblp:3path} and \autoref{fig:dblp:4path}), the performance gap between existing engines and our approach also becomes larger.  We also point out that all engines require a large amount of main memory for query execution. For example, MariaDB requires about $40$GB of memory for executing $\mathsf{DBLP}_{4\mathsf{hop}}$. In contrast, the space overhead of our algorithm is dominated by the size of the priority queue. For DBLP dataset, our approach requires a measly $1.3$GB, $4$GB, $3$GB and $2.7$GB total space for $\mathsf{DBLP}_{2\mathsf{hop}}, \mathsf{DBLP}_{3\mathsf{hop}}, \mathsf{DBLP}_{4\mathsf{hop}}$ and $\mathsf{DBLP}_{3\mathsf{star}}$ respectively. For $\mathsf{DBLP}_{3\mathsf{hop}}, \mathsf{DBLP}_{4\mathsf{hop}}$ and $\mathsf{DBLP}_{3\mathsf{star}}$, we also implement breadth first search (BFS) followed by a sorting step using the idea of~\autoref{algo:lexicographic}. As it can be seen from the figures, BFS and sort provides an intermediate strategy which is faster than our algorithm for large values of $k$ but at the cost of expensive materialization of the entire result, which may not be always possible (and is the case for IMDB dataset). However, deciding to use BFS and sort requires knowledge of the output result size, which is unknown apriori and difficult to estimate. For the IMDB dataset, we observe a similar trend of our algorithm displaying superior performance compared to all other baselines. In this case, BFS and sorting even for $\mathsf{DBLP}_{4\mathsf{hop}}$ is not possible since the result is almost $0.5$ trillion items. For $\mathsf{DBLP}_{3\mathsf{star}}$, none of the engines were able to compute the result after running for $5$ hours when main memory ran out. BFS and sort also failed due to the size being larger than the main memory limit. Lastly, Neo4j was consistently the best performing (albeit marginally) engine among all baselines. While there is little scope for rewriting the SQL queries to try to obtain better performance, Neo4j has graph-specific operators such as variable length expansion. We tested multiple rewritings of the query to obtain the best performance (although this is the job of the query optimizer), which is finally reported in the figures. Regardless of the rewritings, Neo4j still treats materializing and sorting as a blocking operator which is a fundamental bottleneck.

\introparagraph{Lexicographic ordering} Figures~\ref{fig:dblp:2path:lex},\ref{fig:dblp:3path:lex},\ref{fig:dblp:4path:lex} and \ref{fig:dblp:3star:lex} show the running time for different values of $k$ in the limit clause for lexicographic ranking function on DBLP (i.e. we replace $\mathsf{A}_1.\mathsf{weight} + \mathsf{A}_2.\mathsf{weight}$ with $\mathsf{A}_1.\mathsf{weight}, \mathsf{A}_2.\mathsf{weight}$ in the \sqlhighlight{ORDER BY} clause) for random weights. The first striking observation here is that the running time for all baseline engines is identical to that of sum function. This demonstrates that existing engines are also agnostic to the ranking function in the query and fail to take advantage of the additional structure. However, lexicographic functions are easier to handle in practice than sum because we can avoid the use of a priority queue altogether. This in turn leads to faster running time since push and pops from the priority queue are expensive due to the logarithmic overhead and need for re-balancing of the tree structure. Thus, we obtain a $2\times$ improvement for lexicographic ordering as compared to the sum function.

{\introparagraph{Join ordering} At this point, the reader may wonder what is the impact of different join orderings on the query execution time for DBMS engines in the presence of \sqlhighlight{ORDER BY}. To investigate this, we supply join order hints to each of the engines. We run the queries on all possible join order hints to find the best possible running time. We found that the join order hints had virtually no impact on execution time. For instance, $\mathsf{DBLP}_{4\mathsf{hop}}$ on Neo4J takes $5521.61s$ without any join hints and the best possible join ordering reduces the time to $5418.23s$, a mere $1.8\%$ reduction. This is not surprising since the bottleneck for all engines is the materialization of the unsorted output, which is orders of magnitude larger than the final output and ends up being the dominant cost. In fact, for queries containing only self-joins, join order hints do not have any impact on the query plan because all relations are identical. Further, the number of possible join orderings that may need to be explored is exponential in the number of relations. On the other hand, our algorithm has the advantage of bypassing the materialization due to the delay based problem formulation and use of multi-way joins.
	
\introparagraph{Logarithmic weights} Instead of choosing the weights randomly, we also investigate the behavior when the weights scale logarithmically w.r.t. to the degree. We observed that all systems as well as our algorithm had identical execution times. This is not surprising considering that no algorithm takes into account the actual distribution of the weights. This observation points to an additional opportunity for optimization where one could use the weight distribution to allow for fine-grained, data-dependent processing. We leave the study of this problem for future work.

}

\begin{figure*}[!htp]
	\begin{subfigure}{0.24\linewidth}
		\hspace*{-1em}
		\includegraphics[scale=0.17]{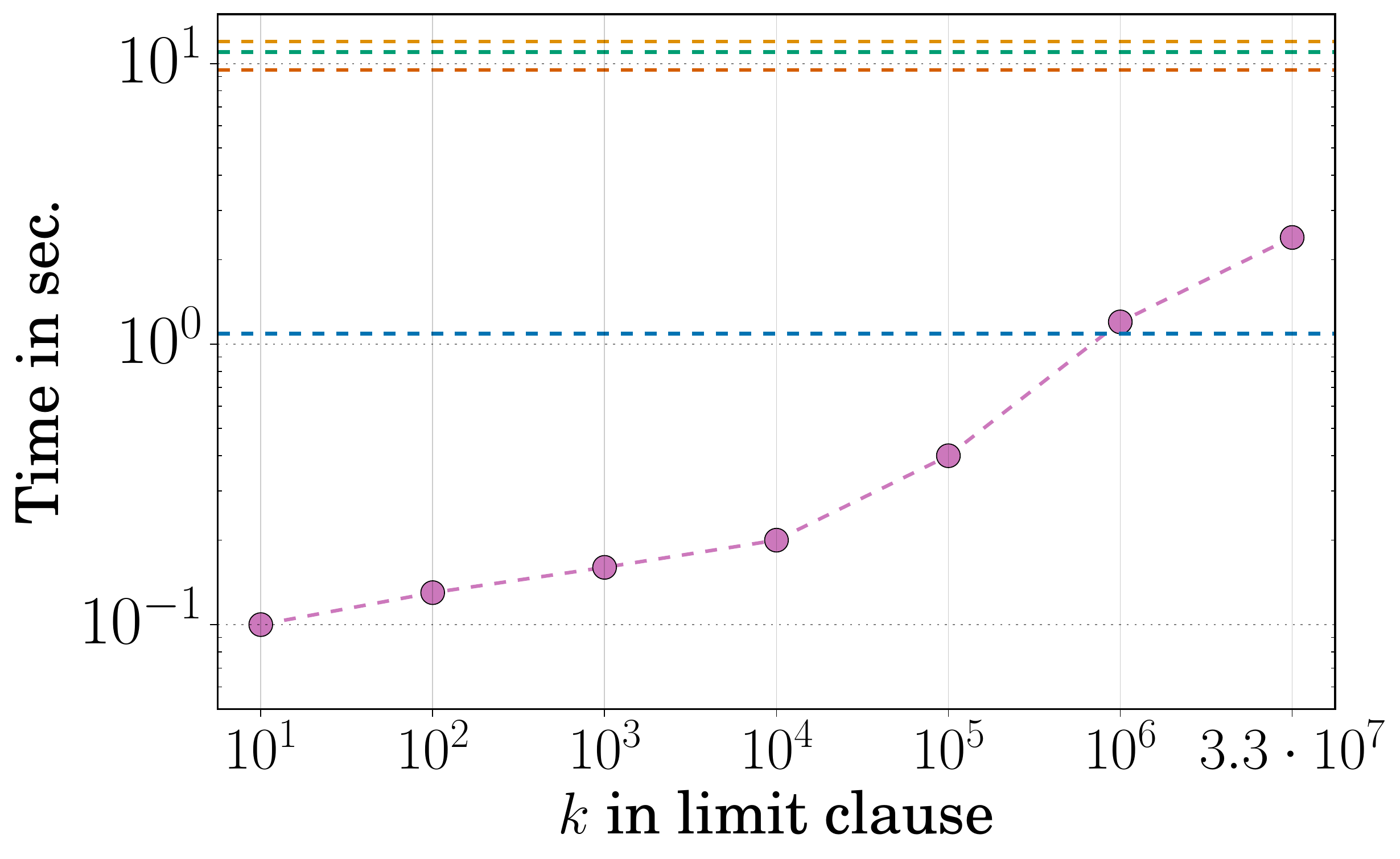}
		\caption{$\mathsf{DBLP}_{2\mathsf{hop}}$}  \label{fig:dblp:2path}
	\end{subfigure}
	\begin{subfigure}{0.24\linewidth}
		\vspace*{-1.2em}
		\hspace*{-1.25em}
		\includegraphics[scale=0.17]{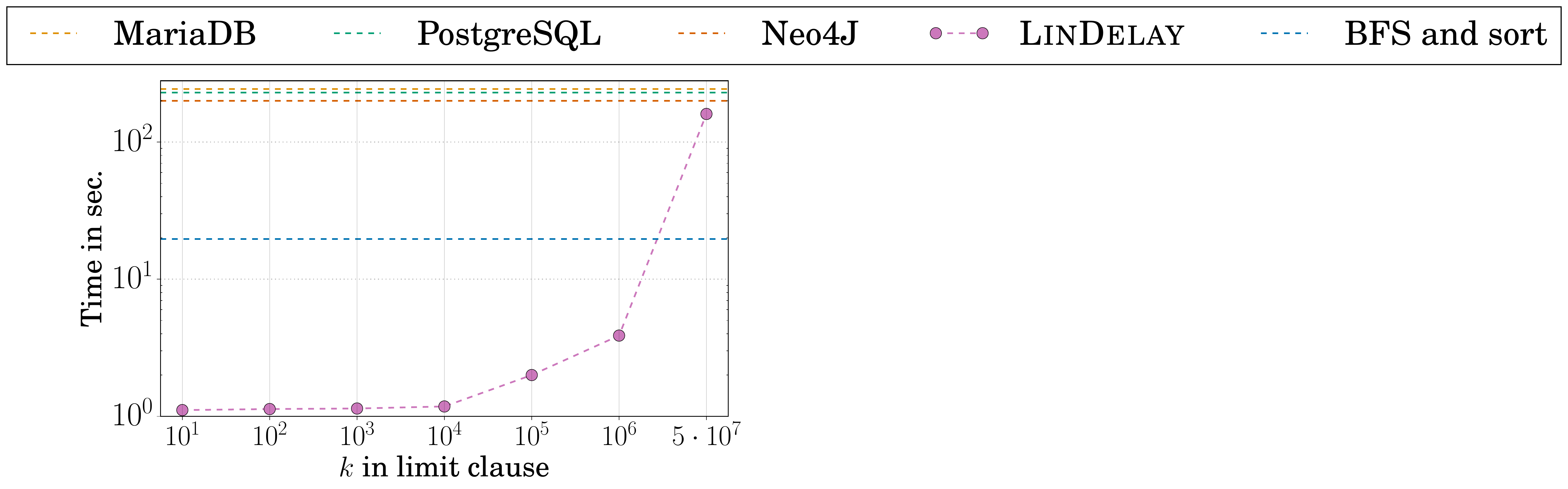}
		\caption{$\mathsf{DBLP}_{3\mathsf{hop}}$}  \label{fig:dblp:3path}
	\end{subfigure}
	\begin{subfigure}{0.24\linewidth}
		\includegraphics[scale=0.17]{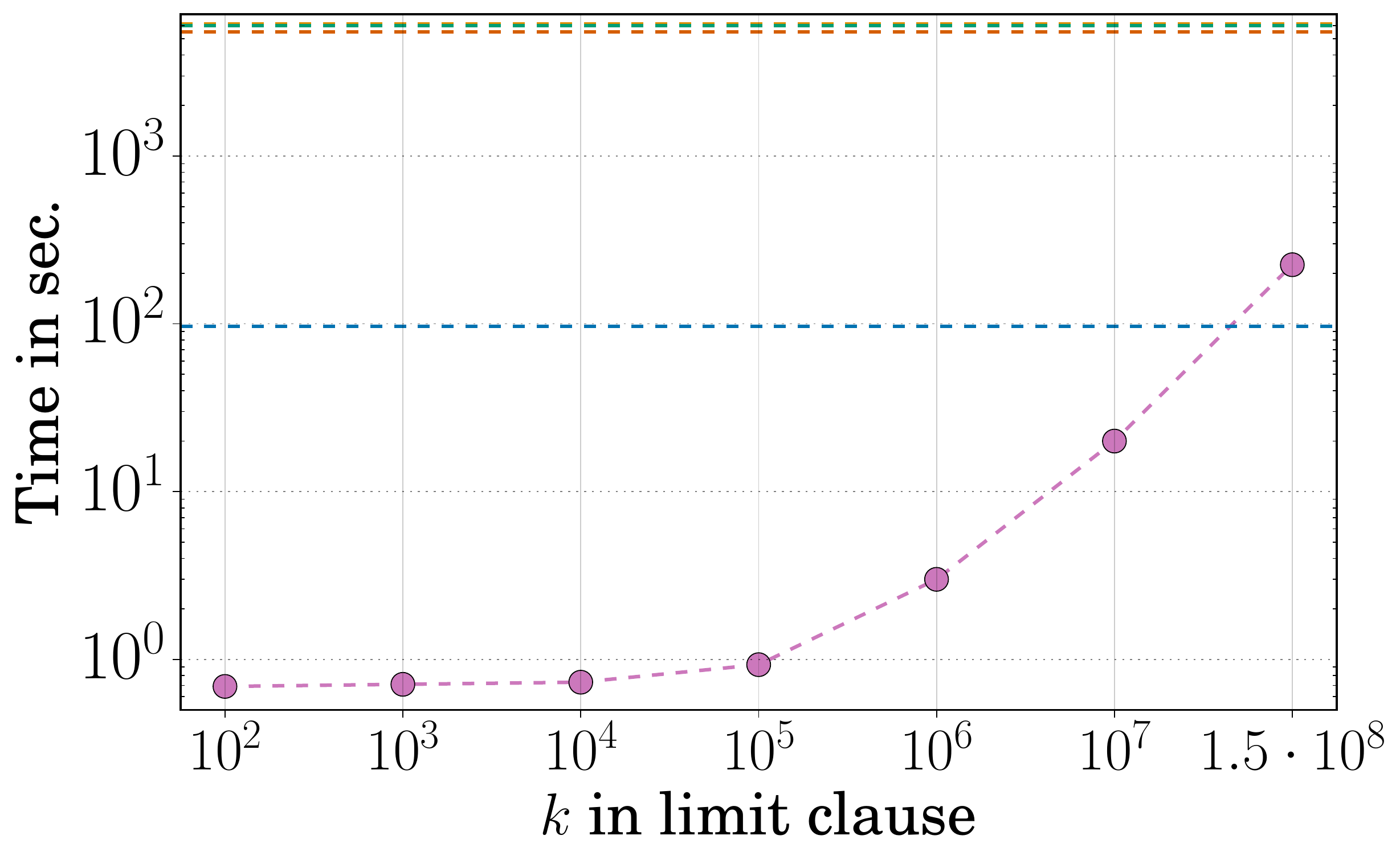}
		\caption{$\mathsf{DBLP}_{4\mathsf{hop}}$}  \label{fig:dblp:4path}
	\end{subfigure}
	\begin{subfigure}{0.24\linewidth}
		\includegraphics[scale=0.17]{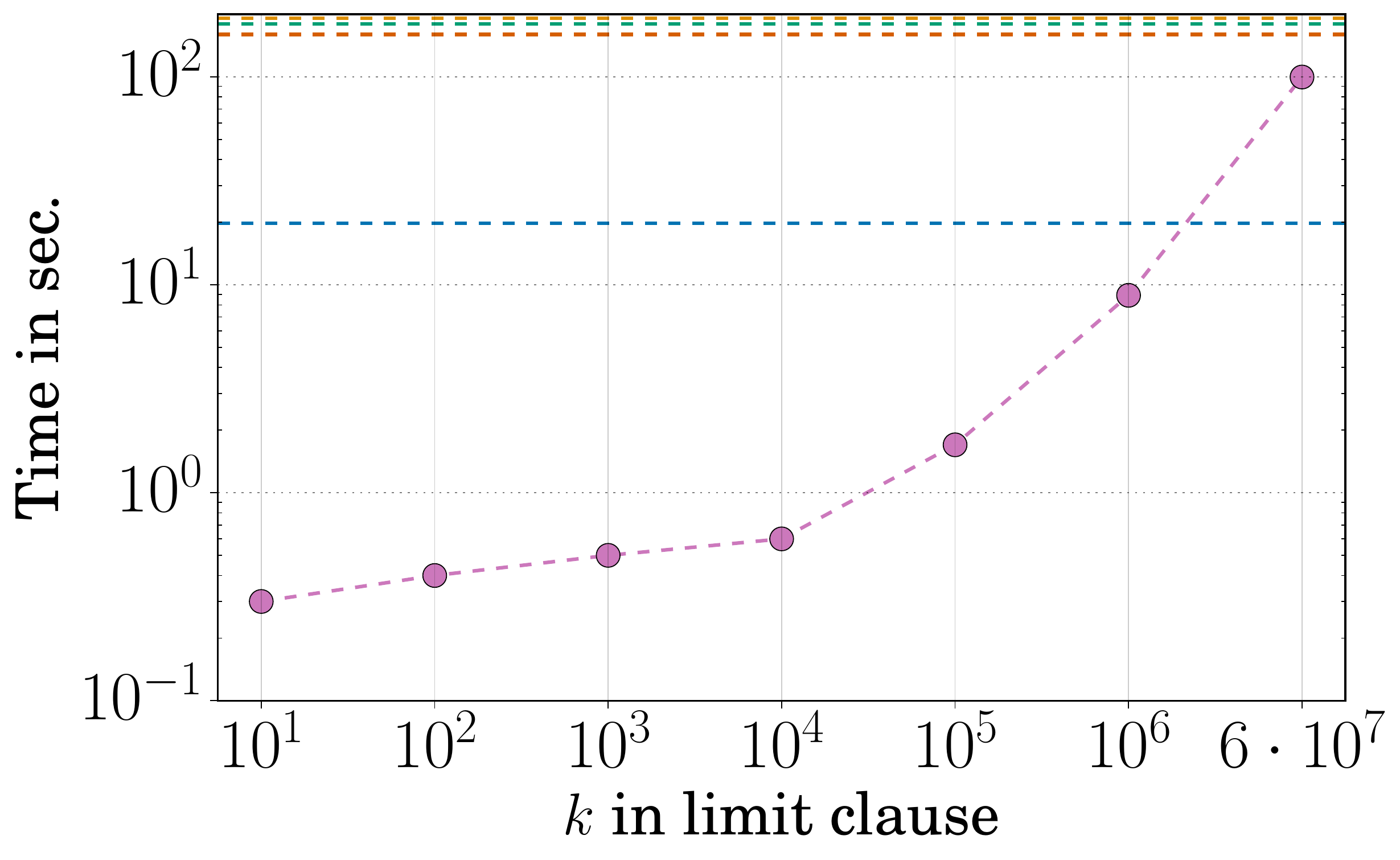}
		\caption{$\mathsf{DBLP}_{3\mathsf{star}}$}  \label{fig:dblp:3star}
	\end{subfigure}

		\begin{subfigure}{0.24\linewidth}
			\includegraphics[scale=0.17]{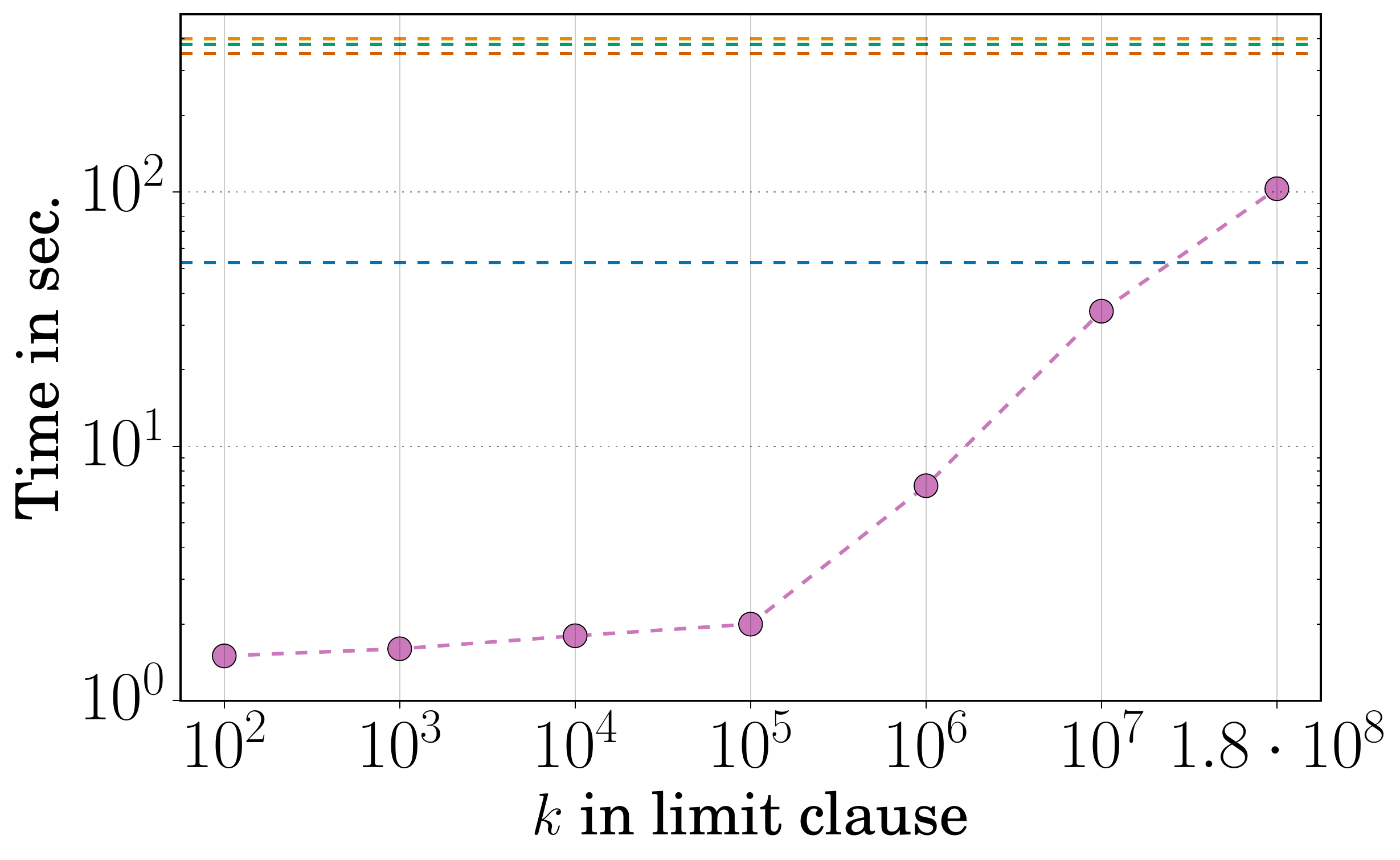}
			\caption{$\mathsf{IMDB}_{2\mathsf{hop}}$}  \label{fig:imdb:2path}
		\end{subfigure}
		\begin{subfigure}{0.24\linewidth}
			\includegraphics[scale=0.17]{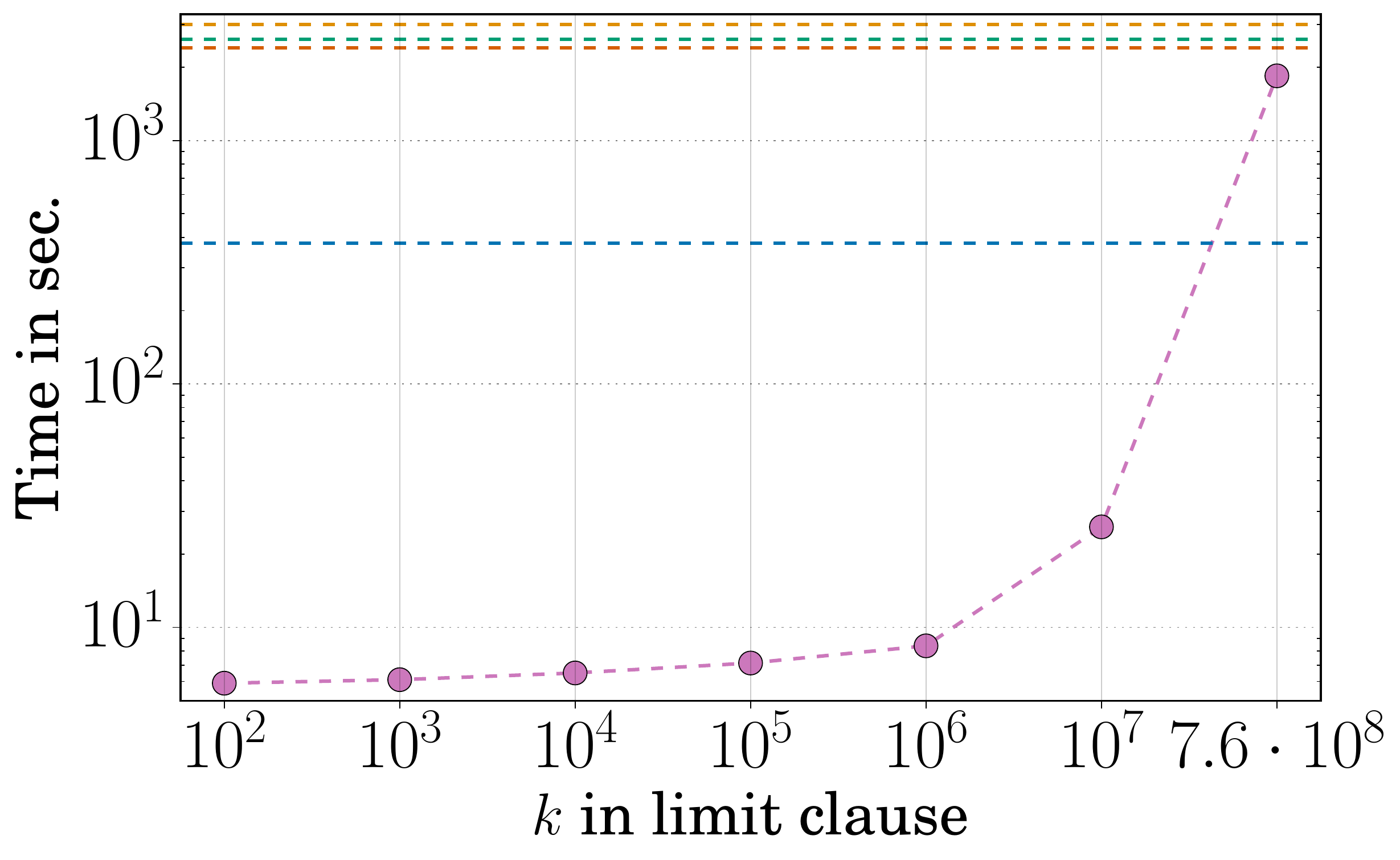}
			\caption{$\mathsf{IMDB}_{3\mathsf{hop}}$}  \label{fig:imdb:3path}
		\end{subfigure}
		\begin{subfigure}{0.24\linewidth}
			\includegraphics[scale=0.17]{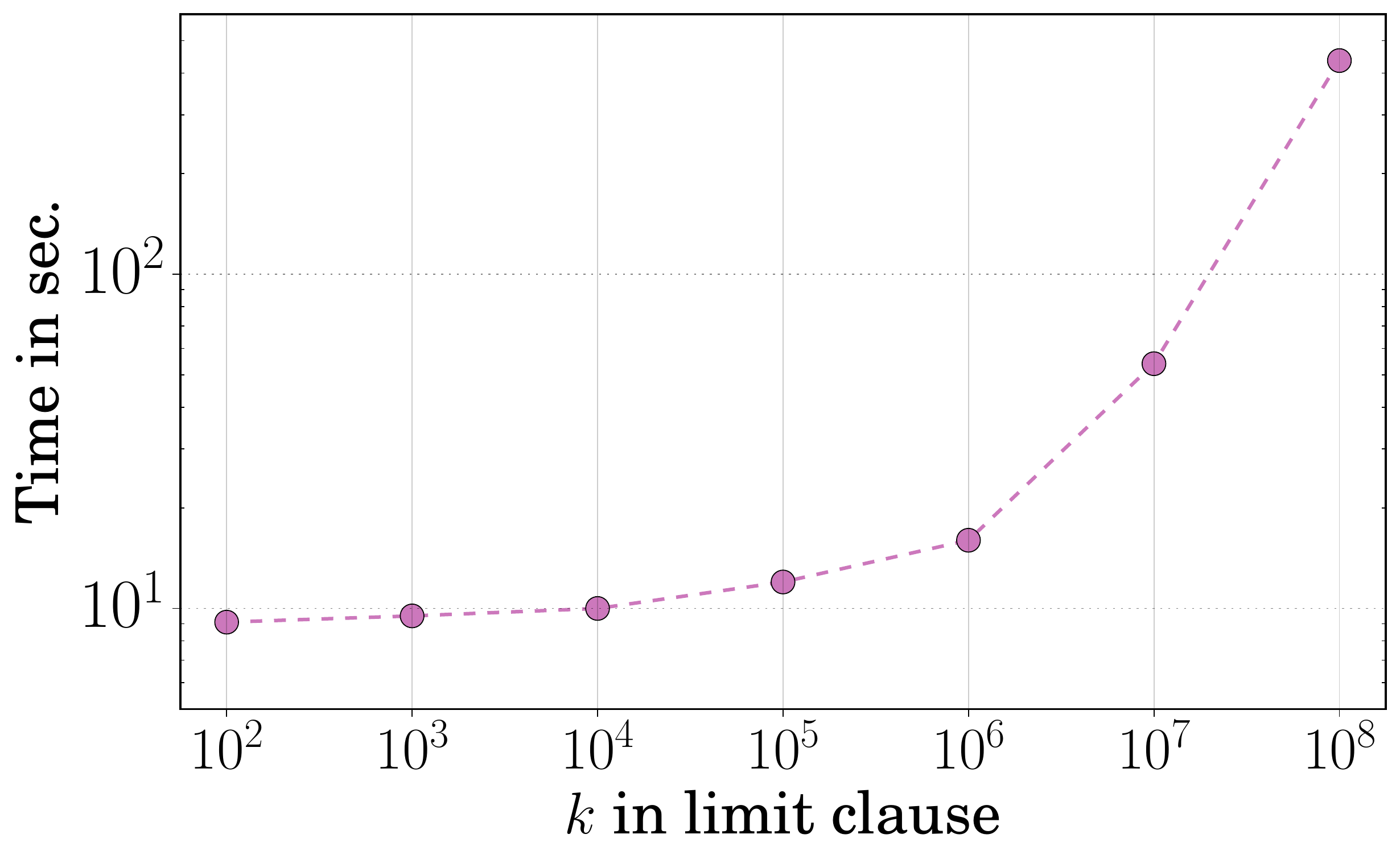}
			\caption{$\mathsf{IMDB}_{4\mathsf{hop}}$}  \label{fig:imdb:4path}
		\end{subfigure}
		\begin{subfigure}{0.24\linewidth}
			\includegraphics[scale=0.17]{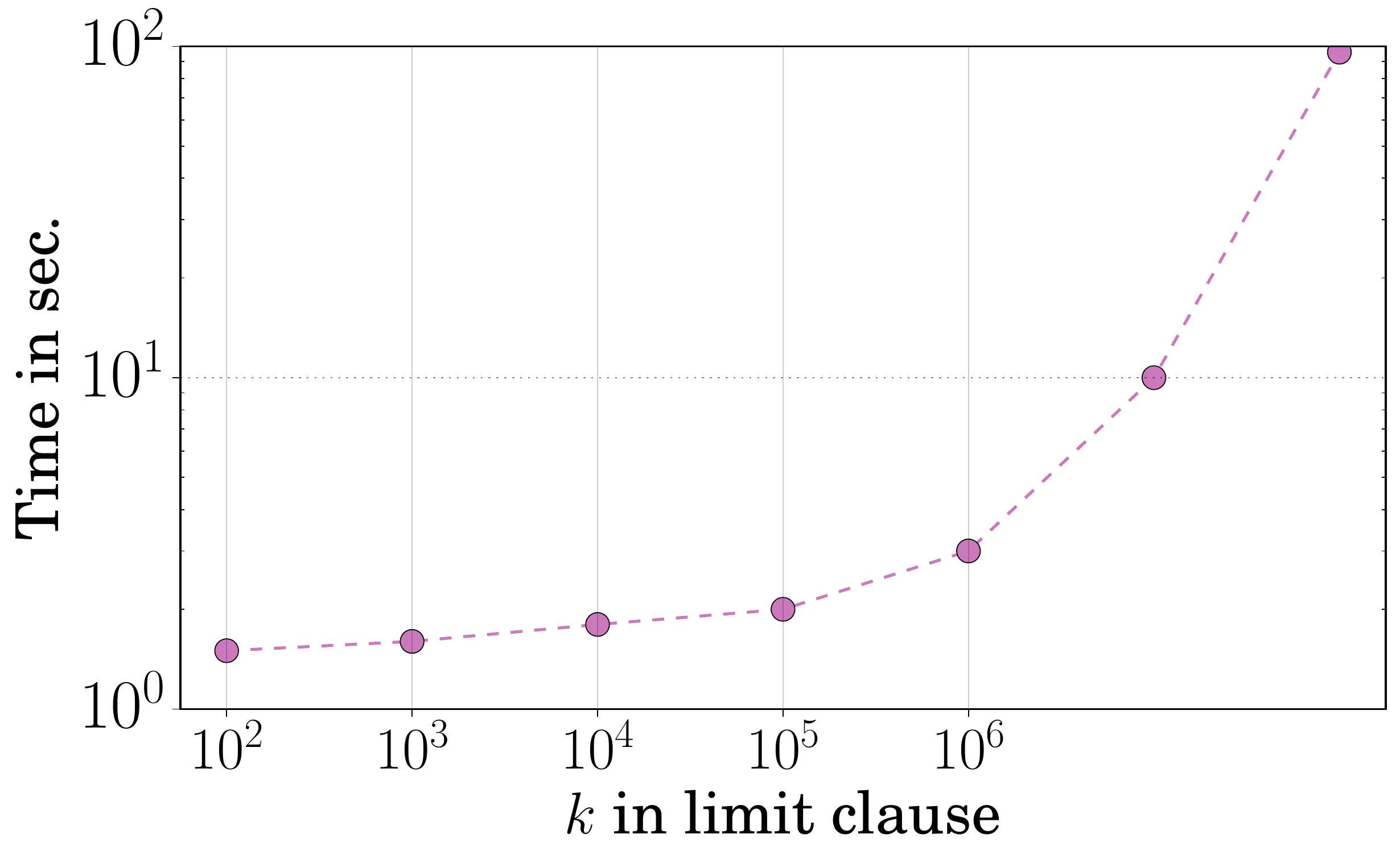}
			\caption{$\mathsf{IMDB}_{3\mathsf{star}}$}  \label{fig:imdb:3star}
		\end{subfigure}
	\caption{Comparing our algorithm with state-of-the-art engines for sum function} \label{fig:linear:notradeoff}
\end{figure*}

\begin{figure*}[!htp]
	\begin{subfigure}{0.24\linewidth}
		\includegraphics[scale=0.17]{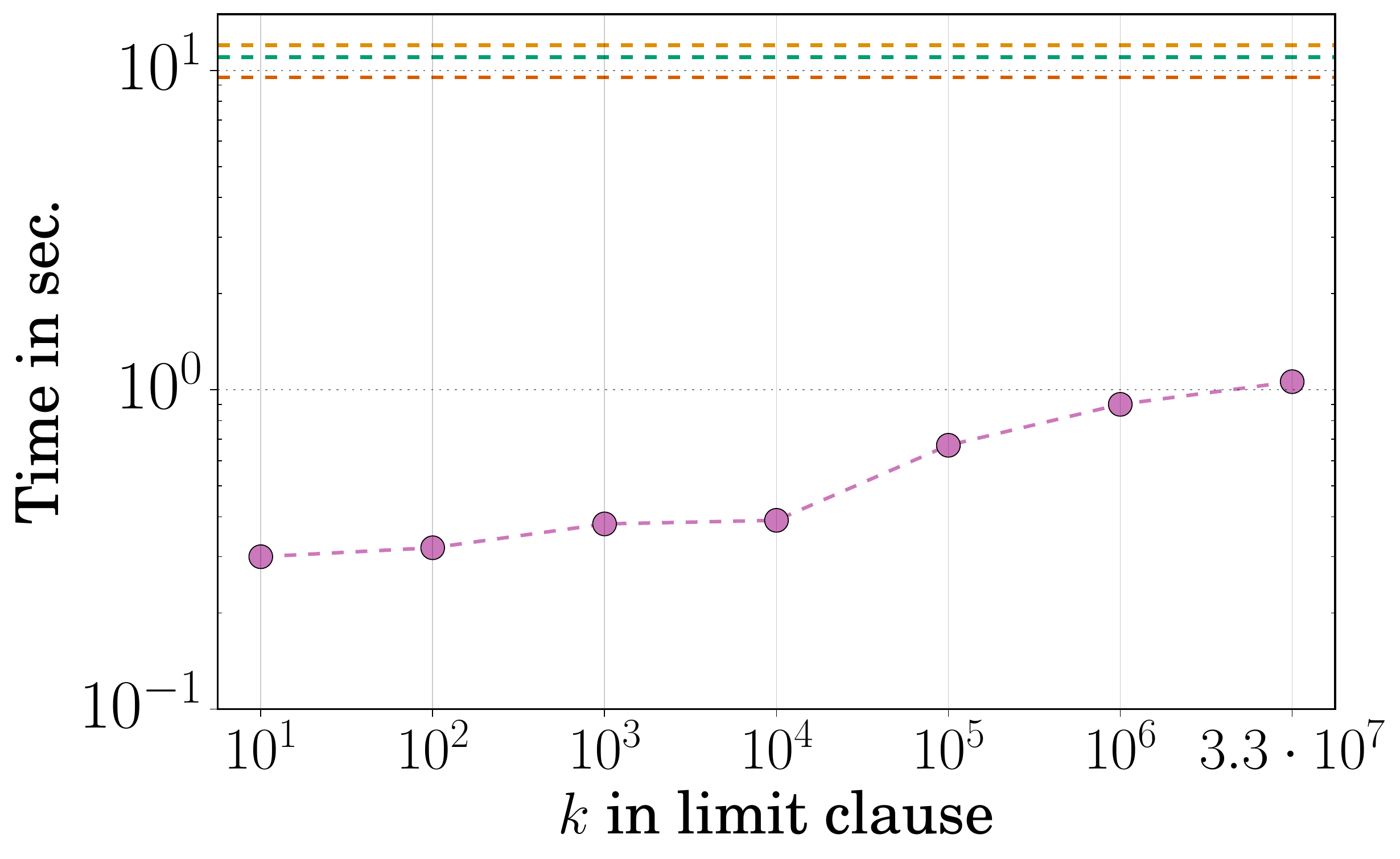}
		\caption{$\mathsf{DBLP}_{2\mathsf{hop}}$}  \label{fig:dblp:2path:lex}
	\end{subfigure}
	\begin{subfigure}{0.24\linewidth}
		\includegraphics[scale=0.17]{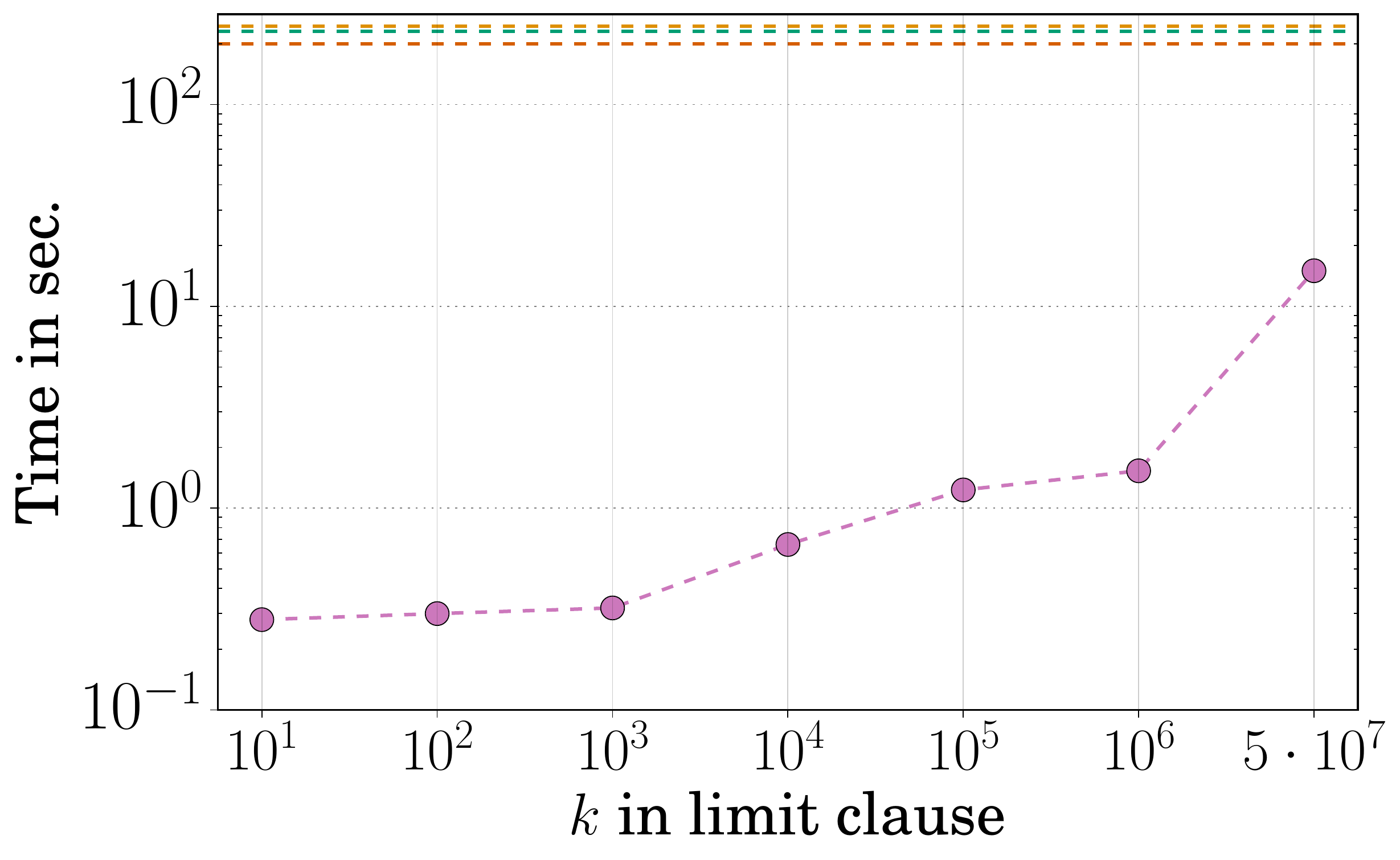}
		\caption{$\mathsf{DBLP}_{3\mathsf{hop}}$}  \label{fig:dblp:3path:lex}
	\end{subfigure}
	\begin{subfigure}{0.24\linewidth}
		\includegraphics[scale=0.17]{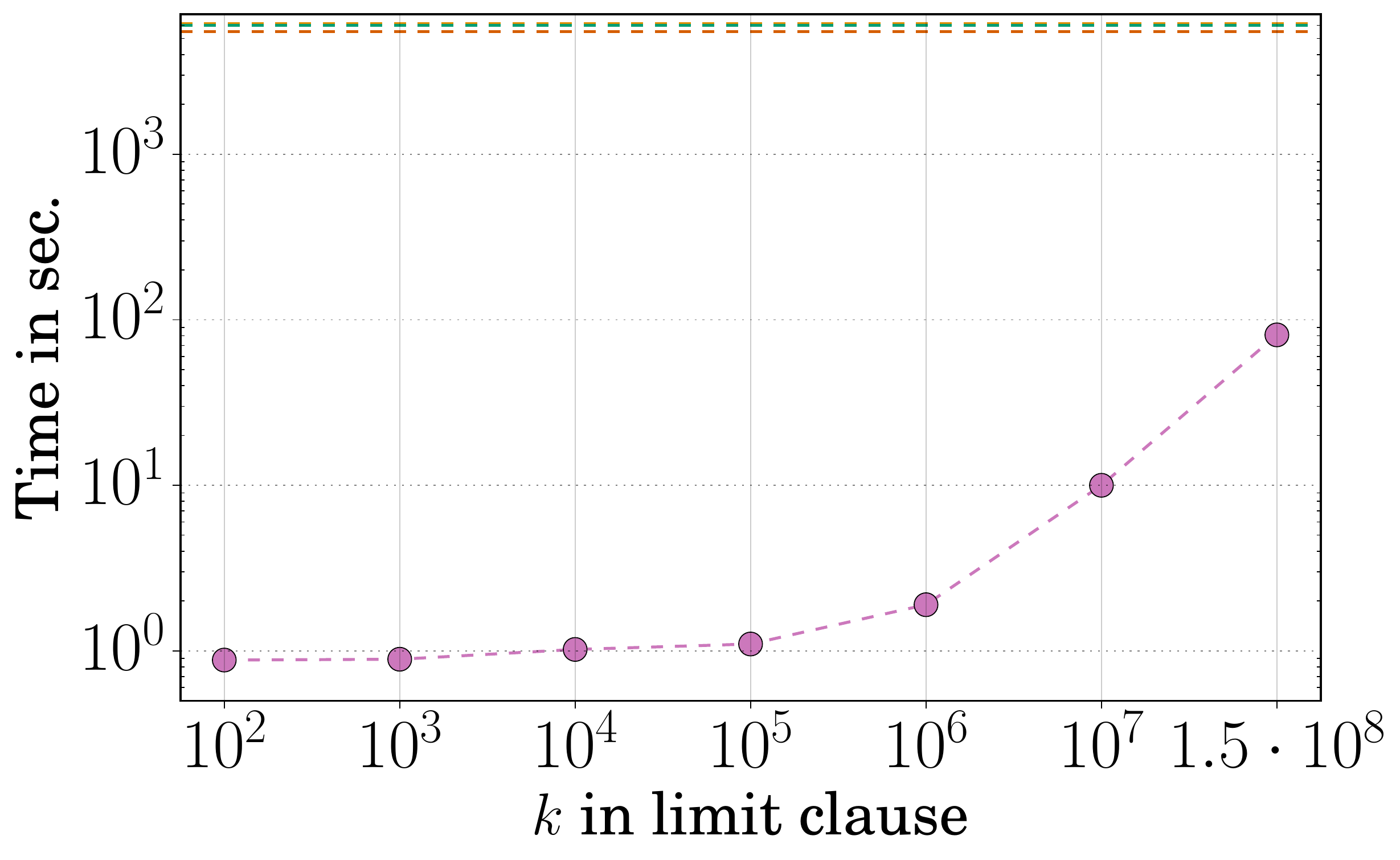}
		\caption{$\mathsf{DBLP}_{4\mathsf{hop}}$}  \label{fig:dblp:4path:lex}
	\end{subfigure}
	\begin{subfigure}{0.24\linewidth}
		\includegraphics[scale=0.17]{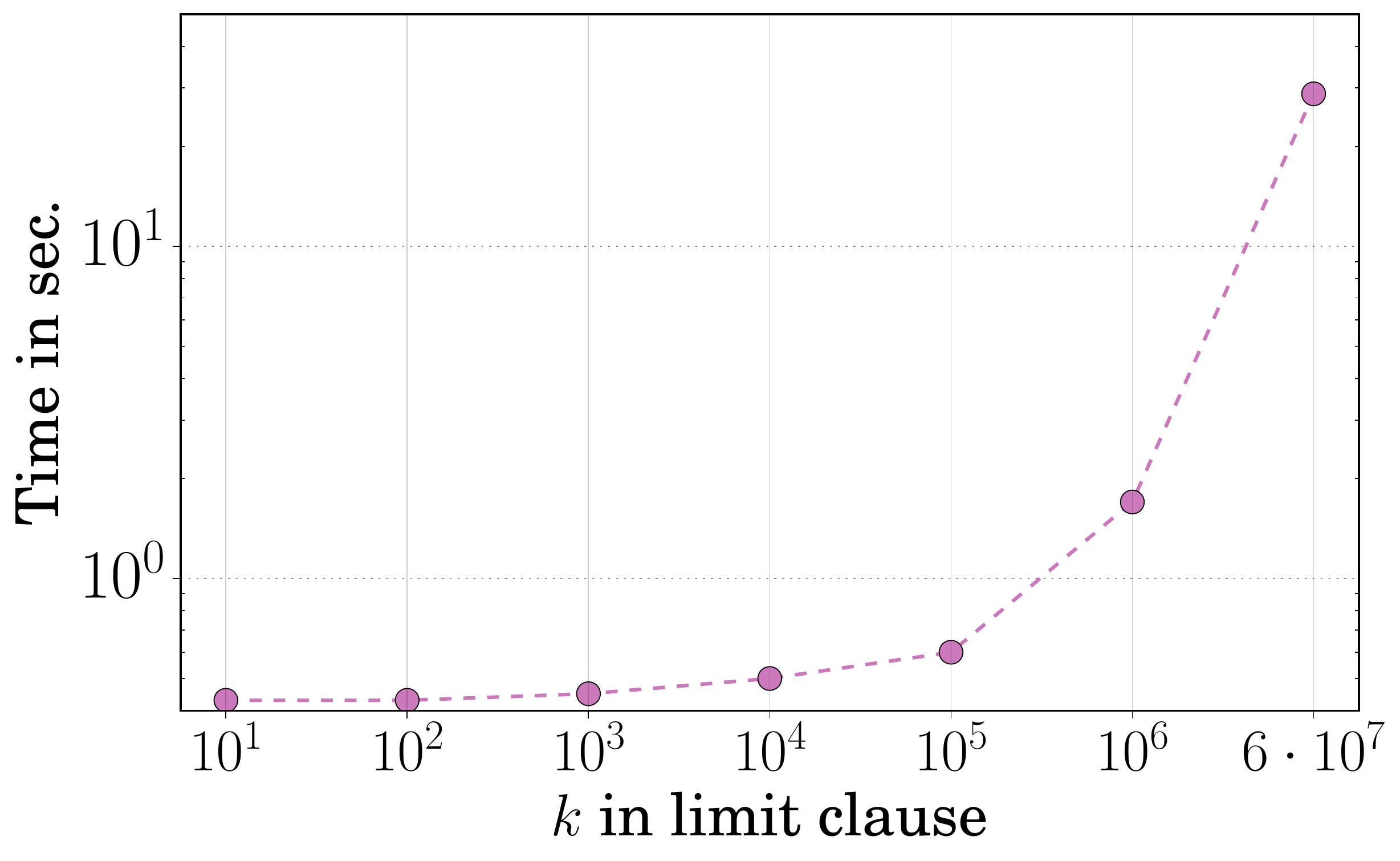}
		\caption{$\mathsf{DBLP}_{3\mathsf{star}}$}  \label{fig:dblp:3star:lex}
	\end{subfigure}
	\caption{Comparing our linear delay algorithm with state-of-the-art engines for lexicographic function.} \label{fig:linear:notradeoff:lex}
\end{figure*}



\subsubsection{Enumeration with Preprocessing} \label{sec:preprocess}

We next investigate the empirical performance of the preprocessing step and its impact on the result enumeration as described by~\autoref{thm:ranked:star}. For all experiments in this section, we fix $k$ to be large enough to enumerate the entire result (which is equivalent to having no limit clause at all).

\introparagraph{Sum ordering} \autoref{fig:dblp:2path:tradeoff} and \autoref{fig:imdb:2path:tradeoff} show the tradeoff between space used by the data structure constructed in the preprocessing phase and running time of the enumeration algorithm for $\mathsf{DBLP}_{2\mathsf{hop}}$ and $\mathsf{IMDB}_{2\mathsf{hop}}$ repectively. We show the tradeoff for $6$ different space budgets but the user is free to choose any space budget in the entire spectrum. As expected, the time required to enumerate the result is large when there is no preprocessing and it gradually drops as more and more results are materialized in the preprocessing phase. The sum of preprocessing time and enumeration time is not a flat line: this is because as an optimization, we do not use priority queues in the preprocessing phase. Instead, we can simply use the BFS and sort algorithm for all chosen nodes which need to be materialized. This is a faster approach in practice as we avoid use of priority queues but priority queues cannot be avoided for enumeration phase. We observe similar trend for $\mathsf{DBLP}_{3\mathsf{star}}, \mathsf{IMDB}_{3\mathsf{star}}$ on both datasets as well. 

\begin{figure*}
	\vspace{2em}
	\begin{subfigure}{0.24\linewidth}
		\includegraphics[scale=0.17]{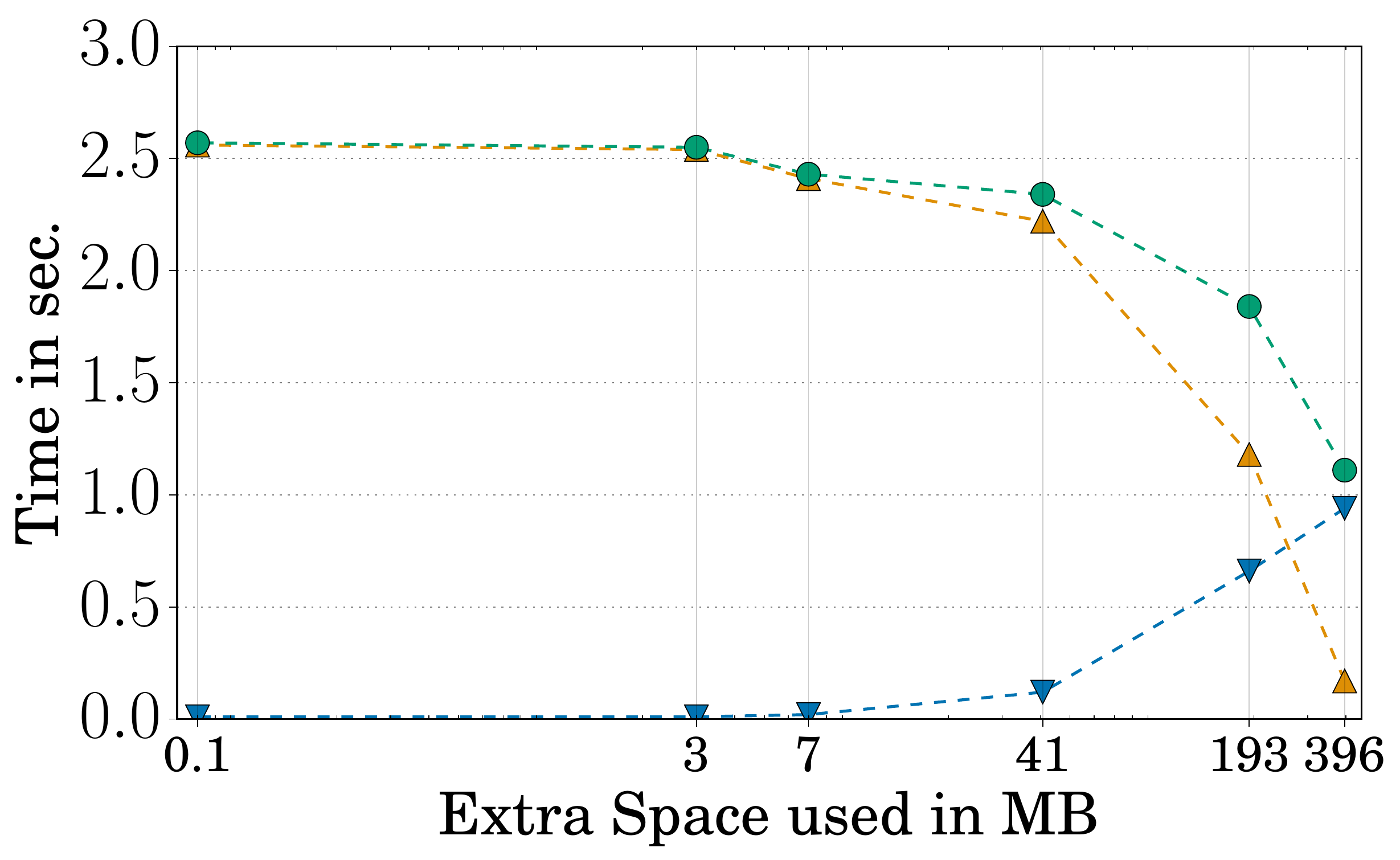}
		\caption{$\mathsf{DBLP}_{2\mathsf{hop}}$}  \label{fig:dblp:2path:tradeoff}
	\end{subfigure}
	\begin{subfigure}{0.24\linewidth}
		\vspace*{-1.4em}
		\includegraphics[scale=0.17]{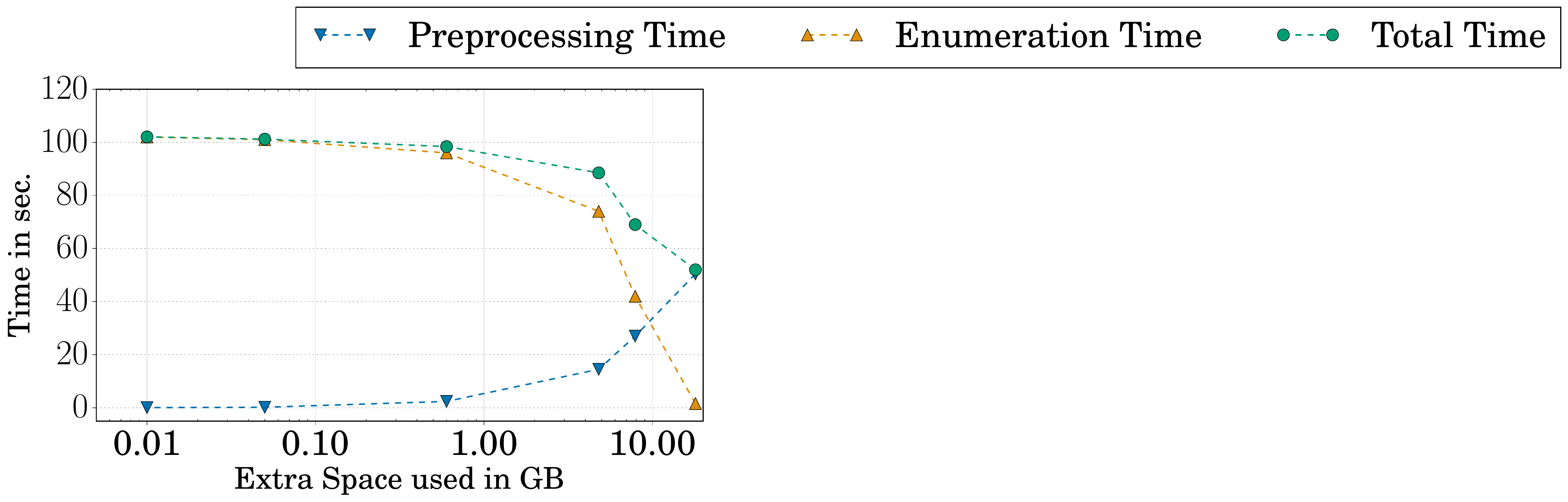}
		\caption{$\mathsf{IMDB}_{2\mathsf{hop}}$}  \label{fig:imdb:2path:tradeoff}
	\end{subfigure}
	\begin{subfigure}{0.24\linewidth}
		\includegraphics[scale=0.17]{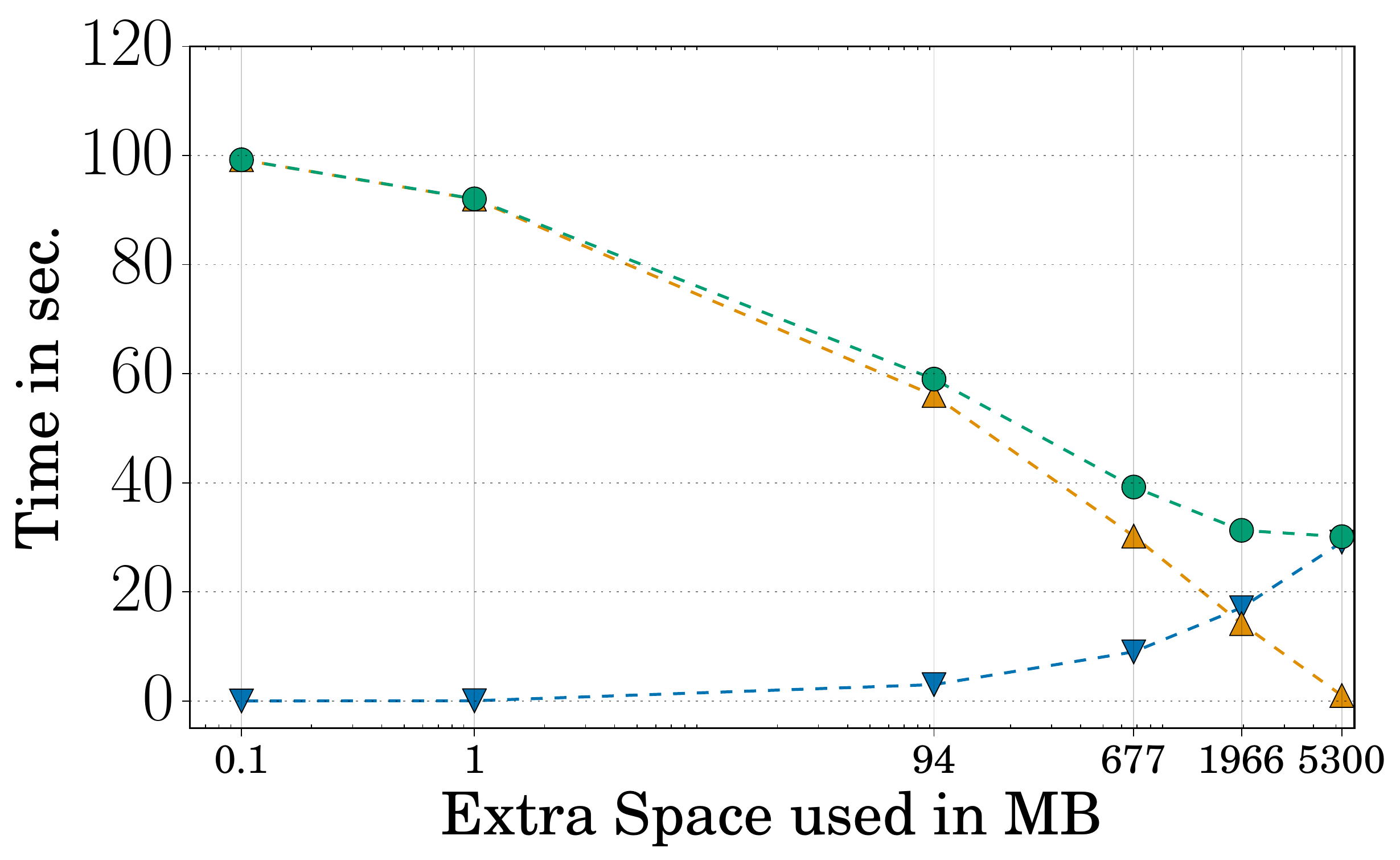}
		\caption{$\mathsf{DBLP}_{3\mathsf{star}}$}  \label{fig:dblp:3star:tradeoff}
	\end{subfigure}
	\begin{subfigure}{0.24\linewidth}
		\includegraphics[scale=0.17]{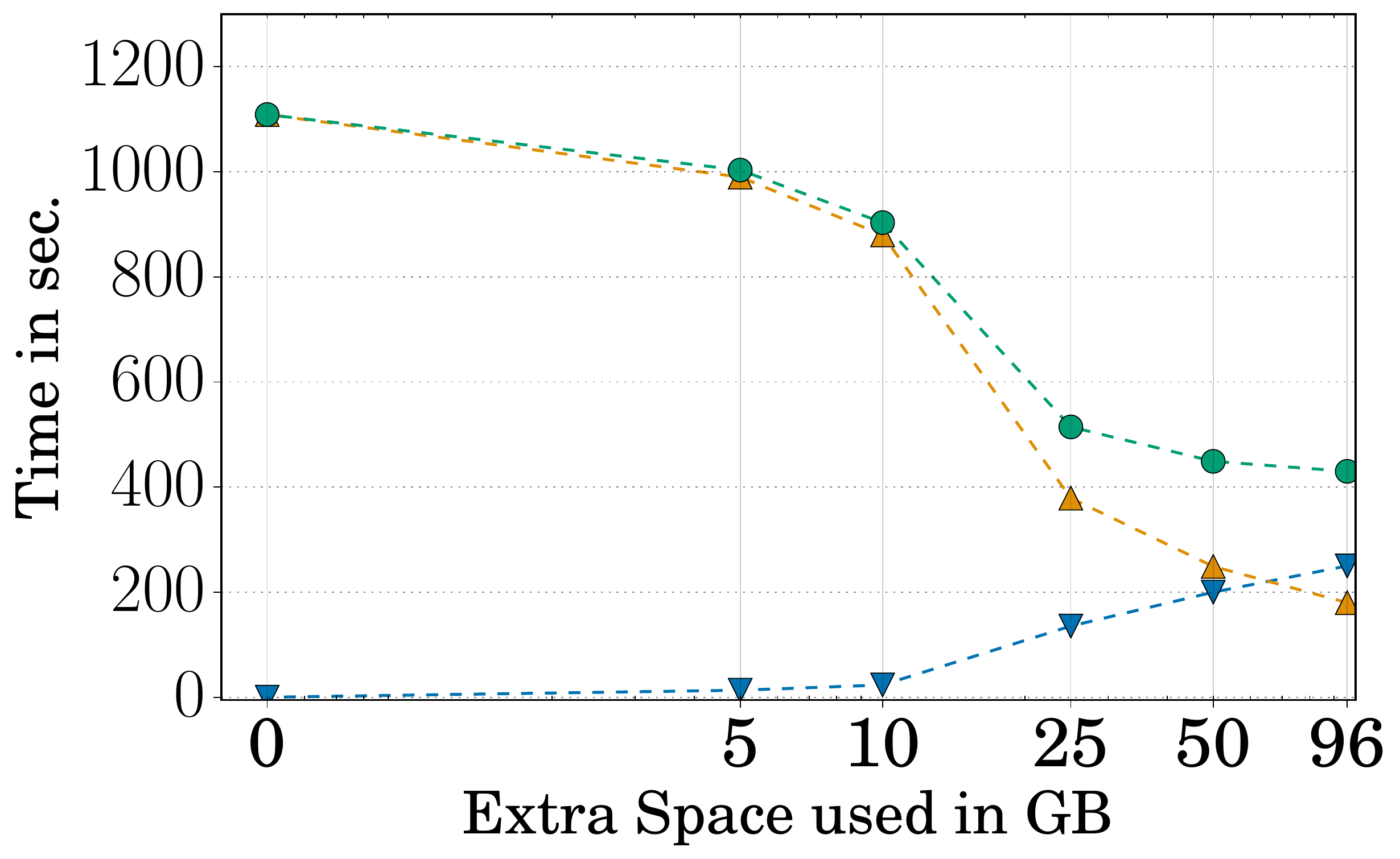}
		\caption{$\mathsf{IMDB}_{3\mathsf{star}}$}  \label{fig:imdb:3star:tradeoff}
	\end{subfigure}
	\caption{Comparing the preprocessing and enumeration tradeoff for sum function when enumerating the entire result} \label{fig:linear:tradeoff}
\end{figure*}

{
\subsubsection{Cyclic Queries} \label{subsec:cyclic}

We also compare the performance of our algorithm to other systems for cyclic queries. We choose four cyclic queries found commonly in practice inspired by~\cite{tziavelis2020optimal}: four cycle, six cycle, eight cycle and bowtie query (two four cycles joined at a common attribute).~\autoref{table:cyclic} shows the performance of our algorithm on the \textsf{DBLP} dataset for the sum function. As the table shows, our algorithm is able to process all queries within $200$ seconds, with the bowtie query being the most computationally intensive. In contrast, for $k=10$ the fastest performing engine Neo4J required 240s (450s) for four cycle (six cycle).  It did not finish execution for eight cycle and bowtie query due to an out of memory error. For the \textsf{IMDB} dataset, our algorithm was able to process all queries, while Neo4J was not able to process any query (except four cycle) due to its large memory requirement. We defer those experiments to~\autoref{sec:moreexp}.}

\subsection{Large Scale Experiments and Scalability}

In this section, we investigate the performance of our techniques on the large scale datasets. \autoref{fig:friend:2path} and~\ref{fig:friend:3path} shows the time to find the top-$k$ answers for the Memetracker dataset on two neighborhood and three neighborhood queries. Compared to the small scale datasets, the execution time increases rapidly even for low values of $k$. This is attributed to the high duplication of answers, which leads to a rapidly increasing priority queue size. \textit{None of MariaDB, Postgres and Neo4J were able to finish, or even to find the top-$10$ answers, within $5$ hours in our experiments}. The same trend is also observed for the Friendster dataset as shown in~\autoref{fig:meme:3path} and~\ref{fig:meme:4path}. Similar to the small-scale datasets, lexicographic functions were faster than the sum function for our algorithm but DBMS engines were unable to finish query execution. We also conduct scalability experiments on LDBC benchmark queries that contain the \sqlhighlight{ORDER BY} clause.~\autoref{table:ldbc} shows the scalability of our algorithm for finding answers of queries \textbf{Q3, Q10, Q11}. As the scale factor increases, the execution time also increases linearly. For each of these queries, all engines require more than $3$ hours to compute the result even for \textsf{SF }$=10$ and $k=10$. This is because of the serial execution plan generated by the engines, forcing the materialization of the unsorted result before sorting and filtering for top-$k$.

\begin{figure*}
	\begin{subfigure}{0.24\linewidth}
		\hspace*{-1em}
		\includegraphics[scale=0.17]{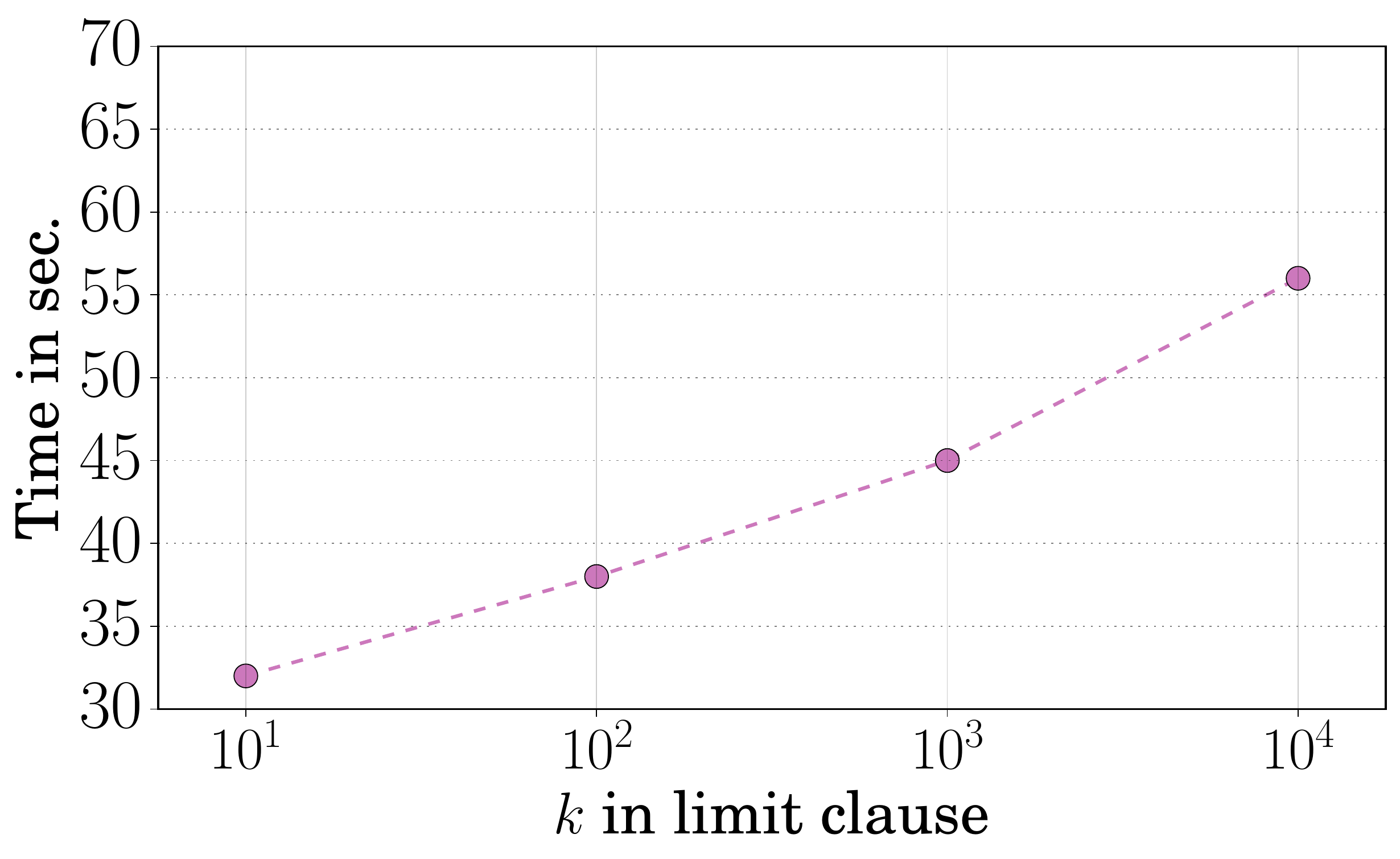}
		\caption{Memetracker 2-neighborhood}  \label{fig:friend:2path}
	\end{subfigure}
	\begin{subfigure}{0.24\linewidth}
		\hspace*{-1.25em}
		\includegraphics[scale=0.17]{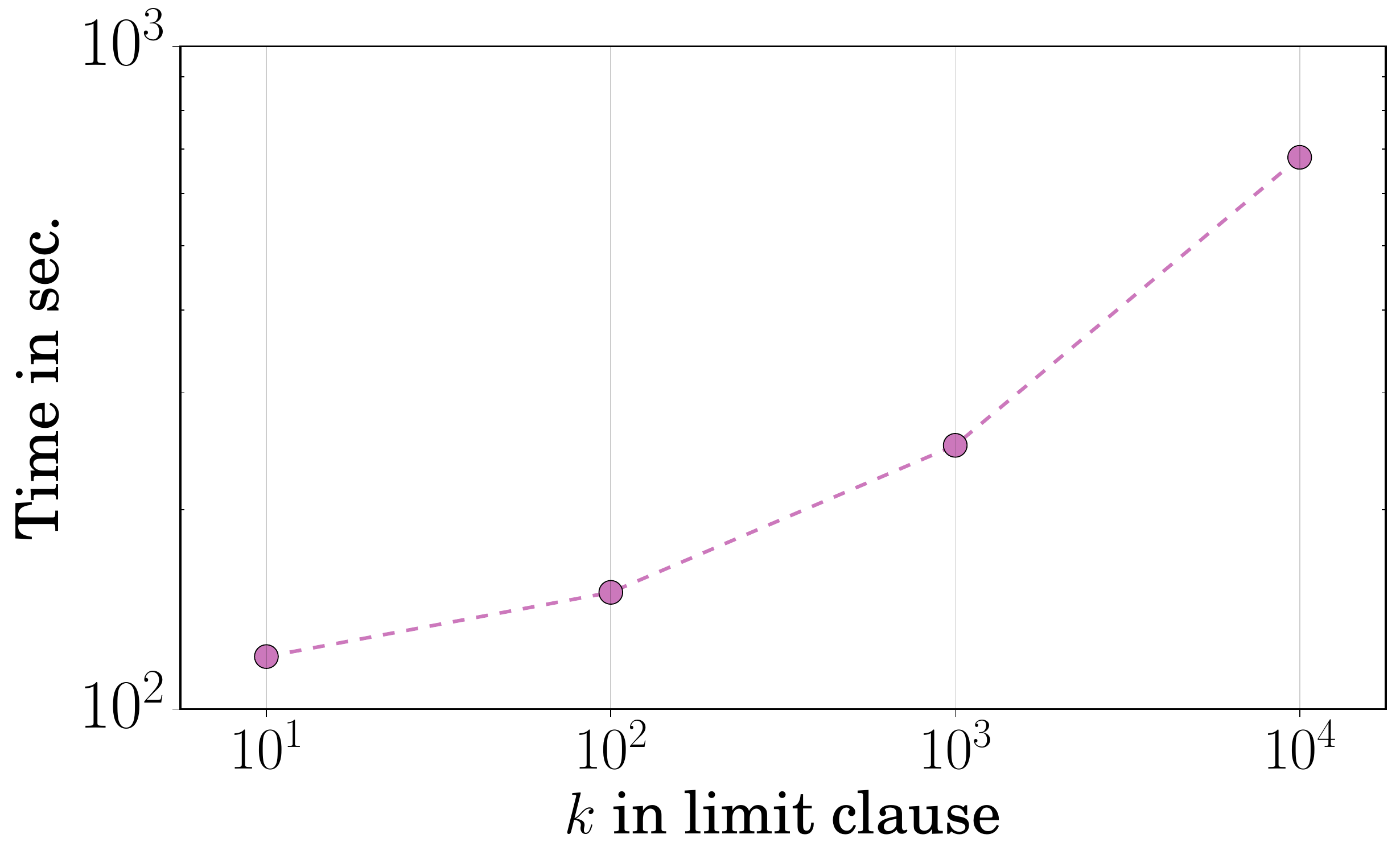}
		\caption{Memetracker 3-neighborhood}  \label{fig:friend:3path}
	\end{subfigure}
	\begin{subfigure}{0.24\linewidth}
		\includegraphics[scale=0.17]{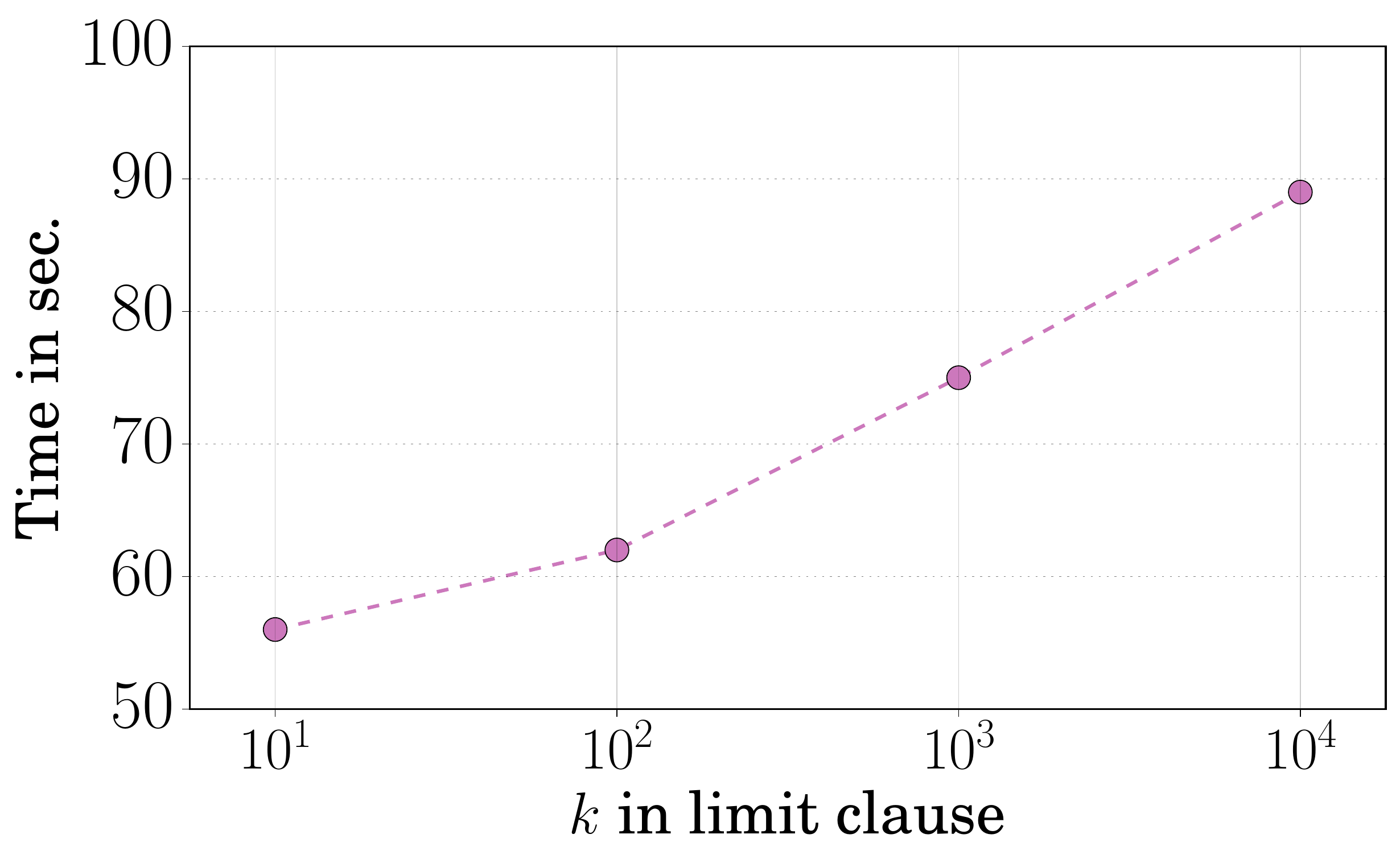}
		\caption{Friendster 2-neighborhood}  \label{fig:meme:4path}
	\end{subfigure}
	\begin{subfigure}{0.24\linewidth}
		\includegraphics[scale=0.17]{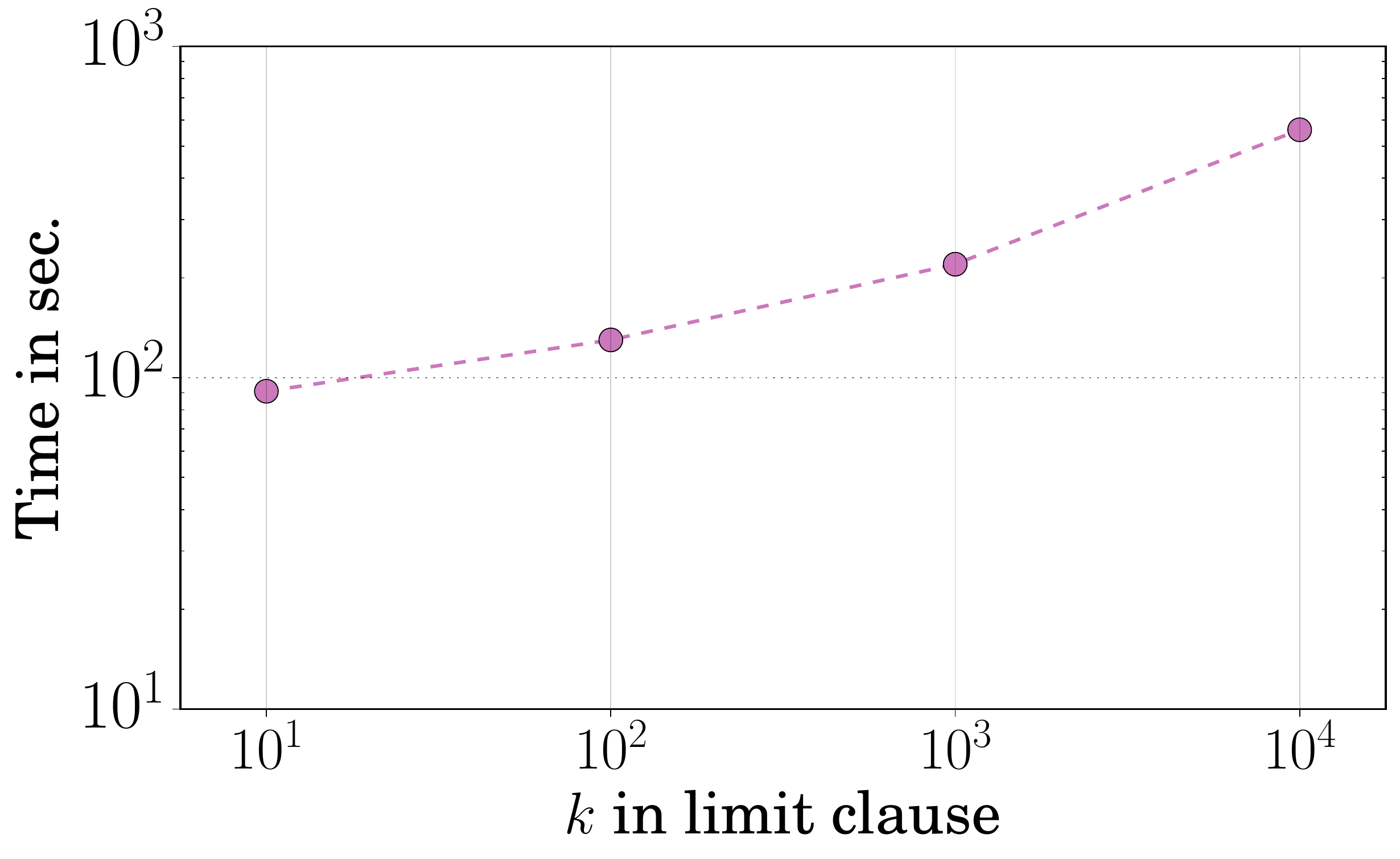}
		\caption{Friendster 3-neighborhood}  \label{fig:meme:3path}
	\end{subfigure}
\caption{\textsc{LinDelay} performance on large-scale datasets} \label{fig:largescale}
\end{figure*}

\begin{figure}
\scalebox{0.9}{
\begin{tabular}{lccccc}\toprule
	& $\mathsf{SF}=10$  & $\mathsf{SF}=20$ & $\mathsf{SF}=30$ & $\mathsf{SF}=40$ & $\mathsf{SF}=40$ \\\midrule
	\textbf{Q3}    & 5.91s & 9.63s & 13.47s & 18.23s & 22.18s  \\
	\textbf{Q10}    & 2.82s & 3.47s & 4.65s & 5.23s & 6.46s  \\
	\textbf{Q11}    & 0.78s & 1.07s & 1.56s & 1.82s & 2.36s  \\\bottomrule
\end{tabular}}
	\caption{Scalability for different scale factors (\textsf{SF}) in LDBC} \label{table:ldbc}
\end{figure}

\begin{figure}
	\scalebox{0.9}{
		\begin{tabular}{lcccc}\toprule
			& $\mathsf{k}=10$  & $\mathsf{k}=10^2$ & $\mathsf{k}=10^3$ & $\mathsf{k}=10^4$ \\\midrule
			\textbf{four cycle}    & 0.85s & 0.95s & 1.2s & 1.8s   \\
			\textbf{six cycle}    & 13.1s & 17.7s & 25.6s & 38.2s   \\
			\textbf{eight cycle}    & 33.4s & 48.7s & 63.9s & 77.8s   \\		
			\textbf{bowtie}    & 112s & 125s & 156s & 192s   \\\bottomrule
	\end{tabular}}
	\caption{Cyclic query performance on the \textsf{DBLP} dataset for different values of $k$ in the \sqlhighlight{LIMIT} clause.} \label{table:cyclic}
\end{figure}	
	\section{Related Work} \label{sec:related}

\introparagraph{Top-$k$} Top-k ranked enumeration of full join queries has been studied extensively by the database community for both certain~\cite{li2005ranksql, qi2007sum, ilyas2004rank, li2005ranksql2,akbarinia2011best,ilyas2008survey,bast2006io,tsaparas2003ranked} and uncertain databases~\cite{re2007efficient, zou2010finding}. Most of these works exploit the monotonicity property of scoring functions, building offline indexes and integrate the function into the cost model of the query optimizer in order to bound the number of operations required per answer tuple. We refer the reader to~\cite{ilyas2008survey} for a comprehensive survey. {We note that none of these works consider non-trivial join-project queries (see~\autoref{sec:semantics} for more discussion). Ours is the first work to consider the ranked enumeration of arbitrary join-project queries.}

\introparagraph{Rank aggregation algorithms} 	Top-k processing over ranked lists of objects has a rich history. The problem was first studied by Fagin et al.~\cite{fagin2002combining, fagin2003optimal} where the database consists of $N$ objects and $m$ ranked streams, each containing a ranking of the $N$ objects with the goal of finding the top-$k$ results for coordinate monotone functions. The authors proposed Fagin's algorithm (\textsc{FA}) and Threshold algorithm (\textsc{TA}), both of which were shown to be instance optimal for database access cost under  sorted list access and random access model. {A key limitation of these works is that it expects the input to be materialized, i.e., $Q(D)$ must already be computed and stored for the algorithm to perform random access.}

\introparagraph{Unranked enumeration of query results}
Recent work by Kara et al.~\cite{kara19} showed that for a small but important fragment of CQs known as hierarchical queries, it is possible to obtain a tradeoff between preprocessing and delay guarantees. Importantly, this result is applicable even in the presence of arbitrary projection. {However, the authors did not investigate how to add ranking because adding priority queues at different location in the join tree leads to different complexities.} In fact, follow up work~\cite{deep2021enumeration} showed that the same unranked enumeration could be performed with better delay guarantees under certain settings. {Our work considers the class of CQs with arbitrary projections and we are also able to extend the main result to UCQs, an even broader class of queries. Naturally, our algorithm automatically recovers the existing results for full CQs as well~\cite{deep2019ranked, tziavelis2020optimal}, in addition to the first extensive empirical evidence on how ranked enumeration can be performed for CQs containing projections beyond free-connex queries.} 

\introparagraph{Factorization and Aggregation}
Factorized databases 
\cite{bakibayev2012fdb,
	olteanu2015size,ciucanu2015worst} exploit the distributivity of product over union to represent query results compactly
and generalize the results on bounded \fhw to the non-Boolean case~\cite{olteanu2015size}. \cite{abo2016faq} captures a wide range of aggregation problems over semirings. Factorized representations 
can also enumerate the query results with constant delay 
according to lexicographic orders of the variables~\cite{bakibayev2013aggregation}. For that to work, the lexicographic order must "agree" with the factorization order. {However, it was shown in~\cite{deep2019ranked} that the algorithm for lexicographic ordering is not optimal. Further, since all prior work in this space using the concept of variable ordering, adding projections to the query forces the building of a GHD that can materialize the entire join query result, which is expensive and an unavoidable drawback.}

\introparagraph{Ranked enumeration}
Both Chang et al. \cite{chang2015optimal} and
Yang et al. \cite{yang2018any} provide any-$k$ algorithms for
\emph{graph queries} instead of the more general CQs;
Kimelfeld and Sagiv \cite{kimelfeld2006incrementally} give an any-$k$ algorithm for acyclic queries with polynomial delay. Recent work on ranked enumeration of MSO logic over words is also of particular interest~\cite{bourhis2020ranked}. {None of these existing works give any non-trivial guarantees for CQs with projections. Ours is the first work in this space that provides non-trivial guarantees.}

	\section{Conclusion}

In this paper, we study the problem of ranked enumeration for CQs with projections. We present a general algorithm that can enumerate query results according to two commonly-used ranking functions (\SUM, \lexi) with near-linear delay after near-linear preprocessing time. We also show how to extend our results to a broader class of queries known as UCQs. For star queries, an important and practical fragment of CQs, we further show how to achieve a smooth tradeoff between the delay, preprocessing time and space used for data structure. Extensive experiments demonstrate that our methods are up to three orders of magnitude better when compared to popular open-source RDBMS and specialized graph engines. This work opens up several exciting future work challenges. The first important problem is to extend our results from main memory setting to the distributed setting. Since the cost of I/O must also be taken into account, it becomes important to identify the optimal priority queue storage layout to ensure that access cost is low. It would also be interesting to develop output balanced algorithms. The second exciting challenge is to incorporate approximation into the ranking. For some applications, it may be sufficient to get an approximately ordered output which could lead to improved running time guarantees. Finally, it would be useful to re-rank the query results when the ranking function is changed by the user and extend our ideas to non-monotone ranking functions.

\smallskip
\introparagraph{Acknowledgments} This research was supported in part by National Science Foundation grants CRII-1850348 and III-1910014. We would like to thank the anonymous reviewers for their careful reading and
valuable comments that contributed greatly in improving the manuscript. We also thank Wim Martens for the reference~\cite{bonifati2020analytical} that motivates the study of star queries.
	
	\bibliographystyle{ACM-Reference-Format}
	\balance
	\bibliography{reference}  
	
	\clearpage
	\onecolumn
	\appendix
\section{Other Enumeration Semantics} \label{sec:semantics}

\cite{tziavelis2020optimalv3} discussed a more fine-grained approach to ranked enumeration over joins queries with projections. The goal is to enumerate the ranked output $Q(D)$ where $Q$ contains projections and the ranking function is defined over all the attributes ($\mathbb{A}$). Since the weights are placed on the attributes/relations and the output tuple $t$ does not contain all the attributes, each answer could have multiple weights associated with it.
In this case, one option is to return $t$ with the list of all the weights. Our algorithm for ranked enumeration can be extended to handle this trivially. The second semantics identified is the min-weight-projection semantics. In this case, $t$ is returned with the smallest possible weight.~\cite{tziavelis2020optimalv3} showed how their algorithm for full join can be extended to incorporate min-weight-projection semantics. We note here that the results for projection queries in~\cite{tziavelis2020optimalv3} can only handle acyclic free-connex queries, a restricted subclass of acyclic CQs (see~\autoref{sec:full} for the definition). Further, in our problem formulation, the ranking function depends only on the projection attributes and not all the attributes.

\section{Using Existing Algorithms} \label{sec:existing}

\begin{algorithm}[t]
	\SetCommentSty{textsf}
	\DontPrintSemicolon 
	\SetKwInOut{Input}{Input}
	\SetKwInOut{Output}{Output}
	\SetKwFunction{len}{\textsf{len()}}
	\SetKwFunction{insertt}{\textbf{insert}}
	\SetKwData{pq}{$\mathsf{PQ}$}
	\Input{Join-project query $Q$, ranking function \sqlhighlight{SUM} and database $D$}
	\Output{$Q(D)$ in ranked order}
	\SetKwData{ptr}{\textsf{ptr}}
	
	\SetKwData{return}{\textbf{return}}
	\SetKwData{dedup}{\textsf{dedup}}
	\SetKwData{last}{\textsf{last}}
	\SetKwData{hasNext}{\textsf{hasNext}}
	\SetKwData{getNext}{\textsf{getNext}}
	
	$Q' \leftarrow \text{ full query obtained by dropping the projection operator from } Q$ \;
	Provide $Q'$, ranking function \sqlhighlight{SUM} that assigns weight zero to all values of attributes $\mathbb{A} \setminus \bA$, and $D$ to $\mA$ for pre-processing; \;
	$\last \leftarrow \emptyset$ \;
	\While{$\mA.\hasNext()$}{
		$o \leftarrow \mA.\getNext()$ \label{line:mA} \tcc*{enumerate the answer from $\mA$ }
		\If{$\last \neq o[\bA]$}{
			output $o[\bA]$, $\last \leftarrow o[\bA]$ \;
		}
	}
	\return \tcc*{Enumeration complete}
	\caption{{\sc UseExistingAlgorithm}}
	\label{algo:existing}
\end{algorithm} 

Given a blackbox algorithm $\mA$ for ranked enumeration of full join queries, we can obtain an algorithm for ranked enumeration of join-project queries. \autoref{algo:existing} shows how $\mA$ can be used for enumeration of join-project queries. The key idea here is that $\mA$ can assign a weight of zero to all values of non-projection attributes. This will guarantee that tuples enumerated by $\mA$ will be sorted only over the sum of attribute values of $\bA$. The while loop contains a variable that stores the last tuple output $\bA$. This helps in ensuring that no duplicates are output for the join-project query. Intuitively, \autoref{algo:existing} is performing a ranked group by over the attributes $\bA$ and outputs an answer when the grouping value changes. We now show that there exist queries where the dealy for this algorithm could be large. Consider the query 
$$Q = \pi_{X_1} (\Join_ {i \in [\ell] } R_i(X_i, Y))$$.

Suppose that each relation $R_i$ has $N$ attribute value for $X_i$ that is connected one $Y$ attribute value $y^\star$ for each $i \in \{1, \dots, \ell\}$. Then, the smallest answer of $QD)$ will be output $N^{\ell - 1}$ times by $\mA$ on~\autoref{line:mA}. This is because the join of $R_2 \Join \dots \Join R_\ell$ has a size of $N^{\ell - 1}$. This directly implies that the delay guarantee is $\Omega(N^{\ell - 1})$.

\section{Missing Proofs} \label{sec:correctness}

\lemdelay*
\begin{proof}
	In order to prove the delay guarantee, we analyze the worst-case running time that can happen in each iteration. Without loss of generality, we assume that there exists at least one non-anchor projection attribute in all leaf nodes (otherwise, we can simply remove the node from the join tree after the full reducer pass).
	
	Let us first analyze the procedure \textsc{Enum}. The root priority queue (denoted $\textsc{PQ}_r$) can contain at most $|D|$ entries for each output result. Indeed, in the worst-case, the output tuple emitted to the user on~\autoref{user:output}  may join with every tuple in the root node relation. Therefore, in the worst-case, for every tuple outputted to the user, we invoke the \textsc{Topdown} procedure $|D|$ times.
	
	We now argue that \textsc{Topdown} takes at most $O(|D| \log |D|)$ time. There are two key observations to be made. First, \textsc{Topdown} takes $O(1)$ time if $c.\textsf{next}$ is already populated. Intuitively, this means that in a previous iteration, there was a recursive call made to \textsc{Topdown} that populated the $\textsf{next}$ and we are now simply reusing the computation. Thus, for every tuple output to the user, it suffices to count how many times \textsc{Topdown} is invoked in total, such that $\textsf{next}$ is not empty. We say that such invocations of \textsc{Topdown} are \emph{non-trivial}. Our claim is that the number of non-trivial \textsc{Topdown} calls is at most $O(|D|)$. Consider the set of all anchor attribute values for some node $j$ (denoted $\mL_{R_j}$). We fix some $u \in \mL_{R_j}$. In the worst-case, \textsc{Topdown}$(c, j)$ such that $u = \pi_{\key{R_j}}(c.t)$ is invoked by all tuples $t'$ in the parent relation of node $j$ such that $\pi_{\key{R_j}}(t') = u$ but only the first call will be non-trivial since the first call populates the $c.\textsf{next}$. Therefore, it follows that for each $u \in \mL_{R_j}$, there is at most one non-trivial call to \textsc{Topdown}$(c, R_j)$ for a fixed node $R_j$. 	Since the size of the join tree is at most $O(1)$, the total number of non-trivial calls over all nodes in the join tree is $\sum_j \sum_{u \in \mL_{R_j}} 1 = \sum_j |D| = O(|D|)$. Finally, we perform one pop operation and a constant number of inserts into the priority queue in every non-trivial invocation which adds another $\log|D|$ factor in the running time, giving us the claimed delay guarantee of $O(|D| \log |D|)$.
\end{proof}

\lempreprocess*
\begin{proof}
	For a node $R_i \in \mT$, we initialize an empty priority queue for each possible  value in $\key{R_i}$. The for loop on~\autoref{for:1} takes $O(1)$ time as each node has $O(1)$ children in the decomposition and all operations on line 8-12 also take $O(\log |D|)$ time. Overall, since each node has at most $O(|D|)$ tuples, the claim readily follows.
\end{proof}

\lemcorrect*
\begin{proof}
	We will prove our claim by induction on height of the tree. We will show that at each relation $R$, the priority queue $\pq_R[u]$ correctly computes the answers over all projection attributes in the subtree rooted at $R$ (denoted $\bA^\pi_R$) in ranked order.
	
	\smallskip
	\introparagraph{Base Case} Correctness for ranked output of leaf relations is trivial. For all tuples $t$ in the relation at a leaf node with the same anchor attribute value $u$, $\pq_R[u]$ contains all tuples $t$ and the priority queue (whose implementation is assumed to be correct) will pop out $t[\bA^\pi_R]$ in the correct order since the score $\rank(t[\bA^\pi_R])$ is used as the priority queue comparator function and all projection attributes $\bA^\pi_R$ are present in $t$ already. It also follows that $R[\bA^\pi_R]$ will be materialized by chaining the cells popped from the priority queue since the leaf relation already contains all the attributes in $\bA^\pi_R$.
	
	\smallskip
	\introparagraph{Inductive Case}	Consider a relation $R$ with children relations $R_1, \dots R_s$. Let $c$ be the cell that was the input to the {\sc TopDown} procedure call under consideration. Let $u = c.t[\key{R}]$ and $u_i = c.t[\key{R_i}]$. Let $t_\pi = \mathsf{output}(c_{\mathsf{next}})$ be the next tuple and let $c_{\mathsf{next}}$ be the cell that is to be returned to the parent of $R$ by the recursive call to $R$ (i.e. $c_{\mathsf{next}}$ follows $c$ in the ranked output of the subtree rooted at $R$). Our goal is to show that $t_\pi$ is generated correctly and $c_{\mathsf{next}}$ is inserted in $\pq_R[u]$. This will guarantee the correct ordering of the tuples over attribute set $\bA^\pi_R$. From the induction hypothesis, we have that the ranked output, over the attributes $\bA^\pi_{i}$, from each $R_i$ is generated correctly. For the sake of contradiction, suppose that $t'_\pi$ is the next smallest candidate, i.e. $\rank(t'_\pi) < \rank(t_\pi)$.
	
	\smallskip
	Lines~\ref{for:loop}-\ref{for:loopend} in~\autoref{algo:enumerate} generates $s$ possible candidate cells (and is the only way to generate new candidates) that could generate $t_\pi$. Let $\hat{c}_i$ be the cell returned from the recursive call to {\sc TopDown} for relation $R_i$. Then, there are $s$ possible cells that are generated by {\sc TopDown} for relation $R$ whose array of pointers will be one of:
	
	\begin{align*}
		&\hat{c}_1 , c_2 , c_3 , \dots, c_s, \\
		&c_1 , \hat{c}_2 , c_3 , \dots, c_s, \\
		&\dots, \\
		&c_1 , c_2 , c_3 , \dots, \hat{c}_s 
	\end{align*}
	
	Here, $c_i$ is $i^{\text{th}}$ pointer in the pointer array of $c$ and $\rank(\mathsf{output}(\hat{c}_i)) \geq \rank(\mathsf{output}(c_i))$ since the priority queue of the $R_i$ will generate answers in increasing order (from the induction hypothesis). Assume for the sake of contradiction that $c'_{\mathsf{next}} = \langle c.t, [c'_1, c'_2, \dots, c'_s], \bot \rangle$ is the cell such that $t'_\pi = \out(c'_{\mathsf{next}})$ and is not one of the $s$ candidate cells generated above. Since $t'_\pi$ has a smaller score than $t_\pi$, it follows that $c'_{\mathsf{next}}$ has not been inserted in the priority queue yet (otherwise the priority queue will return $c'_{\mathsf{next}}$ correctly). We will show that such a scenario will violate the monotonicity property of the sum ranking function.

	\smallskip
	\textbf{Case 1. $\rank(\mathsf{output}(c_i)) \geq \rank(\out(c'_i))$ for $i \in [s]$.} This scenario implies that $t'_\pi$ has been generated before $t_\pi$, which would violate our assumption that  $t'_\pi$ has not been inserted in the priority queue. Indeed, since $\rank(\out(c_i)) \geq \rank(\out(c'_i))$ and lines~\ref{for:loop}-\ref{for:loopend} in~\autoref{algo:enumerate} generate cells that are monotonically increasing, there exists a sequence of generations that generate $c_{\mathsf{next}}$ from $c'_{\mathsf{next}}$ which would mean that $c'_{\mathsf{next}}$ was present in the priority queue at some point.
	
	\smallskip
	\textbf{Case 2. $\rank(\mathsf{output}(c_i)) < \rank(\out(c'_i))$ for $i \in [s]$.}	This scenario implies that $\sum_i \rank(\out(c_i))  =  \rank(t_\pi) < \sum_i \rank(\out(c'_i))  = \rank(t'_\pi)$ which contradicts our assumption that $t'_\pi$ has a smaller score than $t_\pi$.
	
	\smallskip
	\textbf{Case 3. $\rank(\mathsf{output}(c_i))$ and $\rank(\mathsf{output}(c'_i))$ are incomparable.}	If $c'_{\mathsf{next}}$ has not been inserted in the priority queue yet, there are two possible scenarios. Either $c'_{\mathsf{next}}$ will be generated by some cell $c''_{\mathsf{next}}$ that is already in the priority queue. Clearly, $\rank(\out(c'_{\mathsf{next}})) > \rank(\out(c''_{\mathsf{next}})) > \rank(\out(c_{\mathsf{next}}))$ using the same reasoning as \textbf{Case 1} which contradicts our assumption. The second possibility is that $c'_{\mathsf{next}}$ and $c_{\mathsf{next}}$ are generated in the same loop. But this would again imply that $c'_{\mathsf{next}}$ is in the priority queue, a contradiction to our assumption.
	
	Therefore, it cannot be the case that $\rank(t'_\pi) < \rank(t_\pi)$ which proves that for relation $R$, all tuples generated over the attributes $\bA^\pi_R$ are in the correct order and consequently, the claim holds for all relations in the join tree, including the root.
	
	\smallskip
	In the second part of the proof, we will show that every result enumerated must belong to $Q(D)$. This is a direct consequence of the properties of the join tree. Indeed, if some tuple $t \not \in Q(D)$ is enumerated, that would violate the join condition on some attribute $X \in \mathbf{A}$. Next, we will show that every query result in $Q(D)$ will be enumerated using induction on the join tree. In the base case, for some leaf node $R$ in $\mT$, every query result over attributes $\bA^\pi_R$ will be enumerated from $\topdown$ because the relation is materialized. By the induction hypothesis, for each  $i \in [s]$,
	$$
	Q_i = (\Join_{j \in \text{ nodes in subtree rooted at }i} R_j)[\bA^\pi_i]
	$$ is enumerated completely generated. For each tuple $t$ in $R$, lines~\ref{for:loop}-\ref{for:loopend} guarantees that,
	
	$$
		((\times_{i\in [s]} Q_i ) \Join t) [\bA^\pi_R]
	$$
	must be inserted into $\pq[\pi_{\key{R}}(t)]$, and thus enumerated. In other words, the projection of cartesian product between each subquery rooted at $R_i$ (i.e. $Q_i$ as defined above) and $t$ over $\bA^\pi_R$ is generated.
	This implies that the subquery induced by the subtree rooted at $R$ enumerates all join result over $\bA^\pi_{R}$. In the preprocessing phase, all tuples in $R$ are put into the entries of $\pq[\pi_{\key{R}}(t)]$, thus every query result in $Q(D)$ will necessarily be enumerated.
\end{proof}

\lemlexi*
\begin{proof}
	We will prove that the algorithm enumerates the query result in the correct lexicographic order using induction on the parameter $i$ in~\autoref{algo:lexicographic}.
	
	\smallskip
	\textbf{Base case. $i=1$.} Note that since $\domain(A_1)$ is sorted $t'$, will be initialized with $a \in \domain(A_1)$ in sorted order which guarantees that all output tuples with $a$ as the attribute value for $A_1$ will be enumerated before tuples with $a'$ as the attribute value for $A_1$ if $a \leq a'$, thus respecting the lexicographic order for $A_1$.
	
	\smallskip
	\textbf{Inductive case. $i=k$.} From the induction hypothesis, we get that the algorithm will enumerate correctly for the first $k-1$ attributes including $A_{k-1}$. From the previous iteration, $\mL'$ is set to be the sorted list of attribute values for $A_k$ that join with $R_{k}(A_k, A_{k+1})$. Therefore, it follows that for the recursive call where $i=k$, $t'$ will have a attribute values set in sorted order for $A_k$, ensuring that all answers containing $a = \pi_{A_{k}}(t')$ will be enumerated before answers containing $a'$ as attribute value for $A_k$ for all $a \leq a'$. Thus, the argument holds for $i=m$ as well, concluding the proof.

	Next, we analyze the time complexity. The preprocessing phase takes $O(|D| \log |D|)$ time in total for sorting and building hashing tables for each relation.  Observe that after fixing a value in $\domain(A_i)$, it takes at most $O(|D|)$ time to perform the semi-joins and find the set of candidate values in $\domain(A_{i+1})$. Whenever $i$ reaches $m$, a valid query result will be enumerated. As the query size as well as $m$ is a constant, the delay between two consecutive query results is at most $O(|D|)$,  which improves Lemma~\ref{lem:delay} by a log factor.  
\end{proof}

\lemstar*
\begin{proof}
	 \introparagraph{Time and Space complexity} First, we analyze its time complexity. Computing data statistics for each relation $R_i$ takes $O(|D|)$ time. As the degree threshold is $\delta$, there are at most $O(\frac{|D|}{\delta})$ heavy values for each $A_i$. The size of $\mO^H$ can be bounded by $O(|D| \cdot (\frac{|D|}{\delta})^{m-1})$, as well as the time cost for computing $\mO^H$ by the Yannakakis algorithm. At last, initializing data structures for each $Q_i$ takes linear time in terms of $O(|D|)$. As the number of such queries is $O(m)$, the time cost for line 3-6 is $O(|D| \cdot m)$. Overall, the time complexity of Algorithm~\ref{algo:preprocess:star} is $O(|D| \cdot (\frac{|D|}{\delta})^{m-1})$. For the space usage, storing $\mO^H$ takes $O((\frac{|D|}{\delta})^{m})$ space since its size can be bounded by $O((\frac{|D|}{\delta})^{m})$ using the AGM bound.

	\smallskip
	\introparagraph{Delay guarantee} Note that min-extraction and update of the priority queue takes $O(\log |D|)$ time. The expensive part is the invocation of \textsc{EnumAcyclic} for each $Q_i$.  In the join tree $\mT_i$, each node contains at least one projection attribute. Plugging in Lemma~\ref{lem:delay} and the observations in~\autoref{sec:finegrained}, the delay of enumerating query result from $Q_i$ is at most the degree of values of attribute $A_i$, i.e., $\delta$, since tuples in $R_i$ are light in this case, by construction of $Q_i$. Together, the delay of~\autoref{algo:enumeration:star} is $O(\delta \log|D|)$. By setting $\delta = |D|^{1-\epsilon}$ for arbitrary constant $0<\epsilon <1$, we can achieve the result in Theorem~\ref{thm:ranked:star}.
	
	\smallskip
	\introparagraph{Correctness} Recall that $O^H$ is already materialized and sorted order and since the enumeration of each subquery $Q_i$ is performed using~\textsc{EnumAcyclic} (which enumerated the query result correctly), the output of each $Q_i$ is enumerated in the correct order. It remains to be shown that~\textsc{EnumStar} can also combine the answers from $O^H$ and $Q_i$ correctly. Suppose that $t$ was the last answer output to the user. If $t$ was generated by some $Q_i$ (which can be checked in constant time), then we find the next smallest candidate from $Q_i$ and insert it into \textsf{PQ}. Otherwise, if $t \in O^H$, we find the next tuple in the materialized output and insert it into \textsf{PQ}. Thus, since the smallest candidates for each $Q_i$ and $O^H$ are present in the \textsf{PQ}, the smallest answer that must be output to the user necessarily has to be from the priority queue. Thus, the priority queue performs the $m+1$-way merge correctly. This concludes the proof.
\end{proof}

\lemoptimal*
\begin{proof}
	For a star query $Q^\star_m$ and database $D$, assume there is an algorithm that  supports $O(|D|^{1-\epsilon} \log |D|)$-delay enumeration after $O(|D|^{1 + (m-1)\epsilon - \epsilon'})$ preprocessing time. Then, there is an algorithm evaluating $Q^\star_m(D)$ in  (big-Oh of)
	\[ |D|^{1 + (m-1)\epsilon - \epsilon'} + |D|^{1-\epsilon} \log |D| \cdot | Q^\star_m(D)| \]
	time. Let $\alpha = m \epsilon - \epsilon'$ for some constant $\epsilon' < \epsilon$. Consider an instance $D'$ for $Q^\star_m$ with output size $N^{\alpha}/\log N$. The running time of the algorithm is (big-Oh of)
	\begin{align*}
	&	|D'|^{1 + (m-1)\epsilon - \epsilon'} + |Q^\star_m(D')| \cdot (|D'|^{1-\epsilon} \log |D'|) \\
	= &  |D'|^{1-\frac{\epsilon'}{m}} \cdot |Q^\star_m(D')|^{(1-\frac{1}{m})}
	\end{align*}
	which breaks the optimality of \autoref{lem:twopath:evaluation}.
\end{proof}

\begin{lemma}
	The delay guarantee of~\autoref{lem:delay} is optimal up to polylogarithmic factors subject to the combinatorial BMM conjecture.
\end{lemma}

\section{More fine grained results} \label{sec:finegrained}

In this section, we show tighter delay guarantees for~\autoref{thm:general} under certain settings. Consider a query where each relation contains at least one projection attribute (for instance, $Q^\star_{m}$). Let's fix a relation $R$ and let the projection attributes in $R$ be $\mathbf{B}$. Denote $\Delta_R = \max_{u \in \domain(\mathbf{B})} |\sigma_{\mathbf{B} = u} R|$ and $\Delta = \max_{\text{non-leaf relations }R} \Delta_R$. Then, the delay guarantee obtained by~\autoref{thm:general} is at most $O(\Delta \log|D|)$. A trivial upper bound for $\Delta = |D|$ since the size of all relations is $|D|$. However, if there are constraints known about the degree of the tuples in the relations, we can improve our results. We give two concrete scenarios below.

\smallskip
\introparagraph{Bounded Degree Database} For each relation, if every attribute value is present in at most $\delta$ tuples, it implies that the degree is at most $\delta$. In this case, it can be shown that $\Delta = \delta$ and we get improved delay guarantees since the duplication is bounded when each relation has projection attributes.

\smallskip
\introparagraph{Bounded Degree By Preprocessing} Consider the star query $Q^\star_{m}$. Recall that in~\autoref{algo:preprocess:star}, we preprocess to reduce the degree of all $a_i$ from $|D|$ to $\delta = |D|^{1-\epsilon}$. Thus, the discussion above applies and we can achieve $O(|D|^{1-\epsilon} \log|D|)$ delay.

\medskip
\noindent Improved delay guarantee is obtained in~\autoref{algo:enumerate} only when the while loop can terminate in $o(|D|)$ time. For any anchor attribute value $u$ for some relation $R$ in the join-tree, this can only happen when the tuples generated over $\bA_R$ are not duplicated more than a certain number of times. Since every relation contains a projection attribute and has bounded degree $\Delta$, it means that $\pq_R[u]$ can have at most $\Delta$ duplicates which gives the better bound.

\section{Recovering results for full queries} \label{sec:full}

Next, we outline how \autoref{algo:enumerate} achieves delay $O(\log |D|)$ after $O(|D|)$ preprocessing for full acyclic queries. The key observation is that when the query is full, the while loop always terminates after a constant number of operations. This is because no two answer tuples over $\bA^\pi_R$ are the same, i.e, the attribute values for any two answer tuples cannot be identical. Further, any answer tuple at $R$ is inserted at most a constant number of times. This guarantees that \textsf{is\_equal} procedure will return false after popping an answer constant number of times and the while loop terminates. Thus, in the worst-case, \textsc{TopDown} visits each relation of the join tree at most once and each iteration takes only $O(\log |D|)$ time due to the priority queue operations.

\smallskip
\introparagraph{Free-connex queries} A join query $Q$ is said to be free-connex if $Q$ is acyclic and the hypergraph containing edges of $Q$ and a new edge containing only the projection variables is also acyclic. Free-connex is an interesting class of queries which allows constant delay enumeration for a variety of problems~\cite{bagan2007acyclic, carmeli2019enumeration}. If a join-project query is free-connex, then our actually algorithm achieves $O(\log|D|)$ delay enumeration (instead of the general $O(|D|\log|D|)$ bound) after linear time preprocessing. It can be shown that for free-connex queries, all projection attributes are connected and must be at the top of the tree and all non-projection attributes are connected and at the bottom of the tree. This observation allows us to remove all relations that do not contain any projection attributes or relation where only anchor attributes are projection attributes, transforming the free-connex query into a full join query.

\section{Submodular width lower bound} \label{sec:lowerbound}

Abboud, Backurs and, Vassilevska Williams present the Combinatorial $k-$Clique Hypothesis~\cite{abboud2018if}.

\begin{definition}
	Let $C$ be the smallest constant such that a combinatorial algorithm running in $O(n^{Ck/3})$ time exists for detecting a $k-$clique in a
	$n$ node $O(n^2)$ edge graph, for all sufficiently large constants $k$. The Combinatorial $k-$Clique Hypothesis states that $C = 3$.
\end{definition}

Next, we recall a cycle detection hypothesis from~\cite{lincoln2018tight} that depends on the combinatorial $k-$clique hypothesis.

\begin{theorem}[Combinatorial Sparse k-Cycle Lower Bound]
	Detecting a directed odd $k-$cycle in a graph with $n$ nodes and $m = n^{(k+1)/(k-1)}$ in time $O(m^{(1-\epsilon)2k/(k+1)})$ for any constant $\epsilon > 0$ violates the Combinatorial $k-$Clique Hypothesis.
\end{theorem}

Consider an odd $k-$ cycle CQ over binary relations of the form $Q^\circlearrowright(x_1, x_2) = R_1(x_1, x_2) \Join R_2(x_2, x_3) \Join \dots \Join R_k(x_k, x_1)$. Then, we can show the following straightforward result.

\begin{lemma}
	For some odd $k$ cycle query $Q^\circlearrowright$ with $k > 3$, an algorithm that enumerates the query result with delay $O(|D|^{2k/k+1} - \epsilon)$ after $O(|D|^{2k/(k+1)} -\epsilon)$ preprocessing violates the Combinatorial Sparse $k-$Cycle Lower Bound.
\end{lemma}
\begin{proof}
	If the algorithm has delay guarantee $\delta = O(|D|^{2k/k+1} - \epsilon)$, then in $\delta$ time, we can decide whether there exists a cycle or not since the preprocessing time is also $O(|D|^{2k/k+1} - \epsilon)$, a contradiction.
\end{proof}

Note that the submodular width for the odd $k-$cycle query is $\subw = 2 - 2/(k+1) = 2k/(k+1)$. Thus, there exist queries for which the exponential dependence of \subw\ in~\autoref{thm:ucq} cannot be avoided assuming the Combinatorial $k-$Clique Hypothesis.

\section{Deferred Experiments} \label{sec:moreexp}

In this section, we present the experimental results for settings that were not covered in the main paper due to lack of space.

\subsection{Small Scale Experiments}

\begin{figure*}
	
	\begin{lstlisting}[language=SQL,
	deletekeywords={IDENTITY},
	deletekeywords={[2]INT},
	morekeywords={clustered},
	mathescape=true,
	xleftmargin=0pt,
	framexleftmargin=0pt,
	frame=tb,
	framerule=0pt ]
	
	$\mathsf{IMDB}_{2\mathsf{hop}}$ = SELECT DISTINCT $\mathsf{P}_1.\mathsf{name}, \mathsf{P}_2.\mathsf{name}$ FROM $\mathsf{Person}$ AS $\mathsf{P_1}$, $\mathsf{Person}$ AS $\mathsf{P_1}$, $\mathsf{PersonMovie}$ AS $\mathsf{PM}_1$, $\mathsf{PersonMovie}$ as $\mathsf{PM}_2$ WHERE $\mathsf{PM_1.pid} = \mathsf{PM_2.pid}$ AND $\mathsf{PM_1.aid} = \mathsf{P_1.aid}$ AND $\mathsf{PM_2.aid} = \mathsf{P_2.aid}$ AND $\mathsf{P_1.role}$ = 'ACTOR' AND $\mathsf{P_2.role}$ = 'ACTOR' ORDER BY $\mathsf{A}_1.\mathsf{weight} + \mathsf{A}_2.\mathsf{weight}$ LIMIT k;
	
	$\mathsf{IMDB}_{3\mathsf{hop}}$ = SELECT DISTINCT $\mathsf{P}.\mathsf{name}, \mathsf{M}.\mathsf{name}$ FROM $\mathsf{Person}$ AS $\mathsf{P}$, $\mathsf{Person}$ AS $\mathsf{P'}$, $\mathsf{Movie}$ AS $\mathsf{M}$, $\mathsf{PersonMovie}$ AS $\mathsf{PM}_1$, $\mathsf{PersonMovie}$ as $\mathsf{PM_2}$, $\mathsf{PersonMovie}$ as $\mathsf{PM_3}$ WHERE $\mathsf{PM_1.mid} = \mathsf{PM_2.mid}$ AND $\mathsf{PM_2.pid} = \mathsf{PM_3.pid}$ AND $\mathsf{PM_2.pid} = \mathsf{P'.pid}$ AND  $\mathsf{PM_1.pid} = \mathsf{P.pid}$ AND $\mathsf{PM_3.mid} = \mathsf{M.mid}$ AND $\mathsf{P.role}$ = 'ACTOR' AND $\mathsf{P'.role}$ = 'ACTOR' ORDER BY $\mathsf{P}.\mathsf{weight} + \mathsf{M}.\mathsf{weight}$ LIMIT k;
	
	$\mathsf{IMDB}_{4\mathsf{hop}}$ = SELECT DISTINCT $\mathsf{P}_1.\mathsf{name}, \mathsf{P}_2.\mathsf{name}$ FROM $\mathsf{Person}$ AS $\mathsf{P_1}$, $\mathsf{Person}$ AS $\mathsf{P_2}$, $\mathsf{Person}$ AS $\mathsf{P_3}$, $\mathsf{PersonMovie}$ AS $\mathsf{PM}_1$, $\mathsf{PersonMovie}$ as $\mathsf{PM_2}$, $\mathsf{PersonMovie}$ as $\mathsf{PM_3}$, $\mathsf{PersonMovie}$ as $\mathsf{PM_4}$ WHERE $\mathsf{PM_1.mid} = \mathsf{PM_2.mid}$ AND $\mathsf{PM_2.pid} = \mathsf{PM_3.pid}$ AND $\mathsf{PM_3.mid} = \mathsf{PM_4.mid}$ AND $\mathsf{PM_1.pid} = \mathsf{P_2.pid}$ AND $\mathsf{P_3.role}$ = 'ACTOR' AND $\mathsf{PM_2.pid} = \mathsf{P_3.pid}$ AND $\mathsf{P_1.role}$ = 'ACTOR' AND $\mathsf{PM_4.pid} = \mathsf{P_2.aid}$ AND $\mathsf{P_2.role}$ = 'ACTOR' ORDER BY $\mathsf{P}_1.\mathsf{weight} + \mathsf{P}_2.\mathsf{weight}$ LIMIT k;
	
	$\mathsf{IMDB}_{3\mathsf{star}}$ = SELECT DISTINCT $\mathsf{P}_1.\mathsf{name}, \mathsf{P}_2.\mathsf{name}, \mathsf{P}_3.\mathsf{name}$ FROM FROM $\mathsf{Person}$ AS $\mathsf{P_1}$, $\mathsf{Person}$ AS $\mathsf{P_2}$, $\mathsf{Person}$ AS $\mathsf{P_3}$, $\mathsf{PersonMovie}$ AS $\mathsf{AP}_1$, $\mathsf{PersonMovie}$ as $\mathsf{AP_2}$, $\mathsf{PersonMovie}$ as $\mathsf{AP_3}$  WHERE $\mathsf{PM_1.mid} = \mathsf{PM_2.mid} = \mathsf{PM_3.mid}$ AND $\mathsf{PM_1.pid} = \mathsf{P_1.pid}$ AND $\mathsf{PM_2.pid} = \mathsf{P_2.pid}$ AND $\mathsf{PM_3.pid} = \mathsf{P_3.pid}$ AND $\mathsf{P_1.role}$ = 'ACTOR' AND $\mathsf{P_2.role}$ = 'ACTOR' AND $\mathsf{P_3.role}$ = 'ACTOR' ORDER BY $\mathsf{P}_1.\mathsf{weight} + \mathsf{P}_2.\mathsf{weight} + \mathsf{P}_3.\mathsf{weight}$ LIMIT k;
	\end{lstlisting}
	\vspace*{-1.5em}
	\caption{Network analysis queries for IMDB.} \label{fig:imdbqueries}
\end{figure*}

\introparagraph{Lexicographic ordering} Figures~\ref{fig:imdb:2path:lex},~\ref{fig:imdb:2path:lex},~\ref{fig:imdb:4path:lex},~\ref{fig:imdb:3star:lex} show the running time for different values of $k$ in the limit clause for lexicographic function for \textsf{IMDB} dataset. Our conclusions remain the same as they were for \textsf{DBLP} dataset.

\introparagraph{Choice of different root of the join tree} For the \textsf{DBLP} dataset, we test the choice of different root nodes for $\mathsf{DBLP}_\mathsf{3hop}$ and $\mathsf{DBLP}_\mathsf{4hop}$ queries since these queries have a large choice of root nodes to pick from. We observed that for any choice of root node in the join tree, the difference in total execution time for $k$ ranging from $10$ to $10^4$ was less than $3\%$. This demonstrates that as long as the join tree or GHD is of the minimal width, the actual orientation of the GHD has no meaningful impact on the execution times.

\begin{figure*}[!htp]
	\begin{subfigure}{0.24\linewidth}
		\includegraphics[scale=0.17]{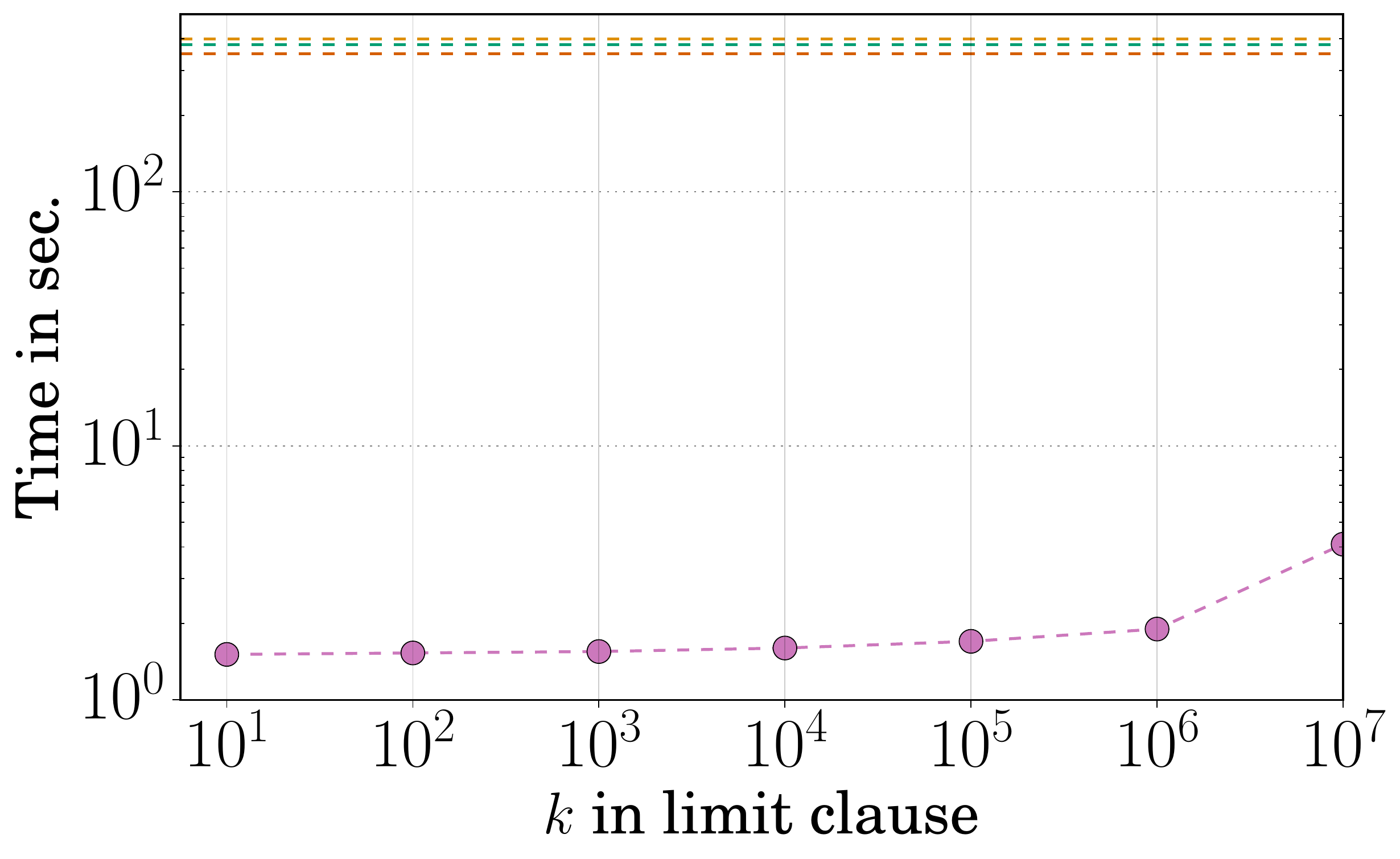}
		\caption{$\mathsf{IMDB}_{2\mathsf{hop}}$}  \label{fig:imdb:2path:lex}
	\end{subfigure}
	\begin{subfigure}{0.24\linewidth}
		\includegraphics[scale=0.17]{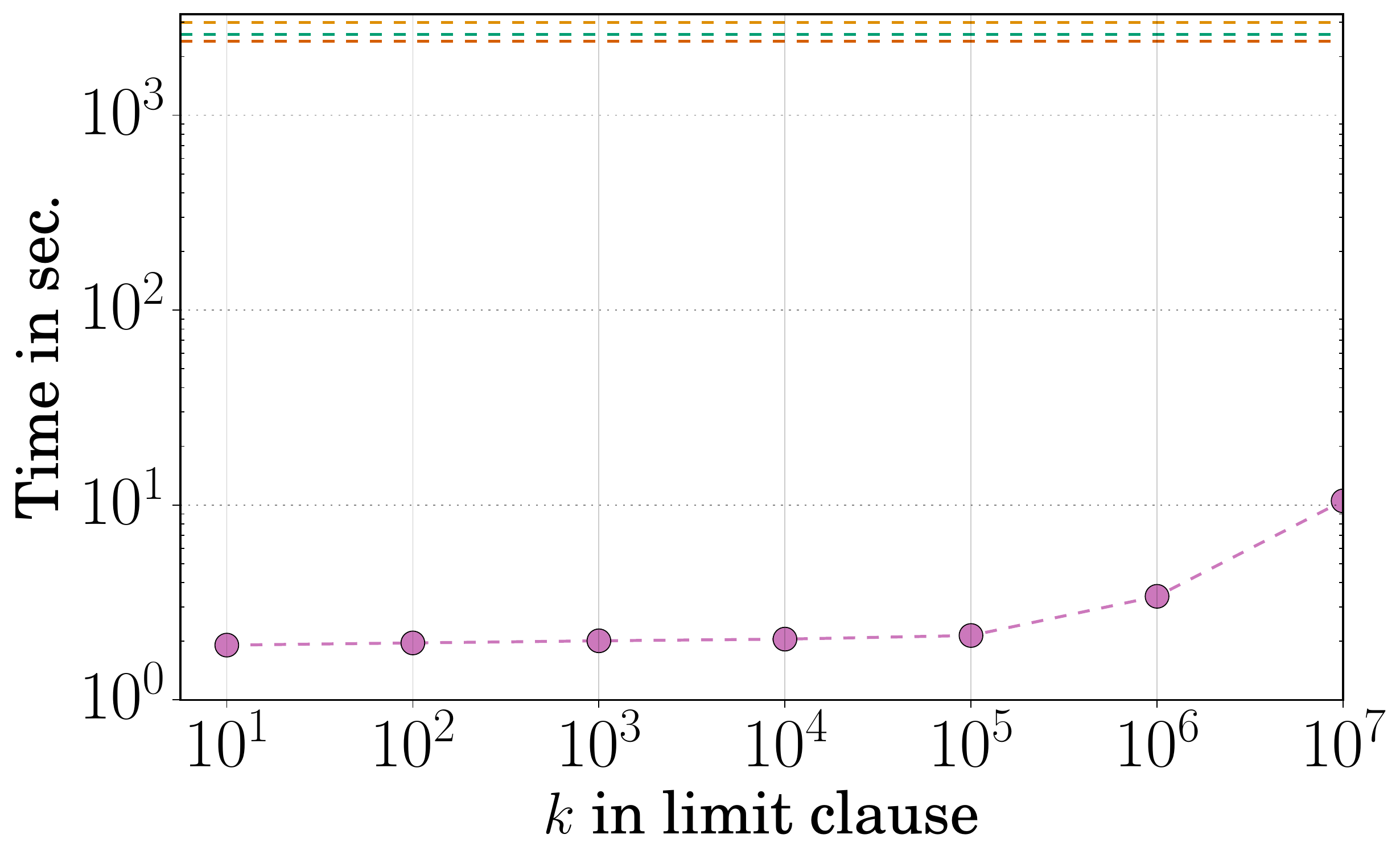}
		\caption{$\mathsf{IMDB}_{3\mathsf{hop}}$}  \label{fig:imdb:3path:lex}
	\end{subfigure}
	\begin{subfigure}{0.24\linewidth}
		\includegraphics[scale=0.17]{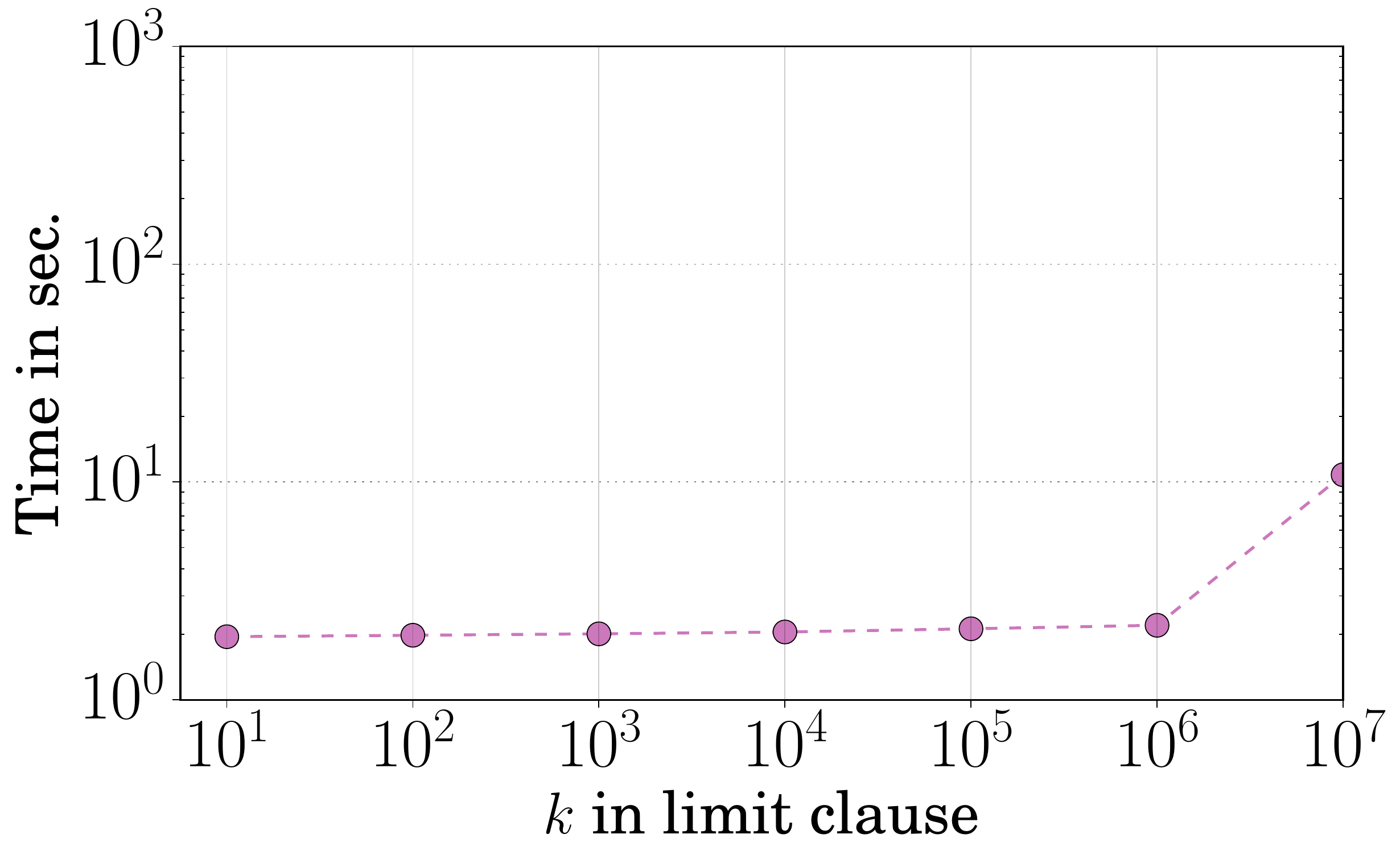}
		\caption{$\mathsf{IMDB}_{4\mathsf{hop}}$}  \label{fig:imdb:4path:lex}
	\end{subfigure}
	\begin{subfigure}{0.24\linewidth}
		\includegraphics[scale=0.17]{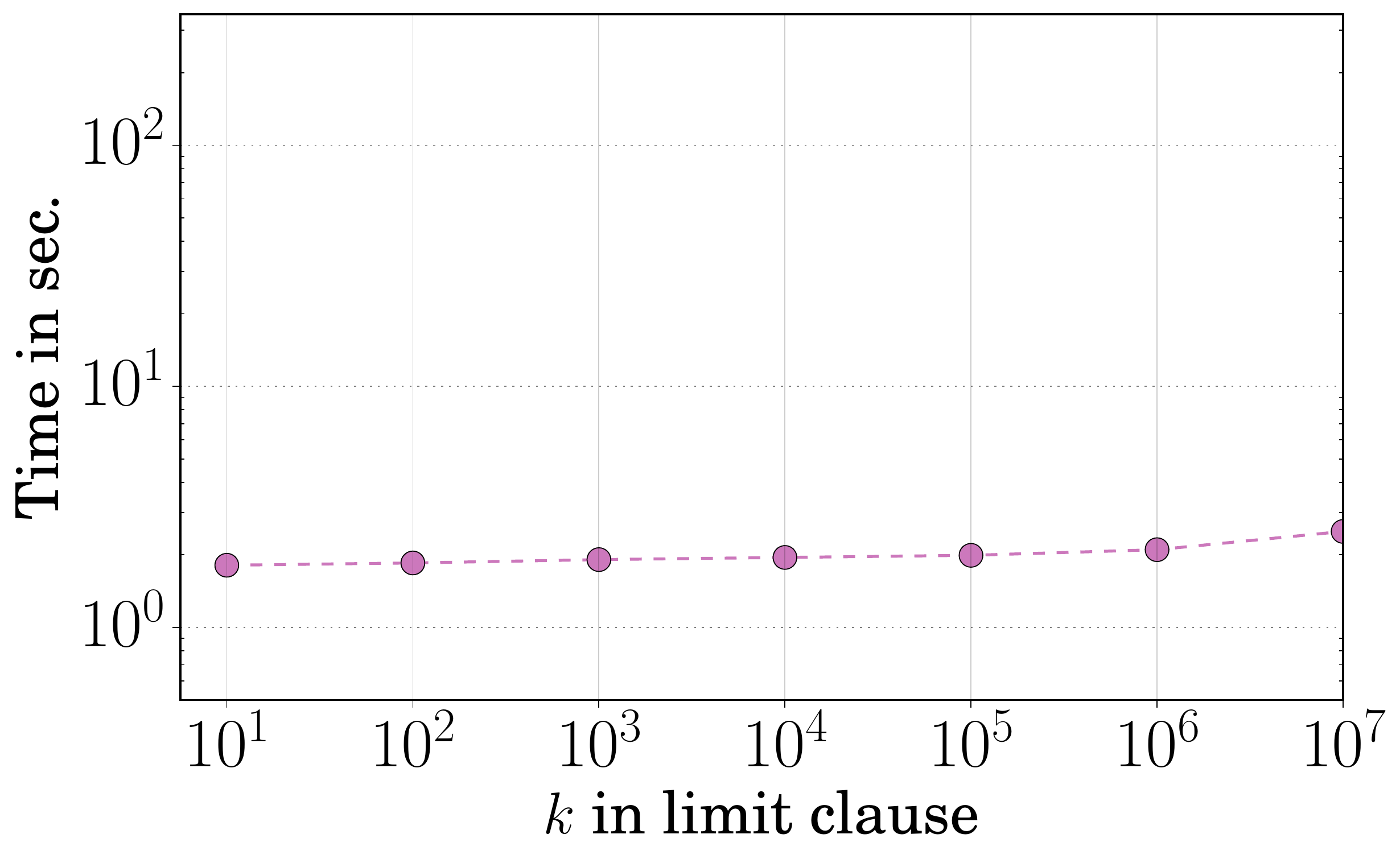}
		\caption{$\mathsf{IMDB}_{3\mathsf{star}}$}  \label{fig:imdb:3star:lex}
	\end{subfigure}

\begin{subfigure}{0.24\linewidth}
	\includegraphics[scale=0.17]{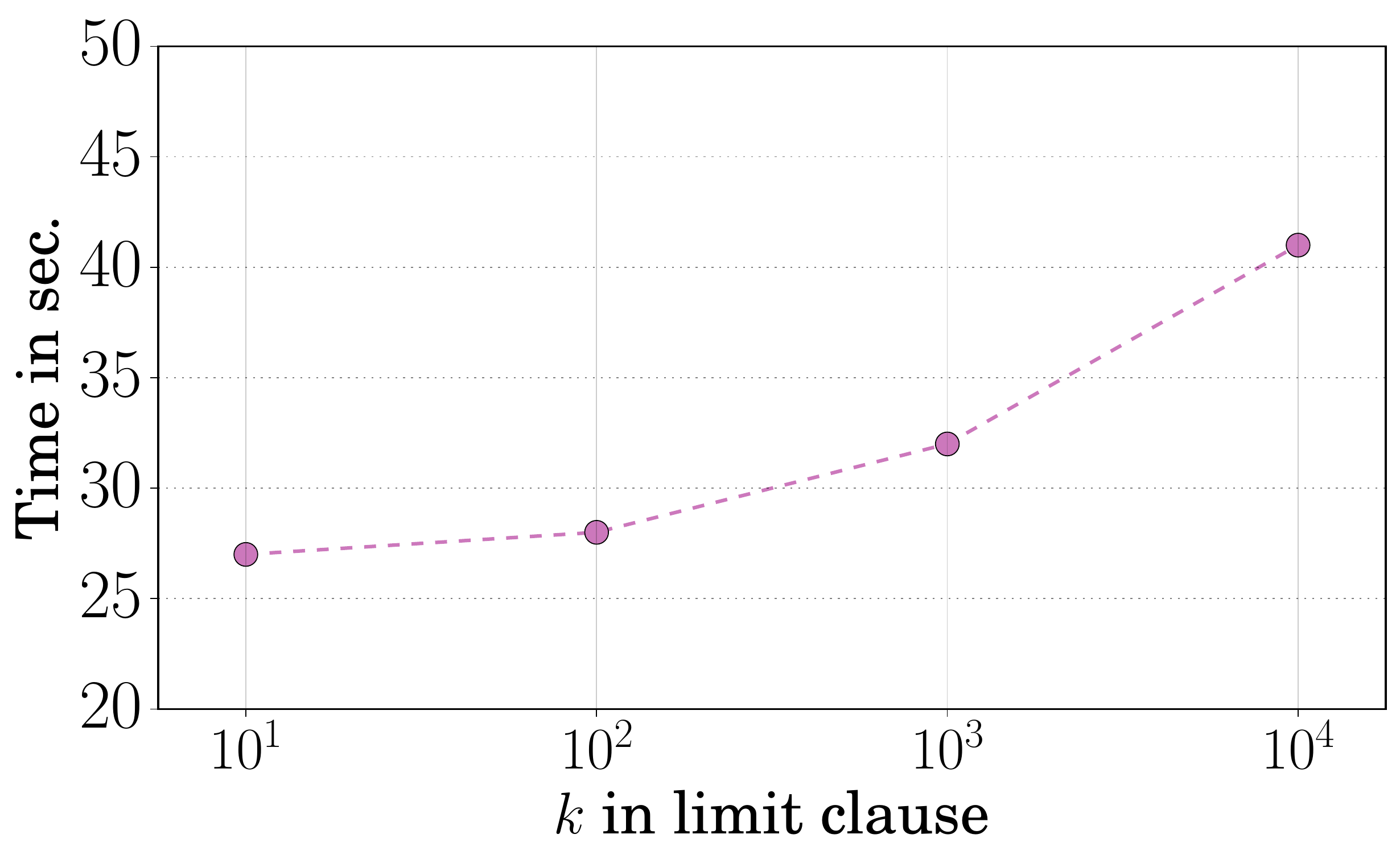}
	\caption{$\mathsf{Friendster}_{2\mathsf{hop}}$}  \label{fig:friendster:2path:lex}
\end{subfigure}
\begin{subfigure}{0.24\linewidth}
	\includegraphics[scale=0.17]{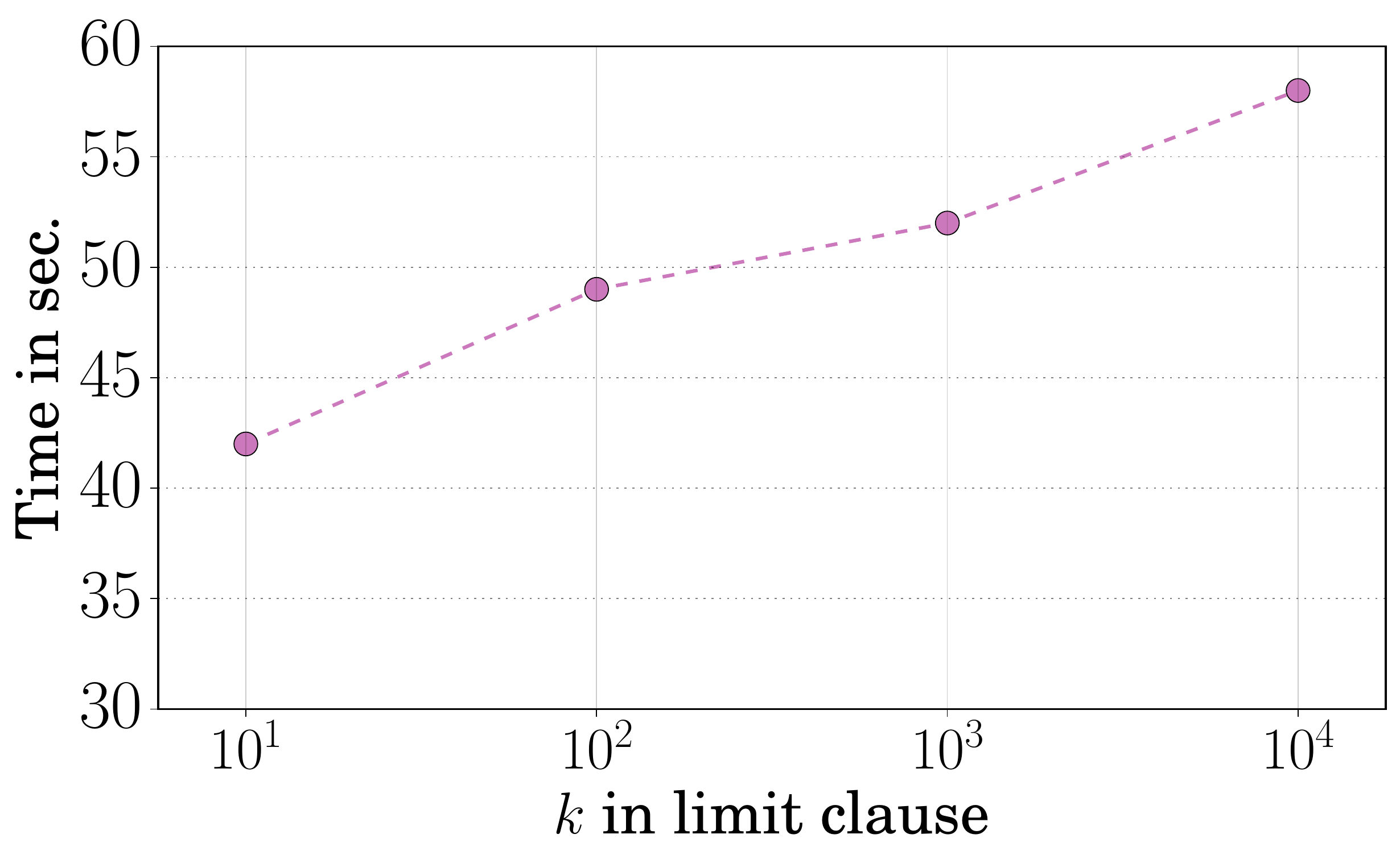}
	\caption{$\mathsf{Friendster}_{3\mathsf{hop}}$}  \label{fig:friendster:3path:lex}
\end{subfigure}
\begin{subfigure}{0.24\linewidth}
	\includegraphics[scale=0.17]{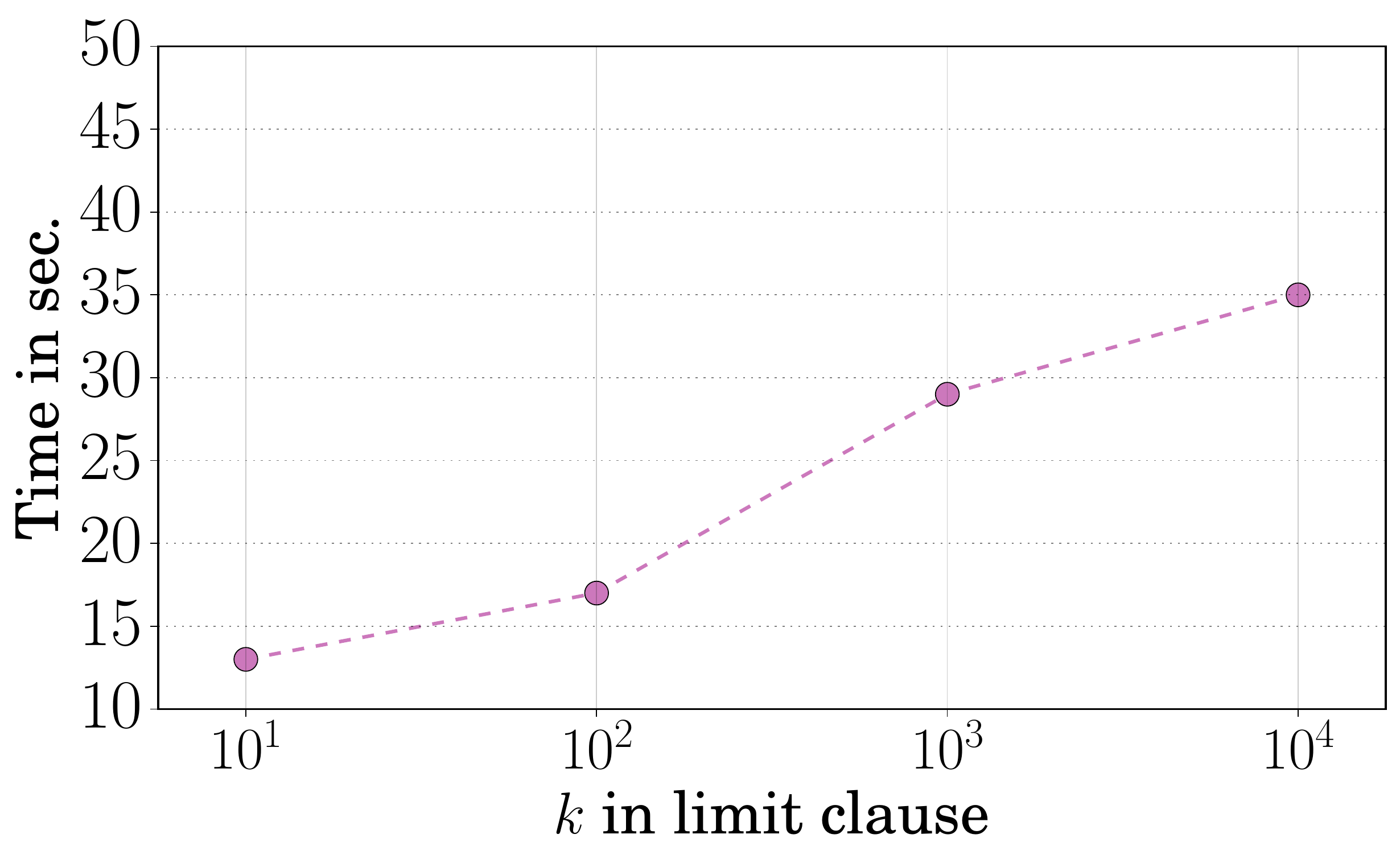}
	\caption{$\mathsf{Memetracker}_{2\mathsf{hop}}$}  \label{fig:memetracker:2path:lex}
\end{subfigure}
\begin{subfigure}{0.24\linewidth}
	\includegraphics[scale=0.17]{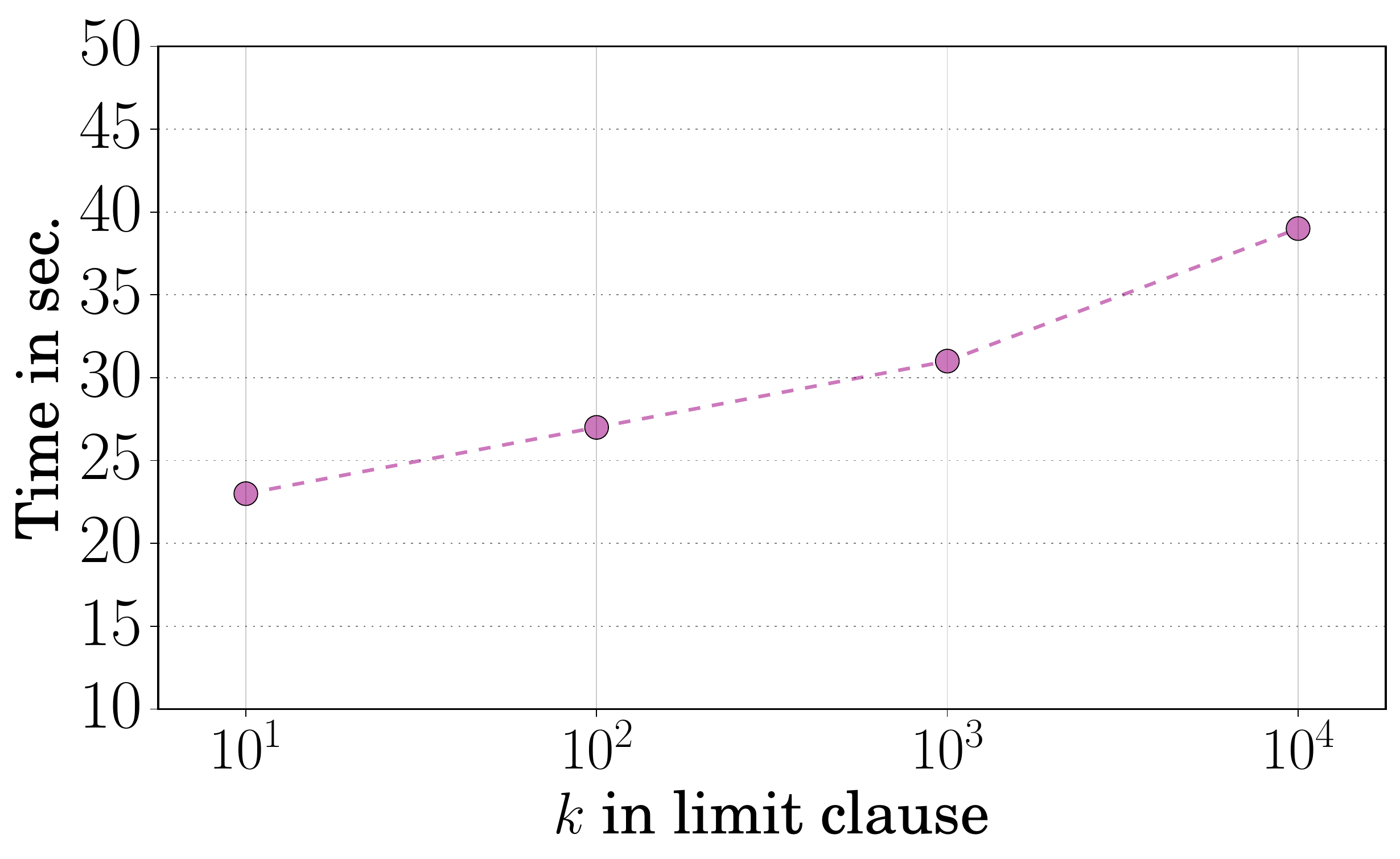}
	\caption{$\mathsf{Memetracker}_{3\mathsf{hop}}$}  \label{fig:memetracker:3path:lex}
\end{subfigure}	
	\caption{Comparing our linear delay algorithm with state-of-the-art engines for lexicographic function.} \label{fig:imdb:notradeoff:lex}
\end{figure*}

\begin{figure*}
	\begin{subfigure}{0.6\linewidth}
		\includegraphics[scale=0.5]{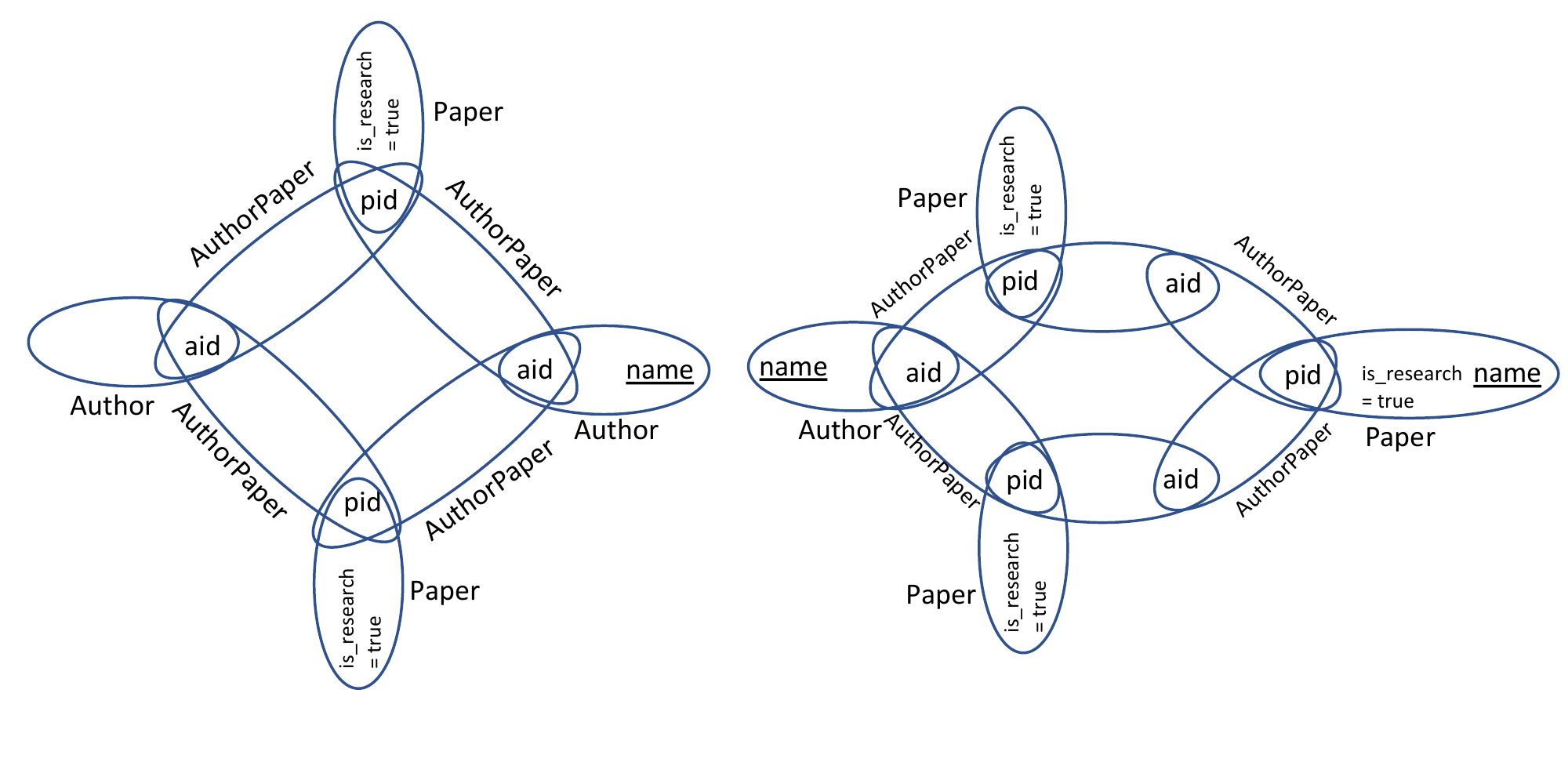}
		\caption{Four cycle query for \textsf{DBLP} (left) and Six cycle query for \textsf{DBLP} (right)}
	\end{subfigure}
	\begin{subfigure}{0.3\linewidth}
		\includegraphics[scale=0.5]{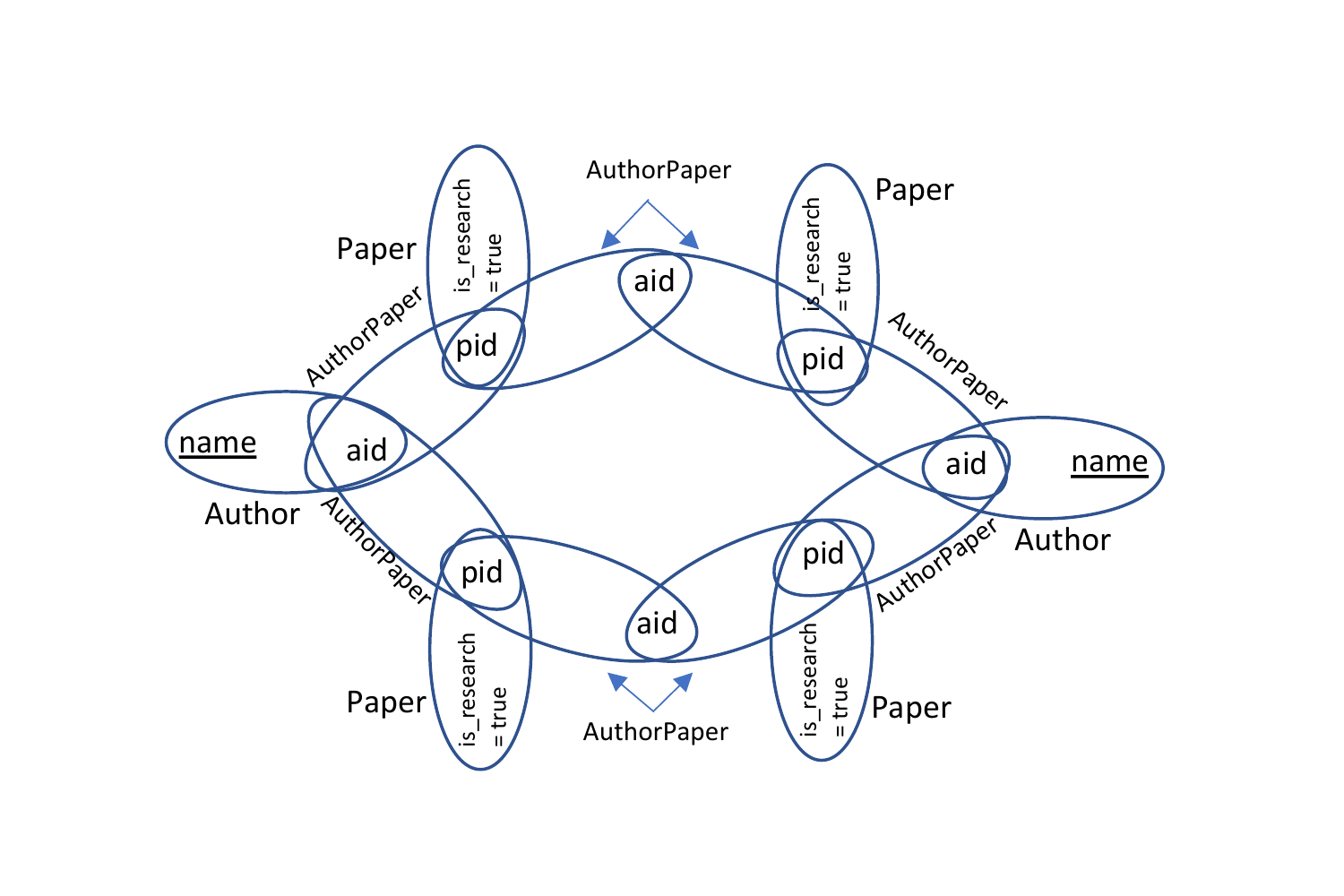}
		\caption{Eight cycle query for \textsf{DBLP}}
	\end{subfigure}	
\caption{Hypergraphs for cyclic queries} \label{fig:cyclic}
\end{figure*}

\introparagraph{Measuring the delay} Our final experiment in this subsection evaluates how many priority queue operations are required for each answer tuple. Recall that for each answer, our algorithm performs $O(|D|)$ priority queue operations in the worst-case (\autoref{thm:general}). Thus, the number of priority queue operations is a proxy for the delay between any two answers. \autoref{delay} plots the fraction of answers that require a certain number of PQ operations for $\mathsf{DBLP}_{2\mathsf{hop}}$. Nearly $70\%$ of answers require only one PQ push and pop operation and $99\%$ of the output requires at most $22$ push and pop operations. However, there are some outliers as well. The largest number of push and pops over all answers is $306$. This is not particularly surprising given the fact that DBLP is a sparse dataset but dense matrix datasets can easily have a large duplication in practice. For the IMDB dataset, we found that $95\%$ of answers require one priority queue operation and $99\%$ of answers needed at most three operations.

\begin{figure}
	\begin{subfigure}{0.4\linewidth}
		\includegraphics[scale=0.2]{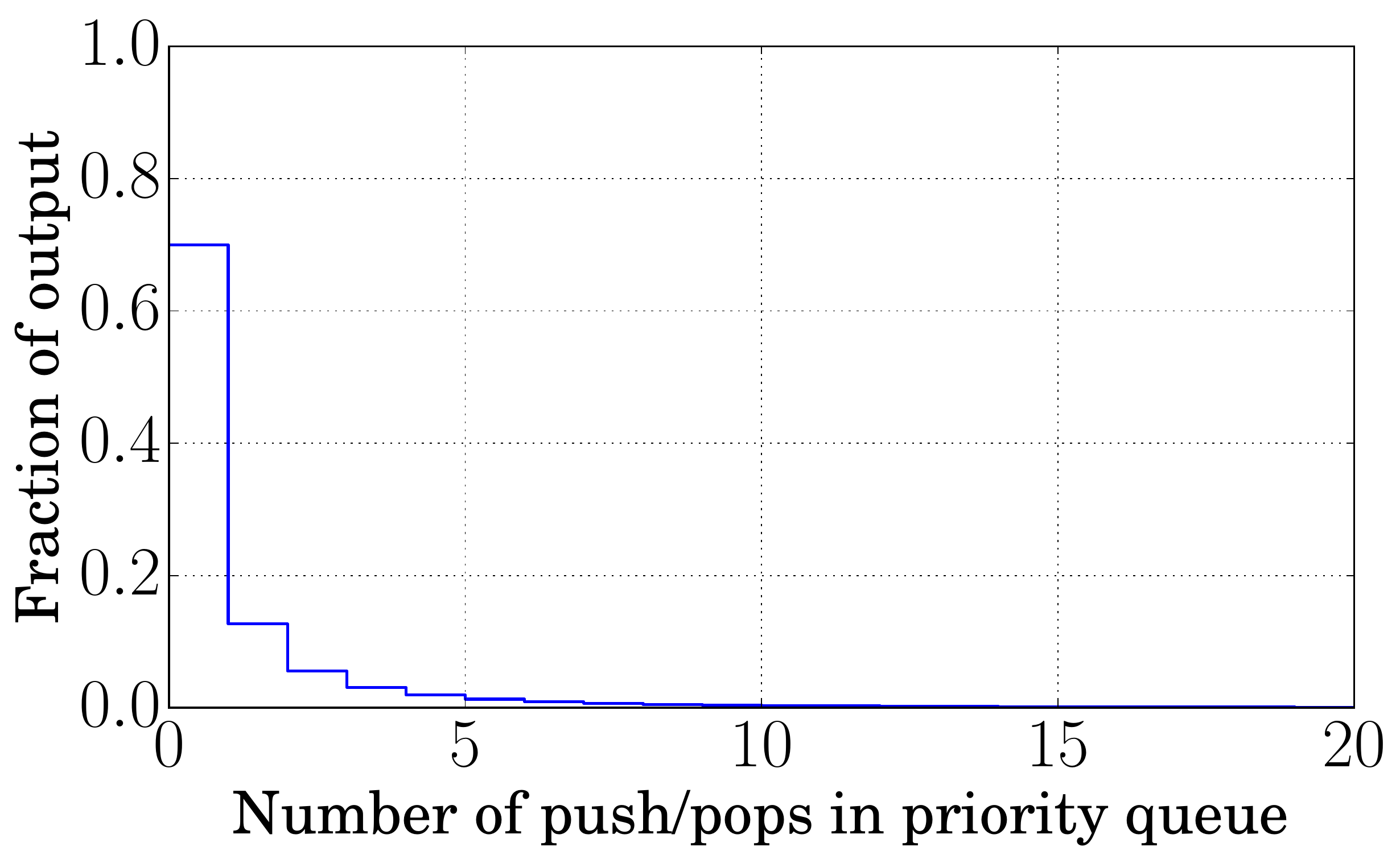}
		\caption{$\mathsf{DBLP}_{2\mathsf{hop}}$\pq\ operations per answer}  \label{delay}
	\end{subfigure}
	\begin{subfigure}{0.5\linewidth}
			\scalebox{1.0}{
			\begin{tabular}{lcccc}\toprule
				& $\mathsf{k}=10$  & $\mathsf{k}=10^2$ & $\mathsf{k}=10^3$ & $\mathsf{k}=10^4$ \\\midrule
				\textbf{four cycle}    & 25.2s & 26.8s & 29.6s & 37.4s   \\
				\textbf{six cycle}    & 323.8s & 325.7s & 339.2s & 360.4s   \\
				\textbf{eight cycle}    & 595.6s & 605.3s & 609.4s & 620.8s   \\		
				\textbf{bowtie}    & 598.2s & 604.7s & 610.89s & 625.8   \\\bottomrule
		\end{tabular}}
		\caption{Cyclic query performance on \textsf{IMDB} dataset for different values of $k$ in \sqlhighlight{LIMIT} clause. Ranking function is sum.} \label{table:cyclic:imdb}
	\end{subfigure}
\caption{(left) \textsf{DBLP} empirical delay; (right) cyclic query performance on \textsf{IMDB} }
\end{figure}

\subsection{Large Scale Experiments}

\introparagraph{Lexicographic ordering} \autoref{fig:friendster:2path:lex},~\ref{fig:friendster:3path:lex},~\ref{fig:memetracker:2path:lex},~\ref{fig:memetracker:3path:lex} show the execution time of the ranked two and three hop queries on Friendster and Memetracker datasets. Similar to the small-scale datasets, lexicographic function is faster to evaluate than sum function but all baselines engines were unable to complete execution in the cutoff time.

\subsection{Cyclic Query Description}
In this subsection, we describe the cyclic queries used in the experimental evaluation.~\autoref{fig:cyclic} shows the hypergraphs for the cyclic queries on \textsf{DBLP} dataset. The queries are defined in a very similar way for all other real-world datasets as well.

\textbf{Four cycle.} Find all co-author pairs who have been on at least two research papers together ranked by sum of weights/lexicographic ordering.  

\textbf{Six cycle.} For every author, find all research paper such that the there exist at least two co-authors and the research paper contains both the co-authors. Rank all the output author, research paper pairs  by sum of weights/lexicographic ordering.  

\textbf{Eight cycle.} Find all co-author of co-author pairs who have been on at least two different research papers together ranked by sum of weights/lexicographic ordering.  

\textbf{Bowtie.} Consider the materialized view $V(a_1, a_2)$ for the eight cycle query. The bowtie query enumerates $Q = \pi_{a_1,a_3} (V(a_1, a_2) \Join V(a_2, a_3))$ ranked by sum of weights/lexicographic ordering.

\end{document}